%% file: main.tex
\documentclass[acmsmall,screen]{acmart}%

\setcopyright{none}

\usepackage{amsfonts}
\usepackage{mathrsfs}
\usepackage{amsmath}
\usepackage{mathtools}
\usepackage{mathpartir}
\usepackage{bbding}
\usepackage{color}
\usepackage[capitalize,nameinlink]{cleveref}
\usepackage{xparse}
\usepackage{ifthen}
\usepackage{thmtools}
\usepackage{thm-restate}
\usepackage{subcaption}
\usepackage{tikz}
\usepackage{newfile}
\usepackage{totcount}

\regtotcounter{@todonotes@numberoftodonotes}

\usepackage{complexity-mod}

\usepackage{listings}

  \lstdefinestyle{tinyc}{
    basicstyle=\scriptsize\ttfamily,
    keywordstyle=\color{blue}
  }

  \lstdefinestyle{normalc}{
    basicstyle=\ttfamily,
    numbers=none,
    keywordstyle=\color{blue}
  }

  \lstdefinestyle{inlinec}{
    basicstyle=\ttfamily
  }

\usetikzlibrary{backgrounds}
\tikzstyle{every picture}+=[remember picture]

\usepackage{todonotes}

\newcommand{\N}{\mathbb{N}}

\DeclareDocumentCommand{\DCPS}{O{}}{\mathsf{DCPS}\ifthenelse{\equal{#1}{}}{}{[#1]}}
\DeclareDocumentCommand{\SRP}{O{}}{\mathsf{SRP}\ifthenelse{\equal{#1}{}}{}{[#1]}  }
\DeclareDocumentCommand{\HP}{O{}}{\mathsf{HP}\ifthenelse{\equal{#1}{}}{}{[#1]}} %

\newclass{\TWOEXPSPACE}{2EXPSPACE}
\newclass{\THREEEXPSPACE}{3EXPSPACE}
\newclass{\FOUREXPSPACE}{4EXPSPACE}

\usepackage{tikz}
\usetikzlibrary{petri,automata}
\tikzstyle{every place}=[minimum size=5mm]
\tikzstyle{every transition}=[minimum size=5mm]
\tikzstyle{every state}=[minimum size=5mm]
\usetikzlibrary{backgrounds,calc}

 \usetikzlibrary{ shapes, fit, arrows}
 \usetikzlibrary{decorations} %
\DeclareMathSymbol{\mdot}{\mathord}{symbols}{"01}
\usepackage{enumerate}
\usepackage{lscape}
\usepackage[utf8]{inputenc}
\usepackage[T1]{fontenc}
\usepackage{microtype}
\input{macros}

\usepackage[linesnumbered,boxed]{algorithm2e}

\title{
Context-Bounded Verification of Thread Pools}  

\newcommand{\OurInstitution}{Max Planck Institute for Software Systems (MPI-SWS)}
\newcommand{\OurStreet}{Paul-Ehrlich-Stra{\ss}e, Building G26}
\newcommand{\OurCity}{Kaiserslautern}
\newcommand{\OurPostcode}{67663}
\newcommand{\OurCountry}{Germany}

\author{Pascal Baumann}
\orcid{0000-0002-9371-0807}             %
\affiliation{
  \institution{\OurInstitution}            %
  \streetaddress{\OurStreet}
  \city{\OurCity}
  \postcode{\OurPostcode}
  \country{\OurCountry}                    %
}
\email{pbaumann@mpi-sws.org}          %

\author{Rupak Majumdar}
\orcid{0000-0003-2136-0542}             %
\affiliation{
  \institution{\OurInstitution}            %
  \streetaddress{\OurStreet}
  \city{\OurCity}
  \postcode{\OurPostcode}
  \country{\OurCountry}                    %
}

\email{rupak@mpi-sws.org}          %

\author{Ramanathan S. Thinniyam}
\orcid{0000-0002-9926-0931}             %
\affiliation{
  \institution{\OurInstitution}            %
  \streetaddress{\OurStreet}
  \city{\OurCity}
  \postcode{\OurPostcode}
  \country{\OurCountry}                    %
}
\email{thinniyam@mpi-sws.org}          %

\author{Georg Zetzsche}
\orcid{0000-0002-6421-4388}             %
\affiliation{
  \institution{\OurInstitution}            %
  \streetaddress{\OurStreet}
  \city{\OurCity}
  \postcode{\OurPostcode}
  \country{\OurCountry}                    %
}
\email{georg@mpi-sws.org}          %

\input{abstract}



\begin{document}

\maketitle

\sloppy %

\input{intro}
\input{prelims}
\input{bddThreads}
\input{tcvass}
\input{binarybound}
\input{discussion}

\begin{acks}                            %
This research was sponsored in part by
the Deutsche Forschungsgemeinschaft project 389792660 TRR 248--CPEC
and by the European Research Council under the
Grant Agreement 610150 (http://www.impact-erc.eu/) (ERC Synergy Grant ImPACT).
\end{acks}

\newpage

\label{beforebibliography}
\newoutputstream{pages}
\openoutputfile{main.pages.ctr}{pages}
\addtostream{pages}{\getpagerefnumber{beforebibliography}}
\closeoutputstream{pages}
\bibliography{bibliography}

\appendix

\input{appendix-succinctPDA}

\input{appendix-succinctTasks}
\input{appendix-tcvass}
\input{appendix-binarybound}

\end{document}

%% file: macros.tex
\def\qed{\rule{0.4em}{1.4ex}}

\def\@envspa{\hspace{0.3em}}
\def\@sa{\hspace{-0.2em}}
\def\@sb{\hspace{0.5em}}
\def\@sc{\hspace{-0.1em}}
\mathchardef\mhyphen="2D

\newtheorem{remark}{Remark}{\bfseries\upshape}{\rm}
{\bfseries\upshape}{\itshape}
{\bfseries\upshape}{\itshape}

\def\set#1{{\{ #1 \}}}
\def\tuple#1{{\langle #1 \rangle }}

\def\multi#1{{[\![ #1 ]\!]}}

\def\nats{{\mathbb{N}}}
\def\integs{{\mathbb{Z}}}

\def\mmap{\mathbf{m}}

\def\tbuf{\mathbf{t}}

\def\card#1{\lvert {#1} \rvert}

\newcommand{\multiset}[1]{{\mathbb{M}[ #1 ]}}
\def\cN{\mathcal{N}}

\def\prod{\mathcal{P}}

\newcommand{\cC}{\mathcal{C}}
\newcommand{\cA}{\mathcal{A}}
\newcommand{\cB}{\mathcal{B}}
\newcommand{\cT}{\mathcal{T}}
\newcommand{\cV}{\mathcal{V}}
\newcommand{\cP}{\mathcal{P}}

\newcommand{\cD}{\mathcal{D}}%
\newcommand{\cE}{\mathcal{E}}

\newcommand{\rnp}{\mathsf{RNP}} %
\newcommand{\srnp}{\mathsf{sRNP}} %
\newcommand{\TDPN}{\mathsf{TDPN}} %
\newcommand{\sTDPN}{\mathsf{sTDPN}} %
\newcommand{\TCVASS}{\mathsf{TCVASS}}
\newcommand{\PDA}{\mathsf{PDA}}

\newcommand{\NFA}{\mathsf{NFA}}
\newcommand{\dsNFA}{\mathsf{dsNFA}}
\newcommand{\DCBPR}{\mathsf{DCBP}}

\newcommand{\scrC}{\mathscr{C}}

\newcommand{\rev}[1]{#1^{\mathsf{rev}}}
\DeclareDocumentCommand{\langof}{O{} m}{%
  \mathsf{L}_{#1}(#2)%
  }
\DeclareDocumentCommand{\autstep}{O{}}{%
        \xrightarrow{#1}%
        }

\DeclareDocumentCommand{\autsteps}{O{}}{%
        \mathrel{{\xrightarrow{#1}}{}^*}%
      }
\DeclareDocumentCommand{\pdastep}{O{} m m}{%
  \ifthenelse{\equal{#3}{}}{%
    \xRightarrow[#1]{#2}%
  }{%
    \mathrel{{\xRightarrow[#1]{#2}}{}_{#3}}%
  }%
}

\DeclareDocumentCommand{\powerset}{O{} m}{
  \mathbb{P}_{#1}(#2)
}

\mathchardef\mhyphen="2D
\DeclareDocumentCommand{\EXPTIME}{O{}}{%
  \ifthenelse{\equal{#1}{}}{%
    \mathsf{EXPTIME}%
  }{%
    {#1}\mhyphen\mathsf{EXPTIME}%
  }%
}

\DeclareDocumentCommand{\NTERM}{O{}}{\mathsf{NTERM}\ifthenelse{
\equal{#1}{}}{}{[#1]}  }
\DeclareDocumentCommand{\UBOUND}{O{}}{\mathsf{UBOUND}\ifthenelse{
\equal{#1}{}}{}{[#1]}  }
\DeclareDocumentCommand{\REACH}{O{}}{\mathsf{REACH}\ifthenelse{
\equal{#1}{}}{}{[#1]}  }
\DeclareDocumentCommand{\FAIRNTERM}{O{}}{\mathsf{FNTERM}
\ifthenelse{
\equal{#1}{}}{}{[#1]}  }
\DeclareDocumentCommand{\STARV}{O{}}{\mathsf{STARV}\ifthenelse{
\equal{#1}{}}{}{[#1]}  }
\DeclareDocumentCommand{\SREACH}{O{}}{\mathsf{SREACH}\ifthenelse{
\equal{#1}{}}{}{[#1]}  }

\def \cP{\mathcal{P}}
\def \cD{\mathcal{D}}

\newcommand{\cM}{\mathcal{M}}

\def \max{\mathsf{max}}

\def \TYPES{\mathsf{TYPES}}

\newcommand{\ltr}[1]{\mathsf{#1}}

\DeclareDocumentCommand\PDA{}{\mathsf{PDA}}
\DeclareDocumentCommand\VASS{}{\mathsf{VASS}}

\DeclareDocumentCommand{\dcl}{o m}{\IfNoValueTF{#1}{#2\mathop{\downarrow}}{{#2\mathop{\downarrow}}_{#1}}}
\newcommand{\subw}{\preccurlyeq}

\newcommand{\init}{\mathsf{init}}
\newcommand{\final}{\mathsf{final}}
\newcommand{\move}{\mathsf{move}}
\newcommand{\ctr}{\mathsf{ctr}}

\newcommand{\glob}{\mathsf{gl}}
\newcommand{\loca}{\mathsf{loc}}

\newcommand{\mdotcup}{\mathbin{\dot\cup}}
\newcommand{\mdotbigcup}{\mathbin{\dot\bigcup}}
\newcommand{\SigmaE}{\Sigma_{\varepsilon}}
\newcommand{\hGamma}{\hat{\Gamma}}

\newcommand{\hati}{\hat{I}}
\newcommand{\hatie}{\hat{I}_{\varepsilon}}
\newcommand{\barI}{\overline{I}}
\newcommand{\bari}{\overline{i}}

\newcommand{\clean}{\mathsf{clean}}
\newcommand{\term}{\mathsf{term}}
\newcommand{\inter}{\mathsf{inter}}
\newcommand{\resump}{\mathsf{resump}}
\newcommand{\emp}{\mathsf{empty}}
\newcommand{\start}{\mathsf{start}}

\newcommand{\vecz}{\vec{\mathbf{0}}}

%% file: abstract.tex
\begin{abstract}
\emph{Thread pooling} is a common programming idiom in which a fixed set of worker threads
are maintained to execute tasks concurrently.
The workers repeatedly pick tasks and execute them to completion.
Each task is sequential, with possibly recursive code, and tasks
communicate over shared memory.
Executing a task can lead to more new tasks being spawned.
We consider the safety verification problem for thread-pooled programs.
We parameterize the problem with two parameters: 
the size of the thread pool as well as the number of context switches for each task.
The size of the thread pool determines the number of workers running concurrently.
The number of context switches determines how many times a worker can be swapped out while executing
a single task---like many verification problems for multithreaded recursive programs, the context bounding is important
for decidability.

We show that the safety verification problem for thread-pooled, context-bounded, Boolean programs
is $\EXPSPACE$-complete, even if the size of the thread pool and the context bound are given in binary.
Our main result, the $\EXPSPACE$ upper bound, is derived using a sequence of new 
succinct encoding techniques of independent language-theoretic interest.
In particular, we show a polynomial-time construction of downward closures
of languages accepted by succinct pushdown automata as doubly succinct nondeterministic finite automata. 
While there are explicit doubly exponential lower bounds on the size of nondeterministic
finite automata accepting the downward closure, our result shows these automata can be compressed.
We show that thread pooling significantly reduces computational power:
in contrast, if only the context bound is provided in binary, but there
is no thread pooling, the safety verification problem becomes $\THREEEXPSPACE$-complete.
Given the high complexity lower bounds of related problems involving binary parameters, 
the relatively low complexity of safety verification with thread-pooling comes as a surprise.
\end{abstract}

%% file: intro.tex
\section{Introduction}

\emph{Thread pooling} is a common programming idiom in which a fixed set of worker threads
are maintained to execute tasks concurrently.
The worker threads repeatedly pick tasks from a task buffer and execute them to completion.
Each task is sequential code, possibly recursive, and tasks communicate over shared memory.
Executing a task can lead to more new tasks being spawned; these tasks get added to the task buffer for execution.
When a worker thread finishes executing a task, it goes to the task buffer and picks another task non-deterministically.
Thread pools reduce latency by executing tasks concurrently without paying the cost of creating and destroying new threads
for each, possibly short-lived, task.
Additionally, they improve system stability because resource requirements are more stable compared to unbounded creation
and destruction of threads.
Thus, most languages and runtimes for concurrency explicitly support thread pooling.

We consider the safety verification problem for thread-pooled programs modeled as reachability of some global state.
Even if the global state is finite, the system is unbounded in multiple dimensions: each executing task can have an unbounded stack, and the task
buffer of pending tasks can be unbounded.
The safety verification problem is undecidable as soon as there are two workers in the thread pool and threads are recursive.
Thus, we take the usual approach of bounding the number of context switches for each task \cite{QR05}: 
the context bound is an upper bound on the number of times a worker thread can be interrupted by other threads while executing a specific task.
The context-bounded safety verification problem for thread pools takes as input a program description and two parameters written in binary:
the size of the thread pool $N$ (i.e., the number of worker threads executing concurrently) and the context switch bound $K$.
Our main result shows that the safety verification problem for this model can be decided in $\EXPSPACE$,
even when we assume the parameters are given in binary. 
In fact, our $\EXPSPACE$ upper bound holds for the model of \emph{Boolean programs}
\cite{BR00,BR01,GodefroidY13},
where the global state of the program is specified succinctly in the form of 
Boolean variables and each thread is allowed its own set of local Boolean variables. 

When there are no local variables, the global state is specified explicitly rather than succinctly, and the size of the thread pool is one, the model degenerates to the well-studied model of asynchronous programs \cite{SenV06,JhalaM07,GantyM12}.
Safety verification for asynchronous programs is already $\EXPSPACE$-complete \cite{GantyM12}.
Thus, we get $\EXPSPACE$-completeness for our problem---hardness holding already for a single executing thread in the thread pool.
(We note that we consider a model in which Boolean parameters are only allowed to be passed to
recursive calls within a thread, but not to newly spawned tasks. 
Allowing formal parameters to spawns as well would result in a 
$\TWOEXPSPACE$-complete problem: membership follows from our methods, 
whereas hardness can be shown by applying
the techniques of %
\citet{BaumannMajumdarThinniyamZetzsche2020a}
to our slightly different problem and setting.)

The $\EXPSPACE$ upper bound requires us to develop a number of new techniques, of independent interest, to deal with \emph{succinct}
representations of computations.
Our goal is to reduce the decision problem to the coverability problem of a polynomial-sized vector addition system with states ($\VASS$),
which can be solved in $\EXPSPACE$ \cite{Rackoff78}.
Our starting point is the $\TWOEXPSPACE$ algorithm for context-bounded reachability of \citet{AtigBQ2009}.
Their proof has the following steps.
First, they observe that safety verification is preserved under ``downward closures,'' where some spawned tasks are lost.
Second, they show that the language of spawns and context switches of a single task execution can be represented as a pushdown automaton ($\PDA$),
and from this $\PDA$, they can construct a non-deterministic finite automaton ($\NFA$) for the downward closure of the language of the $\PDA$ that preserves the context switches
and only loses spawns.
Using the $\NFA$ for each thread, they construct a $\VASS$ of polynomial size that counts the number of threads in each location
of the $\NFA$; safety verification in the original program reduces to coverability in this $\VASS$.

A careful accounting of this algorithm gives a $\THREEEXPSPACE$ algorithm for our problem.
This is because the $\NFA$ representation for the downclosure of a $\PDA$ can be exponentially large \cite{Courcelle1991}; in fact, there are lower bounds showing that an
exponential blow-up is necessary \cite{BachmeierLS2015}.
Second, since the number of context switches is given in binary, the context switch preserving downclosure is in fact \emph{doubly exponential}
in the size of the $\PDA$ and $K$; again, there are examples showing that the blow-up is necessary.
While the size $N$ of the thread pool is given in binary, the $\VASS$ representation does not incur a further
blow-up since we can precisely track the states of $N$ threads with a mutiplicative cost.
Additionally, specifying the global state succinctly via global Boolean variables increases the size of the downclosure in a way that is multiplicative to the size increase caused by a binary encoded $K$, therefore incurring no further blow-up either.
Thus, a naive reduction to a $\VASS$ is doubly exponential, giving an overall triply exponential algorithm.

We overcome the complexity obstacles by developing algorithms for succinct machines.
We introduce \emph{succinct} representations of $\PDA$s and $\NFA$s.
The states of a succinct $\PDA$ are encoded as strings over an alphabet, and the transitions are encoded as transducers.
(Equivalently, one could represent them using Boolean circuits; our choice of transducers simplifies some language theoretic constructions.)
A succinct $\PDA$ encodes an exponentially larger $\PDA$.
We also define \emph{doubly succinct} $\NFA$s, whose states are encoded as exponentially long strings,
which represent $\NFA$s that are doubly exponentially larger.

Our first technical result is that downclosures of succinctly represented $\PDA$s are ``well-compressible'': 
they can be represented by \emph{doubly succinct} $\NFA$ of polynomial size.
The result strengthens a result of \citet{MajumdarTZ21} that shows a single exponential compression of the downclosure of normal $\PDA$s.
Our result is of independent interest: it implies, for example, that downclosures of Boolean programs are doubly succinctly representable.
It is surprising, because the emptiness problem for succinct $\PDA$s is $\EXPTIME$-complete and one would expect an exponential blowup.
Indeed, our construction carries out computations of alternating polynomial space Turing machines represented succinctly within the automaton.
We further strengthen the result to show that even when we preserve the $K$ context switches, the downclosure is still representable by
a doubly succinct $\NFA$.

Our next obstacle is to obtain the language of a thread pool containing $N$ threads, each of which is given by a doubly succinct $\NFA$.
Here, an explicit representation of $N$ separate states of the doubly succinct $\NFA$s will lead to an exponential blowup.
We introduce a succinct representation for $\VASS$, where the control states are represented doubly succinctly---such a succinct $\VASS$ represents a $\VASS$ with
doubly exponentially many control states. 
We show that the language of a thread pool can be represented as a $\VASS$ with a doubly succinct control. 

Finally, we show that the language of a doubly succinct $\NFA$ can be seen as the coverability language of a (normal) $\VASS$. 
The idea of the proof goes back to Lipton's encoding of doubly exponential counter machines succinctly using a $\VASS$ \cite{lipton1976reachability,Esp98a}; 
however, instead of proving a lower bound, as Lipton did, we use the succinct encoding to show a better upper bound.
The overall construction implies that the coverability language of the doubly succinct $\VASS$ is the same as the  coverability language of
a normal $\VASS$ that can be constructed from the succinct representation.
Thus, coverability of doubly succinct $\VASS$ can be decided in $\EXPSPACE$.
 
Overall, a composition of these steps gives us the desired $\EXPSPACE$ upper bound for the thread-pooled reachability problem. 

Our singly exponential space upper bounds are unexpected, since moving from unary to binary parameters usually involves an exponential jump
in the complexity of verification.
For example, in the special case where threads do not spawn new tasks and states are specified explicitly without Boolean variables, \citet{QR05} showed that context-bounded
reachability is $\NP$-hard when the number of context switches is fixed or given in unary.
We show that, even when the thread pool has just two threads ($N=2$), the problem is $\NEXP$-complete if the number of context switches is given in binary.
(The upper bound follows from \citet{QR05}, even when Boolean variables are present, and the lower bound can be shown by reducing from
the following $\NEXP$-complete problem: given context free grammars $G_1$ and $G_2$ and a number $k$ in binary,
are there words $w_1\in\langof{G_1}$ and $w_2 \in \langof{G_2}$ that agree on the first $k$ letters?)

Similarly, \citet{AtigBQ2009} study the context-bounded safety verification problem for multi-threaded shared memory programs with neither thread pooling nor Boolean variables,
that is, in a model where the global state is specified explicitly and each task is executed on a separately created thread.
They show that the problem is in $\TWOEXPSPACE$ for a fixed context bound, or if the context bound is given in {unary}.
\citet{BaumannMajumdarThinniyamZetzsche2020a} show a matching lower bound for this case.
We improve upon these results to show that safety verification (without thread pooling) is $\THREEEXPSPACE$-complete when the context bound is given in binary.
The hard part is the lower bound.
We extend the encoding of \citet{BaumannMajumdarThinniyamZetzsche2020a} to perform an ``exponentially larger'' computation when $K$ is in binary.
Interestingly, succinct computations are now used in this proof to show \emph{lower bounds} and hardness!

Given the above lower bounds, it is indeed surprising that fixing the thread pool parameter in binary 
leads to a doubly exponential reduction in complexity of safety verification (or, seen through a lens of 
complexity theory, in the reduction in expressiveness of languages computable by these machines).
Indeed, the complexity is no higher than the special case of one thread that executes tasks to completion!

\subsubsection*{Related Work}
Programming with thread pools is a ubiquitous concurrency pattern, and almost every language or library supporting concurrent programming supports thread pooling.
Usually, the thread pool needs to be configured by programmers; there are real-world instances that show that such configurations may be tricky to get right.\footnote{
	See, e.g., discussions at 
	\url{https://engineering.zalando.com/posts/2019/04/how-to-set-an-ideal-thread-pool-size.html},
	\url{https://developer.android.com/guide/background/threading}, 
	\url{https://developer.android.com/topic/performance/threads}, and especially
	\url{https://www.techyourchance.com/threadposter-explicit-unit-testable-multi-threading-library-for-android/} for discussions of the complexity of thread pool
	configuration and resulting crashes in the Android settings application. 
}
Since thread pools appear in many large-scale systems, dynamic analysis tools often provide support for thread pools \cite{LiCLLZGGLL18}.
However, the complexity implications of thread pools on static verification had not been considered.

There are, by now, many decidability results for context bounded verification in multi-threaded settings 
\cite{QR05,LalReps,MusuvathiQadeer,TorreMP09,LaTorre,AtigBQ2009,BaumannMajumdarThinniyamZetzsche2021a,MeyerMuskallaZetzsche2018}.
The work of \citet{AtigBQ2009} is closest to ours in the programming model: they consider safety verification for a multithreaded shared memory model with \emph{dynamic thread spawns},
the same as us.
Our model additionally has global and local variables and we consider the effect of thread pools. 
Prior results on context bounded reachability focused on a fixed number of threads \cite{QR05,LalReps,ChiniKKMS17}; interestingly, in most of these papers,
the context bound parameter was assumed to be given in unary, and as we stated before, the complexity results are ``one exponential higher'' when the parameter
is assumed to be binary.
\citet{ChiniKKMS17} carry out a multi-parameter analysis of bounded context switching for a fixed number of threads.
Their main results connect the complexity of this problem to hypotheses in fine-grained complexity.
However, they do not consider binary encodings of the context switch bound.

For a class of graph algorithms, one can prove a metatheorem that shows that succinct encodings give an exponential blowup: $\NP$ becomes $\NEXP$, and so on~\cite{papadimitriou1986note}.
Our $\EXPSPACE$ result shows that the space of safety verification problems is more nuanced.

The core of our results provide new and efficient constructions on \emph{succinct} machines.
Specifically, we show that the downward closure of the language of a succinct pushdown automaton has a small representation as a doubly succinct finite automaton.
It is well known that the downward closure of a context free language is effectively computable \cite{vanLeeuwen78,Courcelle1991}.
\citet{BachmeierLS2015}, following results by \citet{GruberHolzerKutrib}, show exponential lower bounds on the state complexity of an NFA representing the downward
closure.
Later, \citet{MajumdarTZ21} showed that the NFA is ``compressible'', meaning that there is a polynomial succinctly represented NFA for the downward closure.
Our results improve the succinct representability for succinctly defined PDAs: there is a polynomial (doubly succinct) NFA representation, even though there is a
tight doubly exponential lower bound for an NFA representation.
Further, we show that the same result holds for a stronger variant of the downward closure, where a subset of the alphabet is preserved.

While the use of thread pools is usually motivated by performance concerns, the fact that it might lead to a verification problem of lower complexity despite binary encodings,
and indeed the \emph{same} complexity as the special case of one thread and execute-to-completion, came as a surprise.

%% file: prelims.tex
\section{A Model of Thread Pooling} %
\label{sec:preliminaries}

\subsection{Dynamic Networks of Concurrent Boolean Programs with Recursion ($\DCBPR$)} \label{sec:dcps}

\begin{figure}[t]
\begin{center}
\begin{minipage}{.45\textwidth}
  \begin{lstlisting}[style=tinyc,name=bitstream]
global lock l;
main() { spawn handler(); spawn main(); }
handler() {
  if * { a(0); }
  else { a(1); }
  unlock(l);
}
  \end{lstlisting}
\end{minipage} \hfill
\begin{minipage}{.45\textwidth}
  \begin{lstlisting}[style=tinyc,name=bitstream]
a(bool i) {
  if * { a(0); }
  else if * { a(1); }
  else { lock(l); }
  write(i); // critical section
  return i;
}
  \end{lstlisting}
\end{minipage}
\end{center}
\caption{An example concurrent program}
\label{fig:example}
\end{figure}

Figure~\ref{fig:example} shows a simple example of the different features of the programs that we want to verify.
In the example, the main program \emph{spawns} new tasks when it is executed: an instance of a handler task
and a further instance of itself.
A task executes sequential code, with \emph{(recursive) function calls}, \emph{parameters}, and \emph{local variables}.
They can also read and write \emph{shared global variables} (e.g., the lock~$\mathtt{l}$). 
In the example, a handler iteratively stores the value of some outside condition (modeled as non-determinism) in 
the parameter to a recursive function call. 
When the recursion ends (non-deterministically), the handler acquires the global lock $\mathtt{l}$ 
to write out the stored values (in-reverse, due to how recursion works) and then returns the lock.

Note that the program can spawn an unbounded number of handlers, that can all execute in parallel.
We shall consider a \emph{thread-pooled} execution, where a fixed number of worker threads repeatedly pick and execute
the tasks.
The threads run in an interleaved fashion and can be swapped in and out. However, each task is executed to completion
before a new spawned task is picked.

A possible safety property for the program is to ensure mutual exclusion for the write operation on line~12.
The property does hold for this program, but proving this is difficult as the state space is unbounded in several
directions: the unbounded number of spawned tasks and the unbounded stack of an executing task.

To model programs like the one above, we consider an abstract, language-theoretic model of thread-pooled shared memory programs, 
following \cite{AtigBQ2009} and subsequent work.
The model involves tasks, which can be spawned dynamically during execution, and a stack per thread to capture recursion. We furthermore augment this model with global and thread-local Boolean variables, which is a new addition when compared to prior work. The lock from the example program can be easily modeled as a global Boolean variable, and the stack and thread-local variables allow us to simulate features like Boolean function parameters. We will go into more detail about how this model captures program behavior at the end of this subsection.

\subsubsection*{Towards Syntax: Defining transducers.}
To introduce the syntax of our model, we first need to define the notion of transducers.
For $k \in \nats$, a \emph{(length preserving)} \textbf{$k$-ary transducer} 
$T = (Q,\Delta,q_0,Q_f,E)$ consists of a finite set of \emph{states} $Q$,
an alphabet $\Delta$, an \emph{initial state} $q_0 \in Q$, a set of \emph{final states} 
$Q_f \subseteq Q$, and a \emph{transition relation} $E \subseteq Q \times (\Delta^k \cup \set{\varepsilon}^k) \times Q$.
For a \emph{transition} $(q,a_1,\ldots,a_k,q') \in E$, we write 
$q \xrightarrow{(a_1,\ldots,a_k)} q'$.
The \emph{size} of $T$ is defined as $|T| = k\cdot|E|$.

The \emph{language} of $T$ is the $k$-ary relation $L(T) \subseteq (\Delta^*)^k$ 
containing precisely those $k$-tuples $(w_1,\ldots,w_k)$, for which there is a transition sequence
$q_0 \xrightarrow{(a_{1,1},\ldots,a_{k,1})} q_1 \xrightarrow{(a_{1,2},\ldots,a_{k,2})} \ldots \xrightarrow{(a_{1,m},\ldots,a_{k,m})} q_m$
with $q_m \in Q_f$ and $w_i = a_{i,1}a_{i,2} \cdots a_{i,m}$ for
all $i \in \{1, \ldots, k\}$. Such a transition sequence is called an
\emph{accepting run} of $T$.   

We note that in the more general (i.e., non-length-preserving)
definition of a transducer (see, e.g., \cite{Ginsburg}), the transition relation $E$ is a subset
of $Q \times (\Delta_\varepsilon)^k \times Q$, where $\Delta_\varepsilon = \Delta \cup \{\varepsilon\}$. All transducers we
consider in this paper are length-preserving.

\subsubsection*{Syntax}
A \emph{Dynamic Network of Concurrent Boolean Programs with Recursion} ($\DCBPR$) $\cD=(V_{\glob},V_{\loca},\Gamma,\cT_c,\cT_s,\cT_r,\cT_t,\vec{a}_0,\gamma_0)$ consists of
a finite set of \emph{global Boolean variables} $V_{\glob}=\{ v_1,v_2,\ldots,v_m\}$,
a finite set of \emph{local Boolean variables} $V_{\loca}=\{ v'_1,v'_2,\ldots,v'_n\}$,
a finite alphabet of \emph{stack symbols} $\Gamma$, 
an \emph{initial assignment} $\vec{a}_0 \in \{0,1\}^{V_{\glob}}$ of the global  (Boolean) variables,
an \emph{initial stack symbol} $\gamma_0 \in \Gamma$,
and four sets of transducers $\cT_c,\cT_s,\cT_r,\cT_t$ which succinctly describe the allowed transitions in the program, described below.

Let $\bar{\Gamma} = \set{\bar{\gamma} \mid \gamma \in \Gamma}$. 
The four sets of transducers have the following form:
\begin{enumerate}
  \item For each tuple $a=(\gamma,v)$ where $\gamma \in \Gamma \cup \{\varepsilon\}$, $v \in \Gamma \cup \bar{\Gamma} \cup \{\varepsilon\}$, there exists a $2$-ary transducer $T_a \in \cT_c$ which works over the alphabet $\Delta=\{ 0,1\}$ and accepts strings from $(\Delta \times \Delta)^{m+n}$.
  \item For each $\gamma \in \Gamma \cup \{\varepsilon\}$, there exists a $2$-ary transducer $T_\gamma \in \cT_s$ which works over the alphabet $\Delta=\{ 0,1\}$ and accepts strings from $(\Delta \times \Delta)^{m+n}$.
  \item For each $\gamma \in \Gamma$,  there exists a $2$-ary transducer $T_{\gamma} \in \cT_r$ which works over the alphabet $\Delta=\{ 0,1\}$ and accepts strings from $(\Delta \times \Delta)^{m}$.
  \item $\cT_t$ contains one $2$-ary transducer $T_t$ which works over the alphabet $\Delta=\{ 0,1\}$ and accepts strings from $(\Delta \times \Delta)^{m}$.
\end{enumerate}
The \emph{size} $|\cD|$ of $\cD$ is defined as $m+n + |\Gamma| + \sum_{T \in (\cT_c \cup \cT_s \cup \cT_r \cup \cT_t)}|T|$:
the number of global variables, the number of local variables, the stack alphabet, and the sizes of the transducers which define the valid transitions.

\subsubsection*{Intuition}
A $\DCBPR$ represents a multi-threaded program with a thread pool of fixed size and a shared global memory represented by the set $V_{\glob}$ 
of global Boolean variables. 
Computation is carried out by tasks.
The tasks can read and write global variables;
in addition, every task has its own copy of the local variables $V_{\loca}$.
Tasks make potentially recursive function calls, and we use the stack alphabet $\Gamma$ to maintain their stacks.
Each running task can manipulate the global variables, its local variables, and its stack. 
It can also spawn new tasks into a task buffer. 

Since the set of possible assignments to the global and local variables is exponentially large, the 
transitions between two different configurations of a thread, which potentially involve a change in the assignments of these variables,
are given succinctly using transducers. 

Execution of tasks is coordinated by a thread pool, which is an a priori fixed number of concurrent threads that are used to execute the tasks.
A nondeterministic scheduler schedules the threads in the thread pool.
When some thread in the thread pool is idle, it can pick one of the pending tasks from the task buffer
and start executing it.
The task is executed concurrently with other threads in the thread pool, until completion.
On completion, the executing thread will nondeterministically pick another task from the task buffer.

Intuitively, the transducers represent ``transition relations'' relating the assignments to global and local variables
of a program when it takes a step.
Each set of transducers $\cT_c,\cT_s,\cT_r,\cT_t$ corresponds to different types of transitions applicable to the $\DCBPR$. 
The transitions are divided into two major kinds: \emph{creation} transitions applicable to the running of a single thread given by $\cT_c$ 
and \emph{swap}, \emph{resumption}, and \emph{termination} transitions corresponding to the actions of the scheduler, 
given by $\cT_s$, $\cT_r$ and $\cT_t$ respectively. 
A transducer in $\cT_c$ describes the updates to the global and local variables during a single step of a thread.
A transducer in $\cT_s$ describes the updates when a thread is swapped out.
Finally, transducers in $\cT_r$ and $\cT_t$ describe how the global variables are updated when a thread is resumed and on thread termination, respectively.
We explain these in detail below.

\subsubsection*{Towards Semantics: Preliminary Definitions}
In order to provide the semantics of $\DCBPR$, we need a few definitions and notation.
A \emph{multiset} $\mmap\colon S\rightarrow\nats$ over a set $S$ maps each
element of $S$ to a natural number.
Let $\multiset{S}$ be the set of all multisets over $S$.
We treat sets as a special case of multisets 
where each element is mapped onto $0$ or $1$.
We write
$\mmap=\multi{a_1,a_1,a_3}$ for the multiset
$\mmap\in\multiset{S}$ such that $\mmap(a_1)=2$, $\mmap(a_3)=1$, and $\mmap(a) = 0$ for each $a \in S\backslash\set{a_1,a_3}$. 
The empty multiset is denoted \(\emptyset\).
The \emph{size} of $\mmap\in\multiset{S}$, denoted $\card{\mmap}$, is
given by $\sum_{a\in S}\mmap(a)$.
This definition applies to sets as well.
 
Given two multisets $\mmap,\mmap'\in\multiset{S}$ we define $\mmap +
\mmap'\in\multiset{S}$ to be a multiset such that for all $a\in S$,
we have $(\mmap + \mmap')(a)=\mmap(a)+\mmap'(a)$.
For $c\in\nats$, we define $c\mmap$ as the multiset that maps each $a\in S$ to $c\cdot \mmap(a)$.
We also define the natural order
$\preceq$ on $\multiset{S}$ as follows: $\mmap\preceq\mmap'$ if{}f there
exists $\mmap^{E}\in\multiset{S}$ such that
$\mmap + \mmap^{E}=\mmap'$. 
We also define $\mmap - \mmap'$ for $\mmap' \preceq \mmap$ analogously: for all $a\in S$,
we have $(\mmap - \mmap')(a)=\mmap(a)-\mmap'(a)$.

\subsubsection*{Semantics}
The set of configurations of $\cD$ is 
\[
\underbrace{\set{0,1}^{V_{\glob}}}_{\text{\small global state}} 
\times 
\underbrace{\big((\set{0,1}^{V_{\loca}} \times \Gamma^* \times \nats) \cup \set{\#}\big)}_{\text{\small local config or schedule pt}} 
\times
\;
\underbrace{\multiset{\set{0,1}^{V_{\loca}} \times \Gamma^* \times \nats}}_{\text{\small thread pool}}
\;
\times 
\underbrace{\multiset{\Gamma}}_{\text{\small task buffer}}
\] 
A \emph{configuration} of a $\DCBPR$ consists of 
(a) an assignment $\vec{a}$ to the global variables,
(b) a \emph{local configuration} of the active thread (or a special symbol ``$\#$'' signifying a schedule point),
(c) a \emph{thread pool}, which is a multiset of local configurations of inactive, partially executed threads,
and
(d) a \emph{task buffer} of pending (that is, unstarted) tasks.
The local configuration of a thread is a tuple $(\vec{b},w,i)$ where $\vec{b}$ is the assignment of the local variables, $w$ is the stack content,
and $i$ is a context switch number. 
The context switch number tracks how many times an executing task has already been context switched by the underlying scheduler.

In order for us to define state transitions succinctly using transducers, we need to impose an order on the set of global and local variables in 
order to form a well-defined string from $\{0,1\}^*$ given the variable values. 
Hence for a set of variables $V = \{v_1, \ldots, v_m\}$ and an assignment $\vec{a} \in \{0,1\}^V$ we identify $\vec{a}$ with 
the string $\vec{a}(v_1) \cdots \vec{a}(v_m)$.
We shall write $\langle \vec{a}, (\vec{b},w,i), \mmap, \tbuf\rangle$ or $\langle \vec{a}, \#, \mmap, \tbuf\rangle$
for configurations.
We shall assume that the thread pool $\mmap$ is bounded by a number $N > 0$: this represents an a priori 
fixed number of threads used for executing tasks.
We shall also consider the special case of $N = \infty$, that corresponds to no thread pooling
and in which every pending task can start executing immediately.

The initial configuration of $\cD$ is
$\langle \vec{a}_0, \#, \multi{(\vecz,\gamma_0,0)}, \emptyset\rangle$. For a configuration
$c$ of $\cD$, we will sometimes write $c.\vec{a}$ for the state of $c$ and
$c.\mmap$ for the multiset of threads of $c$ (both active and
inactive).  The \emph{size} of a configuration
$c = \langle \vec{a}, (\vec{b},w,i), \mmap,\tbuf \rangle$ is defined as
\[|c| = m + n + |w| + \sum_{(\vec{b}',w',j) \in \mmap}(n + |w'|) + \sum_{\gamma \in \Gamma} \tbuf(\gamma).\]

First let us explain the transitions corresponding to the action of a single thread.
The steps of a single executing task define the following
\emph{thread step} relation $\rightarrow$ on
configurations of $\cD$: 
we have $\langle \vec{a}, (\vec{b}, w,i), \mmap, \tbuf \rangle \rightarrow \langle \vec{a}', (\vec{b}',w',i), \mmap', \tbuf' \rangle$ 
for all $w \in \Gamma^*$ iff

(1) $(\vec{a}\vec{b},\vec{a}'\vec{b}') \in \langof{T_{\varepsilon,v}}$ where $T_{\varepsilon,v} \in \cT_c$, $\mmap' = \mmap$, $\tbuf' = \tbuf$ and one of the following conditions hold:
\begin{enumerate}[i]
  \item $v \in \Gamma \cup \set{\varepsilon}$ and $w' = vw$, or
  \item $v = \bar{\gamma} \in \bar{\Gamma}$ and $\gamma w' = w$.
\end{enumerate}
or (2) $(\vec{a}\vec{b},\vec{a}'\vec{b}') \in \langof{T_{\gamma',v}}$ where $T_{\gamma',v} \in \cT_c$ and $\mmap' = \mmap$ and $\tbuf' = \tbuf + \multi{\gamma'}$. Additionally, one of the stack conditions i or ii from (1) above hold.

We extend the \emph{thread step} relation $\rightarrow^+$ to be the irreflexive-transitive closure of 
$\rightarrow$; thus $c \rightarrow^+c'$ if there is a sequence $c
\rightarrow c_1 \rightarrow \ldots c_k \rightarrow c'$ for some $k\geq 0$.

Secondly, we have the actions of the non-deterministic scheduler which switches between concurrent threads.
The active thread is the one currently being executed and the
multiset $\mmap$ keeps partially executed tasks in the thread pool; the size of $\mmap$ is bounded by the 
size $N$ of the thread pool.
All other pending tasks that have not been picked for execution yet remain in the task buffer $\tbuf$.
The scheduler may interrupt a thread based on the interruption transitions, non-deterministically
resume a thread based on the resumption transitions, and terminate a thread based on the termination transitions. Picking a new task to execute is handled independently of the transducers and does not change any variable assignments.

The actions of the scheduler define the
\emph{scheduler step} relation $\mapsto$ on configurations of $\cD$: 
\[
\small
\begin{array}{c}
\inferrule[Swap]{
(\vec{a}\vec{b},\vec{a}'\vec{b}') \in \langof{T_\gamma} \text{ where } T_\gamma \in \cT_s
}{
\tuple{\vec{a}, (\vec{b}, w, i), \mmap, \tbuf} \mapsto \tuple{\vec{a}', \#, \mmap+ \multi{\vec{b}', \gamma w, i+1}, \tbuf}
}
\quad
\inferrule[Resume]{
(\vec{a},\vec{a}') \in \langof{T_{\gamma}} \text{ where } T_{\gamma} \in \cT_r
}{
\tuple{\vec{a}, \#, \mmap + \multi{\vec{b},\gamma w, i},\tbuf} \mapsto \tuple{\vec{a}', (\vec{b},\gamma w, i), \mmap, \tbuf}
}
\\
\smallskip
\\
\inferrule[Terminate]{
  (\vec{a},\vec{a}') \in \langof{T_t} \text{ where } T_{t} \in \cT_t
}{
  \tuple{\vec{a}, (\vec{b},\varepsilon, i), \mmap, \tbuf} \mapsto \tuple{\vec{a}', \#, \mmap, \tbuf}
}
\quad
\inferrule[Pick]{
  |\mmap| < N
}{
  \tuple{\vec{a}, \#, \mmap, \tbuf + \multi{\gamma}} \mapsto \tuple{\vec{a}, \#, \mmap + \multi{\vecz,\gamma, 0}, \tbuf}
}
\end{array}
\]
If a thread can be interrupted, then {\sc Swap} swaps it out and increases the
context switch number of the executing task. The global and local variables can both be changed and the corresponding transducer accepts words of length $m+n$. In the {\sc Resume} and {\sc Terminate} rules, only the global variables are modified and so the transducers accept words of length $m$.
The rule {\sc Resume} picks a thread that is ready to run based on the current assignment of
global variables and its top of stack symbol and makes it active. 
The rule {\sc Terminate} removes a task on termination (empty stack), freeing up the thread that was executing it. 
The rule {\sc Pick} picks a pending task for execution, if there is space in the thread pool.
Note that if $N = \infty$, the rule is always enabled.

A \emph{run} of a $\DCBPR$ is a finite or infinite sequence of
alternating thread execution and scheduler step relations 
\[
  c_0 \rightarrow^+ c_0' \mapsto^+ c_1 \rightarrow^+ c_1' \mapsto^+ \ldots 
\]
such that $c_0$ is the initial configuration.
The run is \emph{$N$-thread-pooled} if the thread pool is bounded by $N$.
The run is \emph{$K$-context switch bounded} if, moreover, for each $j\geq 0$, the
configuration $c_j = (\vec{a}, (\vec{b},w, i), \mmap,\tbuf)$
satisfies $i \leq K$.
In a $K$-context switch bounded run, each thread is context switched
at most $K$ times and the scheduler never schedules a
thread that has already been context switched $K+1$ times.
When the distinction between thread and scheduler steps is not
important, we write a run as a sequence $c_0 \Rightarrow c_1 \ldots$.

\subsubsection*{Modelling programs}

To model programs like the example at the start of this section, there are two main challenges that we need to overcome. Firstly, $\DCBPR$ define possible changes to the variables via transducers, while programs instead use simple conditional statements and assignments. Secondly, $\DCBPR$ have thread-local variables while programs usually have function-local ones.

Let us begin by arguing that for conditionals and assignments, it suffices to simulate \texttt{if}-statements over only single variables and assignments involving only the constant values $0$ and $1$. Regarding the latter, note that we can rewrite an assignment involving an arbitrary Boolean expression $\varphi$ to two constant assignments by using an \texttt{if}-\texttt{else}-block over $\varphi$. Furthermore, \texttt{else}-statements can be easily rewritten to \texttt{if}-statements over a conjunction of the negations of all preceding Boolean expressions in the block. Finally, to rewrite an \texttt{if}-statement over $\varphi$, we can assume that $\varphi$ is in negation normal form. Then $\varphi_1 \wedge \varphi_2$ can be rewritten as two nested \texttt{if}-statements over the two respective parts, while $\varphi_1 \vee \varphi_2$ can be rewritten as two non-nested \texttt{if}-blocks right after one another. We do not need to rewrite $\neg \varphi$, as in this case $\varphi$ will be a single variable. To not duplicate the code inside of an \texttt{if}-block in the $\vee$-case we can make use of \texttt{goto}-statements, which our transducers can also simulate. This way we get at most a polynomial blow-up.

We continue by explaining how a transducer simulates the resulting \texttt{if}-statements, assignments, and \texttt{goto}-statements. Firstly, we extend the $\DCBPR$ by finitely many additional local Boolean variables, to store the current line number of any given thread. Then an accepting run $\rho$ in the transducer for a particular statement in the program takes the binary representation of the corresponding line number as input, and outputs either the line number directly below the statement, or the target line number in case of a \texttt{goto}. For an \texttt{if}-statement on variable $v$ (respectively $\neg v$) the input on the remainder of $\rho$ matches the output, while the transducer checks that there is a $1$ (respectively $0$) at the position corresponding to $v$. For an assignment of $1$ (respectively $0$) to $v$, the transducer instead outputs $1$ (respectively $0$) at the position corresponding to $v$, while keeping the rest of the input unchanged.

Now we also need to argue that we can simulate function-local variables by thread-local ones. The main trick is to push all values of local variables to the stack before pushing a function call, upon which we reset all values of local variables. When we return, we save the return value in a special variable followed by popping the function call from the stack, upon which we pop all variable values from the stack and store each one in its corresponding variable.

More formally, $V_{\glob}$ would contain the names of all global Boolean variables, and $V_{\loca}$ would contain all local variable names in all functions, all line numbers in the code, and one special variable \texttt{return}. The stack alphabet $\Gamma$ would contain all function names and two copies of each function-local variable in the program, one for each truth value. New tasks would simply be spawned with one function name on the stack. Note that new tasks are always initialised with all local variables set to $0$. Each task's current line number would be stored by setting the corresponding variable to true and all other such variables to false. The transducers would be defined in a way to facilitate all the program operations.

\subsection{Decision Problems and Main Results}

The \emph{reachability problem} for $\DCBPR$ asks, given an assignment $\vec{a}$ to the global variables of $\cD$ (henceforth also called a \emph{global state}),
if there is a run $c_0 \Rightarrow c_1 \ldots \Rightarrow c_\ell$ such
that $c_\ell.\vec{a} = \vec{a}$.
Reachability of multi-threaded recursive programs is undecidable already with two fixed threads (see Section 3.2 in \cite{Ramalingam}).

We consider \emph{context switch bounded} decision questions.
Given $K \in \nats$, a global state $\vec{a}$ of $\cD$
is $K$-context switch bounded reachable
if there is a $K$-context switch bounded run $c_0
\Rightarrow \ldots \Rightarrow c_\ell$ with $c_\ell.\vec{a} = \vec{a}$.

The \emph{thread-pooled context-bounded reachability problem} is the following:
\begin{description}
  \item[Given] A $\DCBPR$ $\cD$, a global state $\vec{a}$, and two numbers $K$ and $N$ in \emph{binary};
  \item[Question] Is $\vec{a}$ reachable in $\cD$ in an $N$-thread-pooled, $K$-context switch bounded execution?
\end{description}

We show the following result.

\begin{theorem}[Safety Verification with Thread-Pooling]
\label{th:complexity-thread-pooling}
The thread-pooled context-bounded reachability problem is $\EXPSPACE$-complete.
\end{theorem}

$\EXPSPACE$-hardness holds already for a $\DCBPR$ $\cD$, $K = 0$, and $N=1$, even when the global state is restricted
to be an explicitly given finite set (rather than as valuations to variables) and there are no additional local variables.
This follows from results on \emph{asynchronous programs} \cite{GantyM12}, 
which can be seen as the special case of a simpler model where the global states are specified as a finite set of possible values, 
and there is a single thread, which executes each task to completion before picking the next task.
In other words, the case with non-succinctly defined states, a thread pool of size $N = 1$, and no context-switching (i.e.\ $K = 0$).
Thus, our main aim is to show the upper bound.

In case $N = \infty$, there is no thread pooling and an arbitrary number of tasks can execute simultaneously.
Then, the general \emph{context-bounded reachability problem} takes as input a $\DCBPR$, a global state, and a context bound $K$ in binary,
and asks if the global state is reachable in a $K$-context switch bounded run, assuming there is no bound on the thread pool. 
If furthermore the global states are defined explicitly and there are no local variables, 
then this results in a setting, which has been studied before: \citet{AtigBQ2009} showed that the context-bounded reachability problem is in $\TWOEXPSPACE$ 
for a fixed $K$ (or when $K$ is given in unary).
A matching lower bound for all fixed $K \geq 1$ (or equivalently, for $K$ given in unary)
was provided by \citet{BaumannMajumdarThinniyamZetzsche2020a}.  
We study the problem when $K$ is given in binary and states are defined in a succinct manner via global and local Boolean variables. 
We show the following result.

\begin{theorem}[Safety Verification without Thread Pooling]
\label{th:complexity-no-thread-pooling}
The context-bounded reachability problem is $\THREEEXPSPACE$-complete. 
\end{theorem}

The $\THREEEXPSPACE$ upper bound follows from a relatively simple analysis of the upper bound of \citet{AtigBQ2009}.
For the lower bound, we extend the $\TWOEXPSPACE$-hardness result of \citet{BaumannMajumdarThinniyamZetzsche2020a} with new encodings.

\subsection{Outline of the Rest}

As discussed in the Introduction, our $\EXPSPACE$ upper bound has three parts, following the general outline of the proof of decidability
for context-bounded reachability \cite{AtigBQ2009}.

The \textbf{first step} is to distill a succinct pushdown automaton for each task, and to construct a finite-state automaton that represents the downward closure
of its language.
In Section~\ref{sec:bounded_task_buffer}, we prove general results about succinct pushdown automata and doubly succinct representations of their downward closures.
We then apply these results to $\DCBPR$ in Section~\ref{sec:succinct_dcl_tasks} to obtain succinct representations of the downward closures of tasks,
where the number of context switches is preserved.

The \textbf{second step} is to construct a vector addition system with states ($\VASS$), a model equivalent to Petri nets, that represents the state of a $\DCBPR$.
Again, in order to overcome the exponential blow-up of representing all $N$ threads in the thread pool (remember, $N$ is in binary), we first introduce
a succinct version of $\VASS$ and show the coverability problem for the succinct version can be solved in $\EXPSPACE$ (Section~\ref{sec:TCVASS}).

In the \textbf{third step}, we show a reduction from $\DCBPR$ to succinct $\VASS$, where we use the doubly succinct representation for the downward closures of all task languages.

Together, we obtain an $\EXPSPACE$ upper bound, as promised in Theorem~\ref{th:complexity-thread-pooling}.

Finally, in Section~\ref{sec:context_bound_as_binary_input}, we sketch the $\THREEEXPSPACE$ lower bound in case the thread pool is unrestricted, and also give a very short argument as to why the matching upper bound follows from the results of \citet{AtigBQ2009}.

\begin{remark}
  We note that our model of $\DCBPR$ assumes that tasks are not spawned with parameters that initialize the local variables; each local variable is always initialized to $0$ instead.
  This in particular means that while we can model passing Boolean parameters as part of recursive function calls, we cannot model passing such parameters as part of spawning new tasks.
  
  If we allowed parameters in task creation, the thread-pooled context-bounded reachability problem would become $\TWOEXPSPACE$-complete.
  For the membership, observe that there would be now exponentially many types of tasks in the buffer, one for each stack symbol $\gamma$ and each assignment $\vec{b}$ to the local variables.
  Since our $\EXPSPACE$ algorithm for this problem relies on a polynomially-sized task buffer, we get a blow-up when using the same methods.
  For the hardness, one can apply a slightly altered version of the result from \cite[Remark~6]{BaumannMajumdarThinniyamZetzsche2020a}.
  While the unaltered version of this result does require a very large number of threads running at the same time, note that this is only necessary because each task needs to start its execution to receive its initial assignment to the local variables.
  With passing of parameters in task creation, initial assignments can be set for unstarted tasks as well, eliminating the need for a large thread pool.
\end{remark}

%% file: bddThreads.tex
\section{Succinct Pushdown Automata and Succinct Downward Closures} %
\label{sec:bounded_task_buffer}

Some of the detailed constructions for Sections \ref{sec:bounded_task_buffer}-\ref{sec:TCVASS} are quite technical, and therefore provided in the full version.

\subsection{Preliminaries: Pushdown Automata}

For an alphabet $\Gamma$, we write
$\bar{\Gamma}=\{\bar{\gamma} \mid \gamma\in\Gamma\}$.  Moreover, if
$x=\bar{\gamma}$, then we define $\bar{x}=\gamma$.
A \emph{pushdown automaton} ($\PDA$)
$\cP = (Q, \Sigma, \Gamma, E, q_0, \gamma_0, Q_F)$ consists of 
a finite set of \emph{states} $Q$, 
a finite input alphabet $\Sigma$, 
a finite alphabet of \emph{stack symbols} $\Gamma$, 
an \emph{initial state} $q_0 \in Q$, 
an \emph{initial stack symbol} $\gamma_0 \in \Gamma$, 
a set of \emph{final states} $Q_F \subseteq Q$,
and a transition relation $E \subseteq (Q \times \Sigma_\varepsilon \times \hat{\Gamma} \times Q)$, 
where $\Sigma_\varepsilon = \Sigma \cup \{\varepsilon\}$ and 
$\hat{\Gamma}=\Gamma \cup \overline{\Gamma} \cup \{\varepsilon\}$. 
For $(q,a,v,q') \in E$ we also write $q \xrightarrow{a|v} q'$.

The set of \emph{configurations} of $\cP$ is $Q \times \Gamma^*$.
For a configuration $(q, w)$ we call $q$ its \emph{state} and $w$ its \emph{stack content}.
The \emph{initial configuration} is $(q_0, \gamma_0)$. 
The set of \emph{final configurations} is $Q_F \times \Gamma^*$. 
For configurations $(q,w)$ and $(q',w')$, we write
$(q,w)\xRightarrow{a}(q',w')$ if there is an edge $(q,a,v,q')$ in $E$
such that (i)~if $v=\varepsilon$, then $w'=w$, 
(ii)~if $v\in\Gamma$, then $w'=wv$, and 
(iii)~if $v=\bar{\gamma}$ for $\gamma\in\Gamma$, then $w=w'\gamma$.

Informally, an edge $(q, a, v, q')\in E$ with $v\in\Gamma$ denotes the ``push'' of $v$ onto the stack,
and, for every $\gamma\in \Gamma$, the letter $\bar{\gamma}$ denotes the ``pop'' of $\gamma$ from the stack. 

For two configurations $c, c'$ of $\cP$, we write $c \Rightarrow c'$ if $c \xRightarrow{a} c'$ for some $a$. 
Furthermore, we write $c \xRightarrow{\smash{u}}^* c'$ for some $u \in \Sigma^*$ if there is a sequence of configurations $c_0$ to $c_n$ with 
\[
  c = c_0 \xRightarrow{a_1} c_1 \xRightarrow{a_1} c_2 \cdots c_{n-1} \xRightarrow{a_n} c_n = c', 
\]
such that $a_1 \ldots a_n = u$. We then call this sequence a \emph{run} of $\cP$ over $u$.
We also write $c \Rightarrow^* c'$ if the word $u$ does not matter.
If $w_i$ is the stack content of $c_i$, then the \emph{stack
  height} of the run is defined as $\max\{|w_i| \mid 0\le i\le n\}$. If the run has stack height at most
  $h$ then we write $c \pdastep{u}{h} c'$.
  A run of $\cP$ is \emph{accepting} if $c$ is initial and $c'$ is final.
Given two configurations $c, c'$ of $\cP$ with $c \Rightarrow^* c'$, we say that $c'$ is \emph{reachable} from $c$ and 
that $c$ is \emph{backwards-reachable} from $c'$. If $c$ is the initial configuration, we simply say that $c'$ is reachable. 

The \emph{language} accepted by $\cP$, denoted $L(\cP)$
is the set of words in $\Sigma^*$, for which $\cP$ has an accepting run. 
We also define a language accepted with bounded stack height. 
For $h\in\nats$, we define 
\[ \langof[h]{\cP}=\{u\in\Sigma^* \mid 
  (q_0,\varepsilon)\pdastep{u}{h}(q_f,w) \text{ where } q_f \in Q_F \}. \]
Note that the \emph{bounded stack language} $\langof[h]{\cP}$ is regular for any $\PDA$ $\cP$ and any fixed $h$.

A \emph{nondeterministic finite state automaton} ($\NFA$) is a $\PDA$ where every edge $q \autstep[a|v]q'$
satisfies $v=\varepsilon$. Equivalently, an $\NFA$ $\cA=(Q,\Sigma,E,q_0,Q_F)$ is obtained by removing the stack 
and having edges of the form $q \autstep[a] q'$ for $q,q' \in Q$ and
$a \in \Sigma_{\varepsilon}$.
The configurations now consist only of the state, and the concepts of runs, language, and
reachability are appropriately modified.

\subsection{Succinct $\PDA$s and Doubly Succinct $\NFA$s}

We shall consider \emph{succinct} representations of $\PDA$s and $\NFA$s.
There are different possible choices for
succinct representations; ours are based on \emph{transducers}. 
A succinct representation is parameterized by a size parameter $n$, given in unary, and the goal is to define
a $\PDA$ or $\NFA$ whose states are words of length $n$ over a fixed alphabet $\Sigma$ and whose transitions
are succinctly described using transducers (of size at most logarithmically dependent on $n$). 
Thus, the number of states of the underlying machine is exponential in the size parameter.
We call such a specification \emph{singly succinct} or just \emph{succinct}. 
When the words used to specify the states are of length $2^n$ (and transducers are of size at most doubly logarithmically dependent on $n$), then we call the specification \emph{doubly succinct.}

For a stack alphabet $\Gamma$ and an input alphabet $\Sigma$, let $
\scrC=\{ (a,\gamma) \mid a \in \Sigma_\varepsilon, \gamma \in \Gamma \cup \overline{\Gamma} \cup \{\varepsilon\} \}$. 
A \emph{succinct $\PDA$} is a tuple $\mathcal{C}=((T_c)_{c \in \scrC},\Sigma,\Gamma,\Delta, w_0,w_f)$,
where $\Sigma$ is an input alphabet, $\Gamma$ is a stack alphabet, $\Delta$ is a transducer alphabet,
$w_0,w_f \in \Delta^n$ for some $n \geq 0$, and
$T_c$, for each $c\in \scrC$, 
is a $2$-ary transducer over the alphabet $\Delta$. Note that in particular, $n$ is implicitly given as the length of the words $w_0$ and $w_f$ and is thus polynomial in the size of the succinct $\PDA$.

A succinct $\PDA$ represents an (explicit) $\PDA$ $\cE(\cC)$ with 
\begin{itemize}
  \item states $Q=\Delta^n$,
  \item initial state $w_0$ and final state $w_f$,
  \item input alphabet $\Sigma$ and stack alphabet $\Gamma$, and 
  \item transition relation $E\subseteq Q\times \Sigma_\varepsilon \times \hat{\Gamma} \times Q$,
  where $q \xrightarrow{a | \gamma } q'$ belongs to $E$ iff $(q,q') \in L(T_{(a,\gamma)})$.
\end{itemize}
The runs, accepting runs, language, etc.\ of a succinct $\PDA$ $\cC$ refer to those of the explicit $\PDA$ $\cE(\cC)$.
The size of $\cC$ is defined as $|\cC| = |w_0| + |w_f| + \sum_{c \in \scrC}|T_c|$.

A \emph{succinct $\NFA$} is a succinct $\PDA$ where the stack is immaterial, and defined in the expected way.

A \emph{doubly succinct} $\NFA$ ($\dsNFA$ in short) is a tuple 
$\cB=((T_a)_{a \in \Sigma_{\varepsilon}},\Sigma,\Delta,M,w_0)$,
where $\Sigma$ is an input alphabet, $\Delta$ is a transducer alphabet, 
$M\in\nats$ is a number given in \emph{binary}, $w_0 \in \Delta^*$ is a prefix of the initial state, and 
$T_a$ is a $2$-ary transducer over $\Delta$ for each $a \in \Sigma_{\varepsilon}$.
In the sequel, we will refer to $w_0$ as the \emph{initial prefix}: the ability to specify $w_0$ will be helpful in uniformly specifying a collection of $\dsNFA$s in a succinct mannner.
We assume $\set{0, 1} \subseteq \Delta$.
A $\dsNFA$ represents an (explicit) $\NFA$
$\cE(\cB) =(Q,\Sigma,E,q_0,\{q_f\})$ where
\begin{itemize}
  \item the set of states $Q=\Delta^{M}$, 
  \item the initial state $q_0$ is $w_00^{M-|w_0|}$ and the unique final state $q_f$ is $1^{M}$, 
  \item there exists a transition $(p \autstep[a]q) \in E$ for $a \in \Sigma_{\varepsilon}$ iff $(p,q) \in L(T_a)$.
\end{itemize}
The size of $\cB$ is defined as $|\cB| = |w_0| + \lceil\log M\rceil + \sum_{a \in \Sigma_{\varepsilon}}|T_a|$.
Note that, for a $\dsNFA$ $\cE(\cB)$, the number of states $|Q|$ is doubly exponential in the description of $\cB$,
since $M$ is written in binary.

\subsection{Succinct Downward Closures}

Next, we move on to computing \emph{downward closures} of languages.

The \emph{subword} order $\subw$ on finite words over an alphabet $\Sigma$ is defined as:
for all $u,v \in \Sigma^*$ with $u=u_1\ldots u_n$ where each $u_i \in \Sigma$, 
$u \subw v$ if and only if there exist $w_0,w_1,\ldots, w_n \in \Sigma^*$ such that
$v=w_0u_1w_1u_2\ldots w_{n-1}u_nw_n$. 
The subword order is a well-quasi-ordering \cite{Higman52}.
Given any $L \subseteq \Sigma^*$, its \emph{downward closure} 
(also called \emph{downclosure}) $\dcl{L}$ is given by
$\dcl{L}=\{ u \mid \exists v \in L, u \subw v\}$.
The downward closure of any language is regular \cite{haines1969free}.

We show the following language theoretic result of independent interest.

\begin{theorem}[Succinct Downward Closures]
  \label{th:succinct-dcl}
For every succinct $\PDA$ $\cC$, there is a polynomial-time construction of
a doubly succinct $\NFA$ $\cB$ of size polynomial in $|\cC|$ such that
$\langof{\cB} = \dcl{\langof{\cC}}$.
\end{theorem}

It is well known that an $\NFA$ for the downward closure of a $\PDA$
can be exponential in the size of the $\PDA$ \cite{BachmeierLS2015}.
Moreover, it is not difficult to construct examples of succinct $\PDA$ for which
a downward closure $\NFA$ is least doubly exponentially large\footnote{Take, for example, the language $L_n=\{w\rev{w} \mid w\in\{\ltr{a},\ltr{b}\}^*,~|w|=2^n\}$. For each $n$, $L_n$ is accepted by a succinct PDA of size polynomial in $n$, but a downward closure $\NFA$ requires at least $2^{2^n}/2^n$ states.}.
A result of \citet{MajumdarTZ21} shows that $\NFA$s for the downward
closure of a $\PDA$ are (singly) compressible.
Theorem~\ref{th:succinct-dcl} strengthens the result and shows that even for succinct $\PDA$s, the $\NFA$s are 
``doubly-compressible'': there is a doubly succinct $\NFA$ for the downward closure.

We prove the theorem using the following two lemmas.
The first lemma converts a succinct pushdown automaton $\cC$ to a
new succinct pushdown automaton $\cC'$ whose bounded stack 
language (for a stack bound that is at most exponential in the size of the
succinct $\PDA$) is the same as the downward closure of $\cC$.
The second lemma turns a succinct pushdown automaton into a 
doubly succinct NFA. 

\begin{lemma}\label{bounded-stack-dcl}
 Given a succinct pushdown automaton $\cC$, one can construct in polynomial time a succinct pushdown 
   automaton $\cC'$ and a number $h$ in binary such that $\langof[h]{\cC'}=\dcl{\langof{\cC}}$. 
 \end{lemma}

\begin{lemma}\label{bounded-stack-nfa}
  Given a succinct pushdown automaton $\cC$ and a number $h$ in 
  binary, one can construct in polynomial time a doubly succinct NFA $\cB$ with 
  $\langof{\cB}=\langof[h]{\cC}$. 
\end{lemma}

Theorem~\ref{th:succinct-dcl} follows by composing the constructions of the two lemmas.

Before proving the lemmas, we first recall the construction from \cite{MajumdarTZ21} that takes an 
ordinary $\PDA$ $\cP=(Q,\Sigma,\Gamma,\delta,q_0,Q_F)$ and constructs in polynomial time a $\PDA$ $\cP^\top$ and a bound $h$
such that the downclosure $\dcl{\langof{\cP}}$ is the bounded stack language of $\cP^\top$.

From $\cP$, we construct an \emph{augmented automaton} ${\cP}^\top =(Q^\top,\Sigma,\Gamma^\top,E^\top,q_0,Q_F)$ 
as follows.
The states $Q^\top$ consist of the states in $Q$ together with some additional states we describe below.
The stack alphabet $\Gamma^\top$ of ${\cP}^\top$ consists of the stack alphabet of $\cP$ together with a fresh
stack symbol $[p,q]$ for every $p,q\in Q$, i.e., $\Gamma^\top = \Gamma \cup \set{[p,q] \mid p, q \in Q}$.

The set $E^\top$ consists of all the edges in $E$ together with the following additional edges, that we describe next.
To begin, define the language
\[ M_{p,q}(\cP)=\{ u\in\Sigma^*\mid \exists v\in\Gamma^*\colon (p,\varepsilon)\autsteps[u](p,v),~ (q,v)\autsteps(q,\varepsilon)\}. \]
In other words, $M_{p,q}$ consists of words which can be read on a cycle on $p$ such that the stack content 
created is consumed by a cycle on $q$. 
Furthermore, define
\[ 
\eta_{p,q}(\cP)=\{a\in\Sigma \mid \exists u\in M_{p,q}(\cP),~|u|_a\ge 1 \}.
\]
where $|u|_a$ denotes the number of occurrences of the letter $a\in \Sigma$ in $u$.

With this, we first add the following edges for each $p, q\in Q$
(the edges below use tuples $(R|v)$ where $R$ is a 
regular language instead of $(a|v)$ where $a \in \Sigma$, but we can simply paste the automaton for $R$ instead to
convert to our current formulation):
\begin{align}
\label{newedges}
  &p\autstep[\eta_{p,q}(\cP)^*|[p,q]] p, &&q\autstep[\eta_{q,p}(\overline{\cP})^*|\overline{[p,q]}] q.
\end{align}
Here,  $\overline{\cP}$ denotes the \emph{dual automaton} of $\cP$, and is obtained from
$\cP$ by changing each edge $p\autstep[a|v]q$ into $q\autstep[a|\bar{v}] p$. 
Note that $\langof{\overline{\cP}}$ is just $\rev{\langof{\cP}}$, i.e., the set of the reversals of words from $\langof{\cP}$.  

Second, we add an edge $p \xrightarrow{\varepsilon|\gamma } q$ for every edge $p \xrightarrow{a|\gamma } q$ 
for any $a \in \Sigma, \gamma \in \Gamma \cup \overline{\Gamma} \cup \set{\varepsilon}$. Here we essentially just drop the reading of input letters.

The sets $\eta_{p,q}(\cP)$ are constructed in polynomial time by checking emptiness of the languages
$M_{p,q} \cap \Sigma^*a \Sigma^*$ for each $a\in\Sigma$, for which one can construct a $\PDA$ $\cP_{a,p,q}$. 
To this end, the state set of $\PDA$ $\cP_{p,q}$ (without $a$) for $M_{p,q}$ consists of two disjoint copies of $Q$, the states of $\cP$. Hence, the state set is $\{0,1\}\times Q$.
The start state of $\cP_{p,q}$ is $ (0,p)\in \{0\}\times Q$ and the final state is $(1,q) \in \{1\}\times Q$. 
The transitions in $\{0\}\times Q$ are the same as in $\cP$ without any change. In the transitions in $\{1\}\times Q$, we drop 
the reading of input letters (replacing them with $\varepsilon$ in the transitions). 
Finally, we add an $\varepsilon$-transition from $(0,p) \in \{0\}\times Q$ to $(1,q) \in \{1\}\times Q$ which does not change the stack. 
Intersecting this $\PDA$ with the regular language $\Sigma^*a \Sigma^*$ gives us the $\PDA$ $\cP_{a,p,q}$.

In the construction of \cite{MajumdarTZ21}, the emptiness checks for $\cP_{a,p,q}$ were computed separately in polynomial time.
The main result of \cite{MajumdarTZ21} shows that there is a bound $h = O(|Q^\top|^2)$ such that $\langof[h]{\cP^\top} = \dcl{\langof{\cP}}$.
However, for our construction on succinct $\PDA$s, this leads to exponentially many checks, each potentially taking exponential time.
(Note that the emptiness problem for succinct $\PDA$ is $\EXPTIME$-complete.)
Therefore, we have to modify the construction as follows.

We observe that the emptiness of the languages 
$M_{p,q} \cap \Sigma^*a \Sigma^*$, for $p,q\in Q$ and $a\in \Sigma$, can in fact be implemented 
directly within $\cP^\top$ using its own stack on the fly by adding some additional states 
(corresponding to the union of all the states of all the $\cP_{a,p,q}$).
Any point when ${\cP}^\top$ is in state $p$, it nondeterministically guesses a letter $a \in \Sigma$ 
and puts the special stack symbol $\$$ on its stack.
Then, it runs the $\PDA$ $\cP_{a,p,q}$ for $M_{p,q} \cap \Sigma^*a \Sigma^*$ using its own stack. 
If $\cP_{a,p,q}$ accepts, the stack is popped all the way down including the $\$$ symbol and the computation is continued from there on. 
This requires augmenting the stack alphabet $\Gamma^\top$ with $\$$ and $\dagger$.
Since there is a run reaching the final state of $\cP_{a,p,q}$ iff there is a run reaching the final state with a stack 
bounded by $O(|Q^\top|^2)$ by a hill cutting argument, 
we obtain $\langof[h]{{\cP}^\top}=\dcl{\langof{\cP}}$ where $h=O(|Q^\top|^2)$. 

We now proceed on to the proofs of the lemmas.

\begin{proof}[Proof of Lemma~\ref{bounded-stack-dcl}]
  The proof idea follows mostly from the adaptation of the construction from \cite{MajumdarTZ21} to succinct $\PDA$s that we discussed above.

  The key idea is to think of state of a succinct machine as a tape on
  which polynomial space computations can be performed. By using 
  a product alphabet $\Delta \times \Delta$ for the tape of $\cC'$, we further
  think of it as having two separate tracks.
  The first track is used to simulate the transitions of $\cC$ while the second is used to perform 'on the fly' computations in order to check for the emptiness of $M_{p,q} \cap \Sigma^*a \Sigma^*$, in order to decide if one is allowed to take edges of the form in Equation \ref{newedges}. We argue that polynomial amount of tape space suffices in order to accomplish this.  

  A detail here is that we can no longer utilize a fresh stack symbol $[p,q]$ for every $p,q \in Q$.
  The reason is the existence of now exponentially many choices for $q$ and $p$, so this would result in a blow-up of the stack alphabet.
  As a remedy for this issue, we instead add the brackets `$[$' and `$]$' as fresh stack symbols, as well as adding the transducer alphabet $\Delta$ of the succinct $\PDA$ to the new stack alphabet.
  That way we can just push and pop the entire string $[p\#q]$ in $2n + 3$ steps, where $\#$ is another new stack symbol used as a separator.
\end{proof}

\begin{proof}[Proof of Lemma~\ref{bounded-stack-nfa}]
  Given a succinct $\PDA$ $\cC=((T_c)_{c \in \scrC},\Sigma,\Gamma,\Delta, w_0,w_f)$ with $|w_0|=n$,
   we construct a doubly succinct NFA $\cB=((T_a)_{a \in \Sigma_{\varepsilon}},\Sigma,\Delta',w'_0,M)$ as follows.
  The states of the explicit NFA $\cE(\cB)$ are used to store configurations $(q,w)$,
  where $q\in \Sigma^n$ is a state of the explicit $\PDA$ $\cE(\cC)$, and $w \in \Gamma^*$ with $|w| \leq h$ is the bounded stack. 

We define the transducer alphabet $\Delta'=\Delta \cup \Gamma \cup \{ \$, \#,\dagger\}$ where $\#,\dagger$ are new symbols, and declare $M=|w_0|+h+3$.
A configuration $(q, w)$ is represented as the string $\$q\# w\dagger 0^{h-|w|}$, where the portion 
of the string between $\#$ and $\dagger$ encodes the stack contents. The prefix $w'_0$ is defined to be $\$w_0\#\dagger$ so that the initial state then becomes $\$w_0\#\dagger0^h$. 

  There is a transducer $T_a$ for each $a \in \Sigma_{\varepsilon}$ that converts an input string representing the configuration
  $(p,w)$ to an output string $(q,w')$ for a transition $(p,w) \xrightarrow{a|v} (q,w')$ of the explicit
  $\PDA$. 
  The transducer nondeterministically guesses $c=(a,v)$ and applies $T_c$ to $p$, converting it 
  to $q$. This part ends when the letter $\#$ is seen.
  It then converts $w$ to $w'$ where the two strings are related by the stack operation $v$.
  The symbols on the stack that are not at the top are maintained as they are; the transducer 
  nondeterministically guesses when the last stack symbol occurs in $w$ and updates the stack based on whether
  $v \in \Gamma$, $v\in \bar{\Gamma}$, or $v =\varepsilon$, keeping the necessary information in its state.

  The transducer $T_\varepsilon$ nondeterministically checks letter by letter if the current state is $\$w_f\# w$ for some $w$ and 
  transforms the state to $1^M$, the accepting state of $\cB$.
 \end{proof}

\section{Succinct Downward Closures of Tasks} %
\label{sec:succinct_dcl_tasks}

 \subsection{Tasks in a $\DCBPR$} \label{sub:threads}
The $\DCBPR$ model does not assign task identifiers.
Nevertheless, it is convenient to be able to talk about the run of a
single task along the execution.
One can formally introduce unique task identifiers by modifying the
thread step and
the operational semantics rules to carry along the identifier in the
local state.
In this way, we can talk about the run of a single task,
the multiset of tasks spawned by a given task, etc.

With this intuition, consider the run of a specific task $t$,
that starts executing from some initial assignment $\vec{a}$ to the global variables (henceforth simply called \emph{global state}) with an initial stack
symbol $\gamma$ and initial local variable values $\vec{b} = \vecz$ (henceforth called \emph{local state}), from the moment it is started by the thread pool (by
executing the {\sc Pick} rule).
In the course of its run, the thread executing the task $t$ updates its own local stack and spawns new tasks,
but it can also get swapped out and swapped back in. 

The run of such a task $t$ corresponds to the run of an associated succinct $\PDA$ $\cC_{(\vec{a},\gamma)}$ (called a succinct task-$\PDA$) that
can be extracted from $\cD$; our construction is a simple modification
of the construction in \cite{AtigBQ2009} which now accounts for succinctness. 
The reason we require a succinct $\PDA$ in our setting is that the global states are now succinctly represented by a set 
of global Boolean variables and the number of possible assignments to these variables is exponential. 
Since the different succinct task-$\PDA$s only differ in their initial state based on the value of $\vec{a}$ and 
otherwise have the same set of transducers defining them, the entire set of succinct task-$\PDA$s have a small description. 
In the sequel, we will describe the  succinct task-$\PDA$ without mentioning the initial and final states and 
simply denote it by $\cC_{\gamma}=((T'_p)_{p \in \scrC},\Sigma',\Gamma,\Delta')$. 
If $\Gamma$ is the stack alphabet of the $\DCBPR$, and $\Delta = \{0,1\}$, then the alphabet $\Sigma'$ of the succinct $\PDA$ is given by
$\Sigma' = \Gamma \mathbin{\dot\cup} \tilde{\Gamma} \mathbin{\dot\cup} \Delta \mathbin{\dot\cup} \{\bot\}$, 
where $\tilde{\Gamma} = \set{\tilde{\gamma} \mid \gamma \in \Gamma}$ is a decorated copy of $\Gamma$, and 
$\Delta' = \Gamma \mathbin{\dot\cup} \Delta \mathbin{\dot\cup} \{ 0',1',\#,\$_1,\$_2,\$_3, \dagger, \bot_1,\bot_2\}$. 
When $\vec{a}$ is also specified, we will call $\cC_{(\vec{a},\gamma)}$ an initialised task-$\PDA$. 
The initial state of the initialised task-$\PDA$ is then taken to be $\# \vec{a} \# \vecz$ and the final state is the all $1$s string.

The succinct $\PDA$ $\cC_{\gamma}$ updates the global variables, local variables and the stack using the transducers in $\cT_c$. The context switches are nondeterministically guessed and there are two kinds of state jumps possible. 
First, a jump $(\vec{a}_2, \tilde{\gamma}_2, \vec{a}_3)$ in $\cC_{\gamma}$ corresponds to the thread being 
switched out while moving to global state $\vec{a}_2$ and later resuming at global state 
$\vec{a}_3$ with $\gamma_2$ on top of its stack (without being active in the interim). 
Second, a jump $(\vec{a}_2,\bot)$ corresponds to the last time the thread is swapped out during the execution of the associated task 
(leading to global state $\vec{a}_2$) or when the thread is terminated. Both of these kinds of jumps are made visible as part of the input alphabet by using the letters $\{0, 1, \bot\} \mathbin{\dot\cup} \Gamma$ in $\Sigma'$, one symbol at a time.
Additionally, in the case of an initialised task-PDA $\cC_{(\vec{a},\gamma)}$, its specified global state $\vec{a}$ is also made visible (in the same manner as the jumps) at the start of its execution.
We describe the details below.

The state of $\cC$ is of the form $\Box \vec{a}_1 \# \vec{b}_1$ where $\vec{a}_1$ is a string in $\{ 0,1\}^m$ representing the global variables, $\vec{b}_1$ is a string in $\{ 0,1\}^n$ representing the local variables and $\Box$ is a special symbol which is used to store information about whether the automaton is currently in the midst of outputting a jump transition or the initial global state.

\newcommand{\TransitionType}[1]{\noindent\textbf{#1}}
\TransitionType{Initialization.} We want an initialised task-PDA $\cC_{(\vec{a},\gamma)}$ to output the specified global state $\vec{a}$, using the transducers $T'_{0,\varepsilon}$ and $T'_{1,\varepsilon}$.
When seeing $\#$ as the first state symbol, either of these transducers will look for the first occurrence of its output symbol, $0$ or $1$, respectively, and replace it with the corresponding symbol $0'$ or $1'$, respectively. From the initial state $\# \vec{a} \# \vecz$ this eventually results in $\# \vec{a}' \# \vecz$, once all of $\vec{a}$ has been output. Here $\vec{a}'$ is the counterpart to $\vec{a}$, where every symbol in $\{0,1\}$ has been replaced with its primed version.
Then finally the transducer $T'_{\varepsilon,\varepsilon}$ looks for a state of the form $\# \{0',1'\}^m \# 0^n$ and changes each $0'$, respectively $1'$, back to $0$, respectively $1$, while also replacing the first $\#$ with $\dagger$. This results in the state $\dagger \vec{a} \# \vecz$.

\TransitionType{Non-jump Transitions.} For each $p=(\gamma,v)$ where $\gamma \in \Gamma \cup \{\varepsilon\}$, and $v \in \Gamma \cup \bar{\Gamma} \cup \{\varepsilon\}$, the transducer $T'_p$ of the task-$\PDA$ simulates $T_p \in \cT_c$ of the $\DCBPR$ on $\vec{a}_1\vec{b}_1$ while ignoring the remaining symbols in $\Box \vec{a}_1 \# \vec{b}_1$. When simulating a non-jump transition, the first symbol of the state is always $\dagger$.

\TransitionType{Context Switching.} Let $(\vec{a}_2,\gamma_2,\vec{a}_3)$ be a context switch, where the 
thread is swapped out when the global state is $\vec{a}_2$ with the top of stack being $\gamma_2$ and then swapped back 
in when global state is $\vec{a}_3$. Context switching takes place in three phases.

\begin{description}
	\item[Phase 1] Let us assume that a thread is swapped when moving from global state $\vec{a}_1$ to $\vec{a}_2$. This is done by application of a transducer from $\cT_s$ of the $\DCBPR$. While doing so, 
$\cC$ changes the value of the first state symbol from $\dagger$ to $\$_1$ to indicate that it is Phase 1 of a context switch. Transducers $T'_{0,\varepsilon}$ and $T'_{1,\varepsilon}$ are used to output the string $\vec{a}_2$. Since this is done one symbol at a time, the exact position is remembered by changing the binary string $\vec{a}_2$ into its corresponding primed version $\vec{a}'_2$. For example $110$ is converted to $1'10$, then $1'1'0$ and finally to $1'1'0'$. Once all bits have been output, the first state symbol is changed from $\$_1$ to $\$_2$ to signal the start of Phase 2.

	\item[Phase 2] The succinct $\PDA$ $\cC$ uses its transducer $T'_{\tilde{\gamma}_2,\bar{\gamma}_2}$ in order to output $\tilde{\gamma}_2$ while at the same time checking that $\gamma_2$ is indeed the top of the stack, by popping it from the stack. During this, the global state $\vec{a}_3$ is guessed as a resumption point for the context switch, and its primed version $\vec{a}'_3$ is saved in the state in place of $\vec{a}'_2$.
Additionally, the first state symbol is changed from $\$_2$ to $\$_3$ to move onto the next phase, and the $\#$ symbol between the global and local state is rewritten to $\gamma_2$ in order to momentarily store this stack symbol. The state at this point is of the form $\$_3 \vec{a}'_3 \gamma_2 \vec{b}_1$.

	\item[Phase 3] The global state $\vec{a}_3$ is read out in binary similar to Phase~1, except the state symbols now change from the primed version to their non-primed counterparts, i.e.\ replacing $\vec{a}'_3$ with $\vec{a}_3$ in the state. Then $\cC$ simulates the transducer $T_{\gamma_2} \in \cT_r$ of the $\DCBPR$ on global state $\vec{a}_3$ to essentially resume the thread. This is done as part of the transducer $T'_{\varepsilon,\gamma_2}$ in order to push $\gamma_2$ back onto the stack. Additionally, $\gamma_2$ in the state is changed back to $\#$ and the first symbol $\$_3$ is turned back to $\dagger$.
\end{description}

\TransitionType{Termination.} This can take place in two ways: either a thread is switched out and not switched back in again (case 1), or a transition rule from $T_t$ is applied (case 2).

\begin{description}
\item[Case 1] Here, $\cC$ proceeds as in Phase 1 of context switching but can nondeterministically choose to write $\bot_1$ on the first cell instead of $\$_1$ to indicate imminent termination. The rest of Phase 1 is carried out as before, transitioning from $\vec{a}_1$ to $\vec{a}_2$, writing out the bits of $\vec{a}_2$, and rewriting the first symbol of the state to $\bot_2$. However, instead of the steps in Phase 2, $\cC$ outputs $\bot$ and moves to its final state (the all $1$s string) on seeing $\bot_2$ as the first state symbol.

\item[Case 2] Here, $\cC$ also proceeds like in Phase 1 of context switching, but applies $T_t \in \cT_t$ of the $\DCBPR$ to transition from $\vec{a}_1$ to $\vec{a}_2$, instead of a transducer from $\cT_s$. During this, the first state symbol is also changed to $\bot_2$. Then $\cC$ outputs $\bot$ and moves to its final state as in Case 1.
\end{description}

\subsection{Alphabet-preserving Downward Closures}

For an alphabet $\Theta\subseteq\Sigma$, let $\pi_\Theta\colon\Sigma^*\to\Theta^*$ denote the projection onto $\Theta^*$. 
In other words, for $u\in\Sigma^*$, the word $\pi_\Theta(u)$ is obtained from $u$ by deleting all occurrences of letters in $\Sigma\setminus\Theta$.
Let $L\subseteq\Sigma^*$, let $\Theta\subseteq\Sigma$ be a subset, and let $k\in\N$. We define 
\[ \dcl[\Theta,k]{L}=\{u\in\Sigma^* \mid \exists v\in L\colon \text{$u\subw v$ and $\pi_\Theta(u)=\pi_\Theta(v)$ and $|v|_\Theta\le k$}\}. \]

The following theorem shows that the alphabet-preserving downward
closure has a doubly succinct representation.

\begin{theorem}\label{preserving-dcl}
  Given a succinct $\PDA$ $\cC$ over the input alphabet $\Sigma$, %
  a subset $\Theta \subseteq \Sigma$, and
  $K\in\N$ written in binary, one can construct in polynomial time a doubly succinct $\NFA$ $\cB$ 
  with $\langof{\cB}=\dcl[\Theta,K]{\langof{\cC}}$. 
\end{theorem}

\begin{remark}
  Recall that for the ordinary downward closure $\dcl{\langof{\cC}}$, 
  one can construct a singly succinct NFA for 
  $\dcl{\langof{\cC}}$. However, for the $\Theta$-preserving downward 
  closure $\dcl[\Theta,K]{\langof{\cC}}$, this is not true. Consider, 
  for example, the language 
  $L=\{w\rev{w} \mid w\in\{\ltr{a},\ltr{b}\}^*\}$ where $\rev{w}=w_n\ldots w_1$ is 
  the reverse of the word $w=w_1\ldots w_n$.  Then, for 
  $K=2^{n}$ and $\Theta=\{\ltr{a},\ltr{b}\}$, the set 
  $\dcl[\Theta,K]{L}$ consists of all palindromes of length 
  $2^n$. Clearly, an NFA for $\dcl[\Theta,K]{L}$ requires at least 
  $2^{2^{n-1}}$ states. 
\end{remark}

We begin with a simple lemma that states that succinct $\PDA$s can be modified to keep a count.
The proof of the lemma follows by modifying the transducers to maintain and increment the count of the occurrences
of $\Theta$ on the string.
 
\begin{lemma}\label{normal-to-succinct}
  Given a succinct $\PDA$ $\cC$ over the input alphabet $\Sigma$, 
  a subset $\Theta \subseteq \Sigma$, and a number $K$ 
  in binary, one can construct in polynomial time a succinct $\PDA$ $\cC'$ over the alphabet $\Sigma \cup \set{\#}$,
where $\#\not\in\Sigma$ is a new symbol, such that
  \[
L(\cC')=\{ w\#^l \mid w \in L(\cC), |w|_{\Theta}+l=K\}
\]
\end{lemma}
\begin{proof}
  Let $u_K \in \{0,1\}^{\lceil\log{K}\rceil}$ be the binary encoding of $K$. We obtain $\cC'$ from $\cC$ as follows.
  
  We extend the initial state $w_0$ of $\cC$ to $w_0\dagger u_K$ using a new separator symbol $\dagger$. All transducers not corresponding to input symbols in $\Theta$ are extended to not change the part of the state after $\dagger$. The transducers that do correspond to $\Theta$ implement binary subtraction by $1$ on this part of the state. We add a new transducer $T_\#$ for $\#$ that always checks whether the first part of the state matches the final state $w_f$ of $\cC$ without changing it. If this check succeeds, $T_\#$ also implements binary subtraction by $1$ on the part after $\dagger$. The final state of $\cC'$ is then defined as $w_f\dagger 0^{\lceil\log{K}\rceil}$.
\end{proof}

Given a language $L \subseteq \Sigma^*$ and a subset $\Theta \subseteq \Sigma$ of the alphabet, 
define 
\[ L|_{\Theta,k} := \{ w \mid w \in L, |w|_{\Theta}=k\}, \]
i.e., those words that contain exactly $k$ letters from $\Theta$. 
In the next step, we extract precisely these words from the downward closure. 

\begin{lemma}\label{theta-downclosure-extraction}
  Given a doubly succinct $\NFA$ $\cB=((T_a)_{a \in \Sigma_{\varepsilon}},\Sigma,\Delta,w_0,M)$, $\Theta \subseteq \Sigma$ and $k\in\N$ in binary, one can construct in polynomial time another doubly succinct $\NFA$ $\cB'$ such that $L(\cB')=L(\cB)|_{\Theta,k}$. 
\end{lemma}
\begin{proof}
$\cB'$ simply keeps a counter in its state and increments it whenever a letter from $\Theta$ is read while simulating $\cB$. 
When the final state of $\cB$ is reached, $\cB'$ can nondeterministically guess and subsequently check that the counter value is exactly $k$, moving to its own final state afterwards. Such a check is feasible by a transducer of $\cB'$ since the bit representation of $k$ is small.
\end{proof}

We are now ready to prove \cref{preserving-dcl}. 
\begin{proof}
  Suppose we are given a succinct PDA $\cC$ and a number $K$ written in binary. We build a succinct $\PDA$ $\cC'$ which extracts 
  those words of $L(\cC)$ that contain at most $K$ letters from $\Theta$ (padded with $\#$s)
  using Lemma \ref{normal-to-succinct}. 
  Then, using Theorem~\ref{th:succinct-dcl}, we construct a doubly
  succinct $\NFA$ $\cB$ such that $\langof{\cB}= \dcl{\langof{\cC'}}$.
    Then,
    applying Lemma \ref{theta-downclosure-extraction}, we get a doubly
    succinct $\NFA$ $\cB'$ such that
  \[ \langof{\cB'}=\{w \mid \text{$w\in\langof{\cB}$ and $|w|_{\Theta \cup \{\#\}}=K$} \}. \]
  Finally, modify $\cB'$ by changing all transitions reading the input letter $\#$ to read $\varepsilon$ instead. To this end we construct the doubly succinct $\NFA$ $\cB''$ from $\cB'$ by simply combining the two transducers $T'_\#$ and $T'_\varepsilon$ of $\cB'$ into a single transducer $T''_\varepsilon$.
  Observe that $\langof{\cB''}=\dcl[\Theta,K]{\langof{\cC}}$. 
\end{proof}
We can apply the theorem on any initialised task-PDA $\cC_{(\vec{a},\gamma)}$, giving us the following corollary. 

\begin{corollary}[Succinct Task Downclosure]
\label{coro:succinct-task-dcl}
Given a $\DCBPR$ $\cD$, a succinct task-$\PDA$ $\cC_{(\vec{a},\gamma)}$ and a number
$K$ in binary, there is a polynomial time procedure to compute a
(polynomial size) doubly succinct $\NFA$ $\cB$ such that
$\langof{\cB} = \dcl[\Theta, K]{\langof{\cC_{(\vec{a},\gamma)}}}$. Further, for a succinct task-$\PDA$ $\cC_{(\vec{a}',\gamma)}$ with $\vec{a}' \neq \vec{a}$, the corresponding $\dsNFA$ $\cB'$ only differs from $\cB$ in the initial prefix.
\end{corollary}

%% file: tcvass.tex
\section{From $\DCBPR$ to $\VASS$ Coverability} %
\label{sec:TCVASS}

\Cref{coro:succinct-task-dcl} allows us to represent the state of each thread as a succinct machine.
The next goal in the proof of \cref{th:complexity-thread-pooling} is to represent the entire configuration
of a $\DCBPR$.
Intuitively, the configuration will maintain  
the global state and the state of all the threads in the thread pool as ``control states''
and maintain a counter for each stack symbol that tracks the number of (unstarted) tasks of that type
that have been spawned.
The structure suggests the use of \emph{vector addition systems with states} ($\VASS$) as a representation.
The main challenge is to represent the states of exponentially many (in the representation of $N$) threads 
in a concise way.

\subsection{Vector Addition Systems with States and their Succinct Versions}

A \emph{vector addition system with states} ($\VASS$) is a tuple
$V = (Q, I, E,q_0,q_f)$ where
$Q$ is a finite set of \emph{states}, $I$ is a finite set of
\emph{counters}, 
$q_0\in Q$ is the initial state, 
$q_f\in Q$ is the final state, 
and
$E$ is a finite set of edges of the form $q \autstep[\delta] q'$ where
$\delta \in \set{-1,0,1}^I$.\footnote{A more general definition of $\VASS$
would allow each transition to add an arbitrary vector over the integers.
We instead restrict ourselves to the set $\set{-1,0,1}$, since this
suffices for our purposes, and the $\EXPSPACE$-hardness result by
\citet{lipton1976reachability} already holds for VASS of this form.}
A \emph{configuration} of the $\VASS$ is a pair $(q,u) \in Q \times \multiset{I}$.
The elements of $\multiset{I}$ and $\set{-1,0,1}^I$ can also be seen as vectors of length $|I|$ over $\N$ and $\set{-1,0,1}$, respectively, and we sometimes denote them as such.
The edges in $E$ induce a transition relation on configurations: 
there is a transition $(q, u)\autstep[\delta](q', u')$ if
there is an edge $q\autstep[\delta] q'$ in $E$ such that
$u'(p) = u(p) +\delta(p)$ for all $p\in I$.
A \emph{run} of the $\VASS$ is a finite or infinite sequence of
configurations $c_0 \autstep[\delta_0] c_1 \autstep[\delta_1]\ldots$ where $c_0=(q_0,\vecz)$.
A finite run is said to reach a state $q\in Q$ if the last configuration in the run is of the
form $(q,\textbf{m})$ for some multiset $\textbf{m}$.
An \emph{accepting} run is a finite run whose final configuration has state $q_f$.
The \emph{language (of the $\VASS$)} is defined as 
\[ \langof{V}=\{ w \in (\set{-1,0,1}^I)^* \mid w=w_1\cdots w_l, \text{ there is a run }
(q_0,\vecz)=c_0 \autstep[w_1] %
\ldots \autstep[w_l] c_l=(q_f,u)\}. \]
The size of the $\VASS$ $V$ is defined as $|V| = |I| \cdot |E|$. 

The \emph{coverability problem} for $\VASS$ is 
\begin{description}
  \item[Given] A $\VASS$ $V$. 
  \item[Question] Is there an accepting run? %
\end{description}
The coverability problem for $\VASS$ is $\EXPSPACE$-complete \cite{Rackoff78,lipton1976reachability}.

\subsubsection*{A succinct variant of $\VASS$} Let $I=\{ 1,2,\ldots,d\}$ for some $d \in \N$, let $\overline{I}$ be a disjoint copy of $I$ and define
$\hat{I}:= I \cup \overline{I}$ and 
$\hat{I}_{\varepsilon}:= \hat{I} \cup \{\varepsilon\}$.
A \emph{transducer controlled vector addition system with states} ($\TCVASS$, for short) is a tuple 
$\cV=((T_i)_{i \in \hatie},\Delta,M)$ where each $T_i$ is a 2-ary transducer over
 $\Delta \supseteq \{0,1\}$, and $M$ is a number in binary.
  This induces an explicit $\VASS$ 
 $\cE(\cV)=(Q,I,E,q_0,q_f)$ given by 
\begin{itemize}
  \item $Q= \Delta^M$,
  \item $E \subseteq Q \times \integs^d \times Q$ is the set of triples $(q,u,q')$ satisfying one of the 
  following conditions:
  \begin{enumerate}
    \item $(q,q') \in L(T_{\varepsilon}) \text{ and } u=\vecz$,
    \item there exists $i \in I$, $(q,q') \in L(T_i) \text{ and } u=e_i$ where 
    $e_i$ denotes the vector with $1$ in the $i^{th}$ component and $0$ otherwise, or
    \item there exists $\bari \in \barI$, $(q,q') \in L(T_{\bari}) \text{ and } u=-e_i$.
  \end{enumerate}
  \item $q_0=0^M$ and $q_f = 1^M$.
\end{itemize}
We write $q \xrightarrow{u} q'$ to denote a transition of $\cE(\cV)$ from its state $q$ to state $q'$ while
executing the instruction $u$.
If there exists a finite run $0^M \xrightarrow{u_1}q_1 \xrightarrow{u_1}q_2 \ldots \xrightarrow{u_l}1^M$, 
we say the word $u_1\ldots u_l$ is in the language of $\cV$, which is denoted $\langof{\cV}$.

Given a $\TCVASS$ $\cV=((T_i)_{i \in \hatie},\Delta,M)$, its \emph{associated} $\dsNFA$ is given by
$\cB(\cV)=((T_i)_{i \in \hatie},\hat{I},\Delta,M)$. In the sequel, we will associate $i$ (resp. $\bari,\varepsilon$ )
with $e_i$ (resp. $-e_i,\vecz$), so that the languages of $\cV$ and $\cB(\cV)$ are over the same alphabet $\hati$. The size of $\cV$ is simply defined as $|\cV| = |\cB(\cV)|$.

Our main result of this section is that, despite the succinctness, the coverability problem for $\TCVASS$ is $\EXPSPACE$-complete.

\begin{theorem}\label{thm:cover_tcvass}
The coverability problem for $\TCVASS$ is $\EXPSPACE$-complete.
\end{theorem}

\begin{proof}
The lower bound follows from the lower bound for $\VASS$.
To show the upper bound, given a $\TCVASS$ $\cV$, we show that we can construct in polynomial time a $\VASS$ $\cV'$ 
such that the final state is reachable in $\TCVASS$ if and only if the final state is reachable in $\cV'$.

We begin with a language-theoretic observation.
Let the $\VASS$ $V_I$ be defined as $V_I=(\{q_0 \},I,E,q_0,q_0)$ where for each $i \in I$ 
(resp. $i \in \overline{I}$)
   there is a transition $q_0 \autstep[e_i] q_0$ (resp. $q_0 \autstep[-e_i] q_0$); 
   as well as a transition $q_0 \autstep[\vecz] q_0$.
   For any $\TCVASS$ $\cV$, note that 
   \[ \langof{\cV}=\langof{\cB(\cV)} \cap \langof{V_I}. \]
To prove the theorem, it suffices to construct a $\VASS$ $\cV''$ with
$\langof{\cV''} = \langof{\cB(\cV)}$: This would enable us to build a $\VASS$ 
for $\langof{\cV''} \cap \langof{V_I}$. The $\VASS$ $\cV''$ can be constructed due to two facts:
\begin{enumerate}
  \item A classical result that an exponentially bounded automaton can be simulated by a 
  polynomial size $\VASS$, which can be constructed in polynomial time (see \cite{Esp98a,lipton1976reachability}) and
  \item the sequence of reductions used in the above result having a property we call \emph{local simulation}.
  For machines $\cM_1,\cM_2$ we say $\cM_1$ is simulated by $\cM_2$ locally if every step of $\cM_1$ is simulated
   by some fixed sequence of steps
of $\cM_2$ which only depends on a constant sized local part of the global configuration.
By associating the last of the corresponding sequence of steps of $\cM_2$ with the action of 
the single step of $\cM_1$ and the rest of the steps of $\cM_2$ with $\varepsilon$, we ensure that  
$\langof{\cM_2}=\langof{\cM_1}$. \qedhere
\end{enumerate}

\end{proof}

\begin{remark}
An alternate proof that coverability of $\TCVASS$ is in $\EXPSPACE$ follows from the multi-parameter analysis of decision problems for $\VASS$ by
\citet{RosierYen86}, where they show that the Rackoff argument can be modified to show that the coverability problem
for a $\VASS$ with $k$ counters, $n$ states,
and constants bounded by $l$ can be solved in $O((l+\log n)2^{ck\log k})$ nondeterministic space.
For $\TCVASS$, the number of states $n$ is doubly exponential in the description of the $\TCVASS$, so this result shows
that coverability is in $\EXPSPACE$.
We believe our argument is conceptually simpler.
\end{remark}

\subsection{From $\DCBPR$ to $\TCVASS$}

\begin{theorem}\label{dcps-to-tcvass}
  Given a $\DCBPR$ $\cD=(V_{\glob},V_{\loca},\Gamma,\cT_c,\cT_s,\cT_r,\cT_t,\vec{a}_0,\gamma_0)$, its global state
  $\vec{a}_f$ and two binary numbers $K,N$, we can construct in
  polynomial
  time a $\TCVASS$ $\cV=((T'_i)_{i \in \hatie},\Delta',M')$ such that
   $\vec{a}_f$ is  reachable in $\cD$ via an
   $N$-thread-pooled, $K$-context switch bounded execution if and only 
  if the final state of $\cE(\cV)$ is reachable in $\cE(\cV)$.
\end{theorem}

The proof of the theorem requires us to simulate up to $N$ different
threads of the $\DCBPR$. To this end we compute $\dsNFA$s for each
type of task these threads could be executing, using \cref{coro:succinct-task-dcl}. The simulation then needs to keep track of up to $N$ different $\dsNFA$ states at the same time. Thus, we construct $\cV$ in such a way that its state is comprised of exponentially many segments (since $N$ was given in binary) each storing a $\dsNFA$ state, separated by special symbols. Initializing this segmentation is formalized in the following auxiliary result.

\begin{lemma}\label{lem:dsNFA-create-separators}
  Given numbers $M,M'$ in binary and $m$ in unary such that $M'=M \cdot M''$,
  there exists
  a $\dsNFA$ $\cB_{m,M,M'}$ over the empty alphabet,  polynomial in
  the size of $M'$ such that it can reach a
  state $s$ of the form $s = \$^m(0^{M} \#)^{M''}$ from an initial state
  $0^{m+(M+1)M''}$. Moreover, $s$ is the unique
  state in the set of reachable states of $\cB_{m,M,M'}$ which ends
  with the symbol $\#$. 
\end{lemma}
\begin{proof} 
Let us consider for simplicity the case when $m=0$.
We can think of the state of a $\dsNFA$ $\cB$ as the
 tape of a Turing machine $\cM$. While a transducer only reads its
 tape from left to right once, a Turing machine can read a particular
 cell and move either left or right. However, the transducer can
 simulate two-way movement as follows: \\ Let $\cM$ operate over
 alphabet $\Sigma$, then $\cB$ operates over alphabet
 $\Sigma \cup \tilde{\Sigma}$ where $\tilde{\Sigma}$ is a disjoint copy
 of $\Sigma$. The disjoint copy $\tilde{\Sigma}$ is used by $\cB$ to
 keep track of the position of the head. For example, if the tape
 contents of $\cM$ are $w_1abw_2$ for $w_1,w_2 \in \Sigma^*$ and
 $a,b \in \Sigma$ with the tape head on $b$, then the tape contents
 of $\cB$ are $w_1 a \tilde{b}w_2$. When $\cM$ executes a transition
 where $b$ is replaced by $c$ and the head moves left, $\cB$
 guesses, at the time of reading $a$, that the tape head is on the
 next symbol and thus replaces $a$ by $\tilde{a}$ and then replaces
 $\tilde{b}$ by $c$. (In case the transition simulated is one where
 the head is moved right, this guess is not required.) Thus in the
 rest of the proof, we will describe the actions of the required
 $\dsNFA$ as if it were a Turing machine.

  Using a constant number $n_1$ of states, one can implement doubling, i.e. 
   an input of the form $1^n0^{n'}$ can be converted to $1^{2n}0^{n'-n}$.
   The machine turns the first $1$ into a decorated copy $\overline{1}$, then moves to
   the right till it finds the first 0, converting it into a $\overline{1}$. It then moves
   all the way back to the beginning of the tape, making sure that it encounters a $1$. After $i$ such moves,
   the state is of the form $\overline{1}^i1^{n-i}\overline{1}^i0^{n'-i}$. Finally we reach a point
   when no more $1$s are present between the two blocks of $\overline{1}$s i.e. the state is 
   $\overline{1}^{2n}0^{n'-n}$. We now convert all $\overline{1}$s to $1$s.

   Let $N=(M+1)M''$. We initialize by converting $0^{N}$ to $10^{N-1}$. We then implement a counter
   which repeatedly implements the doubling module. In order to create a string $1^{2^k}0^{N-2^k}$,
   we need a counter which uses $k+1$ bits. Thus the counter can be implemented using $n_2$ many 
   states where $n_2$ is a number that is linear in $k$, by first creating a string $u=10^k$ and comparing
   the value of the counter with $u$ after each increment. For an arbitrary number 
   $M=\sum_{i=0}^\ell 2^{k_i}$,
   we can separately implement the above module for each power $2^{k_i}$. This takes roughly $n_2^2$
   many states. We can now create the string $0^{M} \#0^{M'-M-1}$. 

   As the last step to obtain the required string $(0^M\#)^{M''}$, we create another counter which counts up to $M''$ using $n_3$ states, where $n_3$ is polynomial in the
   bit size of $M''$. Using this, we can repeat the above process $M''$ times to get the string $(0^M\#)^{M''}$. We have used
   a total of $n_1+n_2^2+n_3$ states which is polynomial in the bit size of $M'$. Note that the last symbol remains a $0$
   until the very end when it is changed into a $\#$. Hence we also satisfy the condition that the target
   state is the only reachable state with a $\#$ as the last symbol. 

   Let us now consider the case that $m>0$. Then, the $\dsNFA$ can simply count up to
   $m$ using the states of the transducers since $m$ is given in
   unary. It converts the string $0^{m+(M+1)M''}$ to $\$^m0^{
   (M+1)M''}$ first and then runs the above procedure on the suffix $0^{(M+1)M''}$.
\end{proof}

\begin{proof}[Proof of \Cref{dcps-to-tcvass}]
  Application of Corollary~\ref{coro:succinct-task-dcl} to the
  initialised succinct task-$\PDA$
  $\cC_{(\vec{a},\gamma)}$ 
  gives us a $\dsNFA$ 
  $\cB_{(\vec{a},\gamma)}=((T_{\sigma,\gamma})_{\sigma \in
  \Sigma_{\varepsilon}},\Sigma,\Delta,w_0(\vec{a}),M)$, where
  each $T_{\sigma,\gamma}$ has state set $Q_{\sigma,\gamma}$. We call $\cB_{(\vec{a},\gamma)}$ a task-$\dsNFA$. 
   Note that for initialised succinct task-$\PDA$
  $\cC_{(\vec{a}',\gamma)}$, the corresponding $\cB_{(
  \vec{a}',\gamma)}$ only differs from $\cB_{(\vec{a},\gamma)}$ in its initial prefix $w_0(
  \vec{a}')$. The tuple $
  (\vec{a},\gamma)$ is called the \emph{type} of a task. 
   Since we need a counter for each stack symbol (to track the unstarted tasks), we define $I = \{1, \ldots, d\}$, where $d = |\Gamma|$. We also write $\Gamma=\{ \gamma_1,\gamma_2,\ldots,\gamma_d\}$.

  The construction of $\cV$ essentially involves running $N$ different 
  task-$\dsNFA$s in parallel while keeping track of the tasks using its counters. 
  The local information of the $N$ different threads are stored in the
  global state (henceforth also called the \emph{tape}) of $\cV$,  with separator
  symbols between consecutive threads. We call the space between two
  consecutive separators a \emph{segment}. This means that the tape
  contents are always of the form
  $u\Box_1w_1\Box_2\ldots\Box_Nw_N\Box_{N+1}$
  where $u$ is a string of length $m$ that is used to store the global
  state of the $\DCBPR$ at schedule points to help coordinate the
  communication between the $N$ different threads, the $\Box_i$ are
  \emph{separator} symbols which indicate whether a segment is
  currently
  \emph{occupied} by a thread (in which case the string $w_i$
  contains the configuration of the thread) or \emph{unoccupied} (in
  which case the segment is filled with $0$s).
  We briefly outline the salient features of $\cV=((T'_i)_{i \in \hatie},\Delta',M')$ below: 
\begin{itemize}
  \item $\Delta'=\Delta \cup \{ \$_1,\#_1,\dagger_1,\tilde\dagger \} \cup \Gamma
  \cup \{0,1,\bot\}$, where
  \begin{itemize}
    \item  $\#_1,\dagger_1$ are separators which are used to
    indicate occupied and unoccupied segments respectively.
    \item when
    the $\DCBPR$ currently has an active thread then $u = \$_1^m$;
    otherwise it is at a schedule point and $u=\vec{a}$ for some $
    \vec{a} \in \{0,1
    \}^m$,
    \item the symbols in $\Gamma \cup \{0,1\}$ are used to mark each
    thread with its state and top of stack in order to facilitate
    context switches,
    \item $\tilde\dagger$ is a separator used to mark a segment being initialised,
    \item $\bot$ is used to mark a thread meant to be
    terminated.
  \end{itemize}
  \item $M'=m+1+N(M+2m+3)$,
  \item The transducers $T'_i$ consist of different modules as outlined below.
  \begin{enumerate}
    \item For each $i \in I$, $T'_i$ has a single module which is used when a step performed
    by a particular thread spawns a new task $\gamma_i$. 
    \item  For each $\bari \in \barI$, $T'_{\bari}$ has a single module which used to 
    initialize a task of type $\gamma_i$.
    \item $T'_{\varepsilon}$ has the following modules:
    \begin{enumerate}
      \item\label{mod:init} $\init$, which creates the initial set of separators between segments,
      \item\label{mod:start} $\start$, which helps initialize a task-$\dsNFA$ correctly,
      \item\label{mod:inter} $\inter$, which interrupts a thread,
      \item\label{mod:resump} $\resump$, which resumes a thread,
      \item\label{mod:term} $\term$, which handles termination of threads,
      \item\label{mod:clean} $\clean$, which cleans up a segment after termination of a thread, 
      \item\label{mod:final} $\final$, which detects when the $\DCBPR$ has reached
      global state $\vec{a}_f$, and
      \item\label{mod:emp} $\emp$, which handles transitions on input $\varepsilon$.
    \end{enumerate}
  \end{enumerate}
\end{itemize}
The different modules in $T'_{\varepsilon}$ are nondeterministically chosen to be run. In case 
   the running of a module would cause a problem, this is prevented by appropriate flags in the global
   state on which
   no transitions are enabled in that particular module. We explain the working of $\cV'$ in detail below.

  \noindent\textbf{Initialization:}
  $\cV$ initially moves to a state of the form $\$_1^m\dagger_1(0^
  {M+2m+2}\dagger_1)^N$ where $\#_1,\$_1,\dagger_1$ are
   symbols not present in $\Delta$
  (this is feasible by Lemma \ref{lem:dsNFA-create-separators}). We
  then simulate the $\dsNFA$ $\cB_{(\vec{a}_0,\gamma_0)}$
  in the first segment and the state is of the form $\$_1^m\#_1
  \vec{a}_0\gamma_00^m10^M \dagger_1(0^{M+2m+2}\dagger_1)^{N-1}$. We
  explain the contents of a segment next.
 
  The first $2m+1$ letters in a segment are used to store context switch
  information: $m$
  letters for the global variables followed by the top of stack
  symbol, followed by $m$ more letters. The global state of the thread
  can be inferred from this initial segment. The $(2m+2)$nd letter is a
   $1$ or a $0$ based on whether the
  segment contains the active thread or not. The final $M$
  many
  symbols are used to store the state of
  the $\dsNFA$. 
Immediately after initialization, only
  the first segment contains a thread, which is active and starts
  running. Thus the separator symbol before the first segment at this
  point of time is $\#$. In contrast, the other (unoccupied) segments
  are preceded by a $\dagger$. These operations are handled by module $\init$ from (\ref{mod:init}).

  \noindent\textbf{Spawning:} The spawning of a new thread 
  $\gamma_i$ by a thread of type $(\vec{a},\gamma)$ corresponds to a 
  run of a transducer $T_{\gamma_i,\gamma}$ in a task-$\dsNFA$
  and is
   simulated by $T'_{i}$ in $\cV$. Steps in the $\DCBPR$ corresponding
   where the active thread does not spawn anything are handled by module $\emp$ from (\ref{mod:emp}).
   Having spawned some number of threads,  the active thread is 
   switched out for another.

  \noindent\textbf{Context Switching:}
  In order to understand the simulation of the context switch
  behaviour, recall that the $\dsNFA$ $\cB_{\vec{a},\gamma}$ which
  represents the succinct alphabet preserving downclosure has 
  input alphabet $\Sigma=\Gamma \cup \tilde{\Gamma} \cup 
  \{0,1,\bot\}$. Of
  these,
  $\Gamma$ is used to indicate spawned tasks, the initial state is
  indicated by a string
  $\vec{a}_0 \in \set{0,1
  }^m$,  context switches are
  indicated by strings of the form $\vec{a}_2\tilde{\gamma}_2
  \vec{a}_3$ where $
  \vec{a}_2,\vec{a}_3 \in \set{0,1}^m$ and $\tilde{\gamma} \in 
  \tilde{\Gamma}$; and termination is indicated by strings of the form
  $\vec{a}_2 \bot$. 

The module $\inter$ in (\ref{mod:inter}) contains the transducers $T_
{\sigma,\gamma}$
for $\sigma \in \tilde{\Gamma} \cup \set{0,1}$, for each $\gamma
\in \Gamma$. In order to simulate a context switch, $T'_{\varepsilon}$ can nondeterministically
choose to run the module $\inter$. Suppose the context switch is
represented by the string
$\vec{a}_2 \tilde\gamma_2 \vec{a}_3$. Since this string is output by the transducers,
it can be copied onto the first $2m+1$ cells of the segment $s$ containing the active thread.
 On running $\inter$, the global state is stored by rewriting $\$_1^m$
 to $\vec{a}_2$, indicating
 that the $\DCBPR$ is moving
 to a schedule
   point. The segment of the active thread has its $(2m+2)$nd symbol rewritten to $0$. At this point, all segments
  have $0$ as their $(2m+2)$nd symbol, indicating that there is no
  active thread. 

 Resumption of a different thread at the schedule point is executed as follows by
  the module $\resump$ in (\ref{mod:resump}). Two options are available at this point. The first option is to 
  wake up
  some inactive thread. The state $\vec{a}_2$ stored in the first $m$ cells
   is used to wake up a thread in segment $s'$ by comparing symbol by symbol the global state
   stored in cells $m+2$ to $2m+2$ of $s'$. Segment $s'$ is activated by
   writing $1$ on the $(2m+2)$nd symbol in the segment and the first $m$ cells are 
   again replaced with $\$_1^m$.
  
The second option is to invoke a $T'_{\bari}$ transducer, starting a new
  task in one of the unoccupied segments $s$. The transducer rewrites the separator symbol in front 
  of the segment $s$ to $\tilde\dagger$ and writes $\tilde\gamma_i$ on the $(m+1)$st and $1$ on the $(2m+2)$nd cells of 
  the segment. The $\start$ module in (\ref{mod:start}) then copies the global state $\vec{a}_2$
  present on the first $m$ cells of the tape onto the first $m$ 
  cells of segment $s$ and writes the corresponding initial prefix $w_0(\vec{a}_2)$ starting from the
  $(2m+3)$rd cell. At the end of this operation, the separating symbol $\tilde\dagger$ in front of
  segment $s$ is rewritten to $\#_1$, with the active thread contained in segment $s$ at this point.

  \noindent\textbf{Termination:}
  The termination of a thread is handled by module $\term$ in (\ref{mod:term}),
  which contains transducers $T_{\sigma,\gamma}$
   for  $\sigma \in \set{0,1,\bot}$. 
  Whenever a string $\vec{a}_3 \bot$ is output by the transducers, this string is copied onto
  the first $m+1$ cells of the segment containing the active thread.
  This signals that the task-$\dsNFA$ occupying this segment has terminated
  and now the $\clean$ module
  in (\ref{mod:clean}) is the only module
  allowed to run. This $\clean$ module replaces any segment containing $\bot$ with a string of $0$s and also
  changes the symbol prior to the segment from $\#_1$ to $\dagger_1$
  to indicate that the segment is now unoccupied. \\
  \noindent\textbf{Detecting $\vec{a}_f$ reachability of the
  $\DCBPR$:}
  We note that we can modify the $\DCBPR$ so that if there is a run reaching
  $\vec{a}_f$, then there is one where $\vec{a}_f$ occurs at a schedule point.
  To this end we add additional resumption rules whereby a context
  switch is allowed with any top of stack $\gamma \in \Gamma$ when the global
  state is $\vec{a}_f$.
  Thus it suffices to check if $\vec{a}_f$ occurs in the first $m$ cells,
  at which point $\cV$ moves to its final state. This is implemented by the module $\final$ 
   in (\ref{mod:final}).
\end{proof}

We can now prove Theorem~\ref{th:complexity-thread-pooling}.
Given an input $\DCBPR$ $\cD$, its global state $\vec{a}$ and numbers
$K,N$ in binary, we construct task-$\dsNFA$s
 $\cB_{\vec{a}',\gamma}$ to represent
the alphabet preserving downclosure $\dcl[\Theta, m+(2m+1)K+(m+1)]{
\langof{\cP_
{\vec{a}',\gamma}(\cD)}}$ of each task-$\PDA$ 
$\cC_{\vec{a}',\gamma}$ using \cref{coro:succinct-task-dcl}.
Note that the parameter $m+(2m+1)K+(m+1)$ occurs because in the language of an initialized task $\PDA$, the initial assignment to the global variables is represented by a string of length $m$, each context switch is represented by a string of length $2m+1$, and the termination is represented by a string of length $m+1$.
There are only polynomially many transducers corresponding to different $\dsNFA$ which need to be stored, one set for each
$\gamma$, since the assignment $\vec{a}$ to the global Boolean variables only affects the initial prefix. 
These transducers are then used in Theorem \ref{dcps-to-tcvass} to reduce the thread-pooled safety problem to coverability of $\TCVASS$.
Finally, we use Theorem \ref{thm:cover_tcvass} which gives us an $\EXPSPACE$ procedure for coverability of $\TCVASS$.

%% file: binarybound.tex
\section{Context-bounded Reachability is $\THREEEXPSPACE$-hard} %
\label{sec:context_bound_as_binary_input}

To put our results on thread-pooling in perspective, we show that we obtain $\THREEEXPSPACE$-completeness if we remove all restrictions on the thread-pool, 
as formalized by Theorem~\ref{th:complexity-no-thread-pooling}. 
We already mentioned that the upper bound of this theorem follows directly from \citet{AtigBQ2009}, and 
that for the lower bound we extend the $\TWOEXPSPACE$-hardness result of \citet{BaumannMajumdarThinniyamZetzsche2020a} by using succinctly encoded 
models in the intermediate proof steps.

The previous hardness result uses a more restrictive model called $\DCPS$. Compared to $\DCBPR$ it has no local Boolean variables, and uses a set of global states instead of global Boolean variables. This can be easily simulated by $\DCBPR$ with (i)~an empty set of local variables and (ii)~using global variables to store the global state. Then we ensure that there is always exactly one global variable valued $1$, corresponding to the current global state. Due to this simulation it suffices to only consider the model of $\DCPS$ for the lower bound, which we will do for the remainder of this section.

To also give a bit more detail on the upper bound, going from $\DCPS$ with unary context-switch bound $k$ to $\DCBPR$ with binary $k$ incurs two exponential blow-ups, one for Boolean variables instead of states and one for the encoding of $k$. However, both these blow-ups occur multiplicatively with respect to one another, meaning the combined blow-up is still only singly exponential. Therefore the $\TWOEXPSPACE$-procedure by \citet{AtigBQ2009} becomes a $\THREEEXPSPACE$-procedure in our setting.

We proceed by proving the $\THREEEXPSPACE$ lower bound. The full proof is technical, and provided in the full version. %
We sketch the main challenges in this section.

\subsubsection*{Previous result}

Let us begin by recalling some details in the proof of the $\TWOEXPSPACE$ lower bound of \citet{BaumannMajumdarThinniyamZetzsche2020a}.
The hardness proof goes through a series of reductions; 
the main step is a reduction from the coverability problem for \emph{transducer-defined Petri nets} ($\TDPN$) to context-bounded reachability
in $\DCPS$ for a unary context bound.
 
Recall that Petri nets are essentially a variant of $\VASS$ with just one state. 
$\TDPN$ are succinct representations of Petri nets.
Similarly to $\TCVASS$, the (exponentially many) counters of a $\TDPN$ are encoded by strings of polynomial length over an alphabet, and transitions 
are described by three transducers that encode the transitions between counters. 
Three transducers are needed because $\TDPN$ only consider three types of transitions: 
\emph{Fork} transitions that subtract $1$ from one counter and add $1$ to each of two counters, 
\emph{join} transitions that subtract $1$ from two counters and add $1$ to one counter, and 
\emph{move} transitions that add $1$ to and subtract $1$ from a single counter each.
It is known that any Petri net can be brought into a ``normal form'' in polynomial time, where each transition has one of these forms.
Like for $\VASS$, a configuration of a Petri net is a map from its counters to $\N$.
In lieu of more than one state, the coverability problem for Petri nets
takes a counter as input and asks whether a configuration with a nonzero value for that counter is reachable. 
Petri nets have an $\EXPSPACE$-complete coverability problem. 
\citet{BaumannMajumdarThinniyamZetzsche2020a} showed $\TWOEXPSPACE$-hardness for the coverability problem for $\TDPN$. 

The succinctness in $\TDPN$ is different from that in $\TCVASS$: in a $\TCVASS$, the control state is encoded succinctly but the dimension (the number of counters)
is maintained explicitly. 
In contrast, a $\TDPN$ has a dimension that is exponentially larger than the description.  

The main step of the lower bound reduces the coverability problem for $\TDPN$ to context-bounded reachability for $\DCPS$ 
by representing the explicit Petri net configurations as tasks. (Since there is no thread pooling, each task runs in its own thread,
simultaneously with all other spawned and partially executed tasks.)
Each counter $p$ of the $\TDPN$, which is encoded as a string of length $n$, is assigned a number $u(p)$ by such a configuration. 
To encode this configuration, the $\DCPS$ maintains exactly $u(p)$ many tasks that each have $p$ as their stack content. 
Transitions are then simulated by (i)~emptying stacks of tasks whose corresponding counters were decremented, 
and (ii)~spawning new tasks whose stacks are built up with corresponding counters that are incremented. 
During this, $n$ many new tasks are spawned to hold the information of the string encoding the corresponding counter. 
These $2n$ or $3n$ tasks in the task buffer are then used to verify that the corresponding transition existed, 
with the three transducers being kept as part of the global states of the $\DCPS$. 
Coverability for Petri nets then corresponds to reachability of a $\DCPS$ configuration that includes a task 
corresponding to a specific counter $p$. 
The $\DCPS$ can just check for stack content $p$ using $n$ global states, and 
then move to a specific global state $g$, completing the reduction to context-bounded reachability.

\subsubsection*{Challenges}

To lift the previous result to $\THREEEXPSPACE$ we introduce the model of \emph{succinct} $\TDPN$ ($\sTDPN$), 
which use exponentially long strings for their encoding, as opposed to polynomially long ones. 
In our reduction to $\DCPS$, we represent a configuration of (the underlying Petri net of) a $\sTDPN$ in the same way as before, 
the only difference being that stack contents are now of exponential length. 
This causes some challenges not present for the previous proof: 
Firstly, we cannot store the information of a single string of length $m$ via $m$ tasks in the task buffer. 
To be able to unambiguously reconstruct the string, each such task would need to remember the position of its letter, 
$1$ to $m$. 
Since $m$ is now exponential, and each such task just carries a single symbol, this would require an exponentially large alphabet. 
Secondly, we now also need to create a stack content of exponential length to initialize the whole simulation. 
For polynomial length $n$ this could just be done using $n$ many global states. 
Finally, the check for a specific counter $p$ at the end is also more problematic, since it again needs to check an exponentially
long string letter-by-letter.

However, since $K$ is provided in binary, in contrast to the previous proof, we can now make exponentially many context switches per task. 
This allows our proof to overcome the above challenges as follows. 

\subsubsection*{Verifying Transitions}

Since we have exponentially many context switches, our simulation of a single transition
can switch back and forth between the two or three tasks involved in a single transition of the $\sTDPN$, 
removing or adding one stack symbol each time. 
For each pair or triple of such symbols we advance the state of the transducer to verify the existence of this transition. 
Because the verification now happens at the same time as the simulation, we do not need to store additional information for it in the task buffer. 
The problem here is that the stacks that are built-up in this way are in \emph{reverse} order 
compared to the stacks that were emptied.
This presents another challenge: in order to continue the simulation, we have to reverse an exponentially long stack.

\subsubsection*{Reversing Stacks}

Another use case of exponentially many context switches is the transfer of an exponential length stack content from one task to another. 
This is done iteratively by popping a symbol from the first task, remembering it in the global state, 
switching to the second task, and then pushing the symbol before switching back. 
Doing so reverses the order of symbols on the stack, which is exactly what we need to remedy the issue mentioned above. 
If our stacks have exponential length $m$ then this process requires $m$ iterations and therefore $m$ context switches. 
With the simulation of a transition acting on the task beforehand it results in $2m$ context switches in total per task, 
which is still exponential in the binary context bound $K$, and therefore can be performed in the reduction.

\subsubsection*{Initialization and Finalization}

Starting the simulation was ``easy'' for $\TDPN$: we could explicitly add the first stack of polynomial size using the global state.
This is no longer possible as the stack is now exponential.
 
To create the exact exponentially long encoding of a specific counter $p_0$ at the start of the simulation, we do not encode it in the global states. 
Instead we simulate a context free grammar on the initial thread's stack. 
One can construct such a grammar of polynomial size that produces just a single exponentially long word, 
which in this case is $p_0$ (the grammar is a straight line program that performs the ``iterated doubling'' trick to accept exactly one word of
exponential size). 
Due to the grammar's size, we only need to add polynomially many stack symbols to facilitate its simulation.

For the final part of the simulation, where we check for a specific counter $p$, we could also try and utilize a context free grammar. However, due to some technical details in our definitions for $\sTDPN$, this is not even necessary: We only allow a certain shape for the $p$ used as input to coverability, and this shape can be easily described using a regular expression. Therefore we can check for the correct shape by using a finite automaton encoded in the global states of the $\DCPS$.

%% file: discussion.tex
\section{Discussion}
\label{sec:discussion}

We have characterized the complexity of safety verification in the presence of thread pooling.
Surprisingly, thread pooling \emph{reduces} the complexity of verification by a double exponential
amount, even when all parameters are in binary.
Along the way, we have introduced succinct representations and manipulations of succinct machines
that may be of independent interest.

While we have focused on safety verification, we can consider liveness properties as well.
The \emph{thread-pooled context-bounded termination} problem asks, given a $\DCBPR$ $\cD$ and two numbers $K$ and $N$
in binary, does every $N$-thread-pooled, $K$-context switch bounded run terminate in a finite number of steps? 
Our constructions imply that checking termination is $\EXPSPACE$-complete.
The upper bound proceeds as with reachability, but we check termination rather than coverability for the constructed $\VASS$.
The lower bound follows from a ``standard trick'' of reducing reachability of a global state to non-termination: the program initially guesses the
number of steps to reach a global state (by spawning that many threads) and entering an infinite loop iff the global state is reached.
When $N = \infty$, the problem becomes $\THREEEXPSPACE$-complete, following the same ideas as for safety verification.

An infinite run is \emph{fair} if, intuitively, the scheduler eventually picks an instance of every spawned task
and any thread that can be executed is eventually executed. 
The thread-pooled, context-bounded \emph{fair termination} problem 
asks if there is a fair non-terminating run.
\citet{BaumannMajumdarThinniyamZetzsche2021a} study the problem in the case of explicitly specified states (no Boolean variables) without thread pooling and reduce it to $\VASS$ reachability. Here, the reduction can be performed in elementary time.
By turning Boolean variable assignments into explicit states and keeping a counter of threads in the thread pool (in the global state), one can easily turn a $\DCBPR$ into an exponentially large $\DCPS$.
Therefore, fair termination of $\DCBPR$ also admits an elementary-time reduction to $\VASS$ reachability.
On the other hand, $\VASS$ reachability can be reduced in polynomial time to fair non-termination of the special case ($K=0$, $N=1$)
of asynchronous programs \cite{GantyM12}.
Together with complexity results on $\VASS$-reachability \cite{LerouxSchmitz2019, CzerwinskiPNAckermann2021, LerouxPNAckermann2021}, we conclude the problem is Ackermann-complete with or without thread pooling.

In conclusion, we find it surprising that thread pooling results in a significant reduction in the theoretical complexity of safety verification.
Whether this observation has practical implications remains to be explored.

%% file: appendix-succinctPDA.tex
\section{Proof Details for Succinct PDA and Succinct Downward Closures} %
\label{sec:appendix_succinct_pda}

This section contains additional details for various proofs found in Section~\ref{sec:bounded_task_buffer}.

\subsection{Proof Details for Lemma~\ref{bounded-stack-dcl}}
Let us see in detail how the construction from \cite{MajumdarTZ21} can be applied when we start with a succinct 
  $\PDA$ $\cC=((T_c)_{c \in \scrC},\Sigma,\Gamma,\Delta, w_0,w_f)$. 
  We assume that $\{ 0,1\} \subseteq \Delta$. 
  We shall construct in polynomial time, a succinct $\PDA$ $\cC'=((T'_c)_{c \in \scrC'},\Sigma,\Gamma',\Delta', w'_0,w'_f)$ 
  and a number $h$, whose bit representation is polynomial in the size of $\cC$, such that $\langof[h]{\cC'}=\dcl{\langof{\cC}}$, where
  \begin{itemize}
    \item $\Gamma'=\Gamma \cup \{ [,],\#,\dagger,\$\}$,
    \item $\Delta'=\Delta'' \times \Delta''$ where $\Delta''=\Delta \cup \{ \# \}$,
    \item $w'_0=(1w_0\# 0^{n''},0^{3+4|w_0|})$ and $w'_f=(1w_f\# 0^{ n''},0^{3+4|w_0|})$ where $n''=3|w_0|+1$,
  \end{itemize}
  and we describe the transducers of $\cC'$ below. We will also reason below why a length of $3+4|w_0|$ is necessary for the states.
  \begin{enumerate}[(i)]
    \item The idea is to consider the state of $\cC'$ as two separate tapes on which polynomial space Turing machine 
    computations can be simulated. Clearly a single transition of a fixed Turing machine can be simulated by 
    a nondeterministic guess of the transducer. The two tapes can alternatively be seen as a tape alphabet 
    which is a cross product $\Delta'' \times \Delta''$. The first tape is used to simulate the transitions of
    $\cC$---let us called these the ``type 1'' edges---while the second tape is used to simulate transitions from the new edges 
    introduced in Equation (\ref{newedges})---we call these the ``type 2'' edges.
    Note that type 1 edges also include remaining edges of the augmented automaton from \cite{MajumdarTZ21}, so those defined outside of Equation (\ref{newedges}) (the remaining edges were the ones where we dropped the input letters, replacing them with $\varepsilon$).
    \item With the state being a word $w$ over $\Delta'$, it starts with either $\begin{pmatrix} 1 \\ 0 \end{pmatrix} \in \Delta'' \times \Delta''$ or with $\begin{pmatrix} 0 \\ 1 \end{pmatrix}$. 
    The first case indicates that the first tape is \emph{active} while the second indicates that the second tape is active. 
    The first (resp.\ second) tape is active when $\cC'$ makes a type 1 (resp. type 2) transition. 
    The automaton $\cC'$ nondeterminisically chooses which type of transition to make at any given point. 
    In general, we write $\begin{pmatrix} w \\ w' \end{pmatrix}$ to indicate that the first tape contains $w$ and the second $w'$, as its contents.
    \item From a state $\begin{pmatrix} 1w_1\# 0^{ n''} \\ 0^{3+4|w_0|} \end{pmatrix}$ of $\cC'$, a type 1 edge 
    taking $\cC$ from $w_1$ to $w_2$ corresponds to changing the state of $\cC'$ to $\begin{pmatrix} 1w_2\# 0^{ n''} \\ 0^{3+4|w_0|} \end{pmatrix}$. 
    We ensure that the state is always of the form $\begin{pmatrix} a w \# 0^{ n''} \\ w' \end{pmatrix}$ where $a \in \{ 0,1\}$ in general, and in particular, when $a=1$ then $w' \in 0^*$. 
    
    \item Next we consider the simulation of type 2 edges. 
    Let $u=u_1u_2\ldots u_k$ where $u_i \in \Sigma$. Suppose $M_{w_1,w_2} \cap \Sigma^* u_i \Sigma^* \neq \emptyset$ for each $i$, i.e.\ edge $w_1 \xrightarrow{u \mid [w_1,w_2]} w_2$ exists. 
    Let $\cC'$ be in state $\begin{pmatrix} 1w_1\# 0^{ n''} \\ 0^{3+4|w_0|} \end{pmatrix}$. 
    $\cC'$ would like to execute a transition of $\cC$ which pushes the whole word $[w_1\#w_2]$ including the square brackets `$[$' and `$]$' to the stack. It executes the code shown in Algorithm \ref{algo:sim-type2-trans}:
    \begin{algorithm}[t]
      \SetKwBlock{Begin}{begin}{end}
      \SetKw{Break}{break}
      {\small
        Switch active tape from tape 1 to tape 2.\\
        Copy $w_1$ onto tape 2.\\
        Guess $w_2$ and write it down on tape 2, adjacent to $w_1$.\\
        \While{*}{
          Guess $a$\\
          \eIf {$M_{w_1,w_2} \cap \Sigma^* a \Sigma^* \neq \emptyset$}{
            Make a transition reading $a$ with no other changes
          }
          {ERROR}
        }
        Push $[w_1\#w_2]$ onto the stack.\\
        Replace $w_1$ by $w_2$ on tape 1.\\
        Replace all symbols on tape 2 by $0$ and the first symbol on tape 1 by $1$.
      }
      \caption{Simulation of type 2 edges}\label{algo:sim-type2-trans}
    \end{algorithm}
    
    Line 1 involves changing the first letter of the state from $\begin{pmatrix} 1 \\0 \end{pmatrix}$ to $\begin{pmatrix} 0 \\ 1 \end{pmatrix}$. The $w_1$ and $w_2$ written down in Lines 2 and 3 are separated by means of a symbol $\#$. The while condition in line 4 is a nondeterministic guess where, if $|u|=k$ then the loop is run $k$ times. The condition $M_{w_1,w_2} \cap \Sigma^* a \Sigma^* \neq \emptyset$ is checked by the following sequence of steps:
    \begin{enumerate}
      \item Push a symbol $\$$ onto the stack.
      \item Guess a loop on state $w_1$ one step of $\cC$ at a time. Verification that state $w_1$ has been reached again can be done since since $w_1$ has been stored on tape 2. 
      \item At the end of the loop, push a symbol $\dagger$ onto the stack.
      \item Guess a path from $w_1$ to $w_2$ and verify that $w_2$ has been reached.
      \item Pop the symbol $\dagger$ from the stack.
      \item Guess a loop on state $w_2$ and ensure that at the end, the top of the stack is $\$$ which is then popped.
    \end{enumerate}
    The pushing of $[w_1 \# w_2]$ onto the stack in Line 12 can also be done by reading $w_1$ and $w_2$ from tape 2. This completes one type 2 transition and thus we execute Line 13 which replaces the state $w_1$ on tape 1  by $w_2$. Finally in Line 14, we clean up tape 2 and reset the control to tape 1.\\
    Similar to the above procedure, we also execute a procedure for a transition where $\cC$ pops a whole word $[w_1\#w_2]$ simulating the dual automaton instead.
    
    This simulation of type 2 edges stores up to four strings of length $|w_0|$ each on the second tape. The first two are $w_1$ and $w_2$, an additional two may need to be stored while going through paths and cycles of $\cC$. In total this needs $4|w_0| + 3$ space for the four strings plus separator symbols.
  \end{enumerate}
  We note that in the case of a succinct $\PDA$, the number of states is exponential and hence the corresponding height bound $h$ which is $O(|Q'|^2)$ is also exponential (here $Q' = (\Delta')^{|w'_0|}$ is the set of states of $\cE(\cC')$). However, the bit length of $h$ is polynomial in the size of the input.

\subsection{Proof Details for Lemma~\ref{bounded-stack-nfa}}
\label{sec:appendix_bounded-stack_nfa}

Given a succinct $\PDA$ $\cC=((T_c)_{c \in \scrC},\Sigma,\Gamma,\Delta, w_0,w_f)$ with $|w_0|=n$ and $T_c=(Q_c,\Delta,\delta_c,q_{0,c},q_{f,c})$ for each $c \in \scrC$,
as well as a stack height bound $h$ in binary,
we construct a doubly succinct NFA $\cB=((T_a)_{a \in \Sigma_{\varepsilon}},\Sigma,\Delta',M)$.
In detail, $\cB$ is given as:
\begin{enumerate}
  \item $\Delta'=\Gamma \cup \Delta \cup \{ \$, \#,\dagger\}$ where $\#,\dagger$ are new symbols,
  \item $M=|w_0|+h+3$
  \item For each $a \in \Sigma$, $T_a=(Q_a, \Delta, \delta_a,q_{0,a},q_{f,a})$ is constructed as follows:
  \begin{itemize}
    \item $Q_a=\{ q_{0,a},q_{f,a}\} \cup \left(\hat{\Gamma}\times \bigcup_{c=(a,*),* \in \hat{\Gamma}} Q_c \times \{ 0,1,2\}\right)$
    \item $\delta_a$ has the following transitions
    \begin{enumerate}
      \item For each $v \in \hat{\Gamma}$, $q_{0,a} \xrightarrow{\$ | \$} (v,q_{0,a,v},0)$,
      \item For each $v \in \hat{\Gamma}$, if $(q \xrightarrow{b|c} q') \in \delta_{a,v}$ then $((v,q,0) \xrightarrow{b|c} (v,q',0)) \in \delta_a$,
      \item For each $v \in \hat{\Gamma}$, $(v,q,0) \xrightarrow{\#|\#} (v,q,1)$,
      \item For each $v \in \hat{\Gamma}$ and each $b \in \Delta$, $(v,q,1) \xrightarrow{b|b} (v,q,1)$,
      \item For each $v \in \Gamma$ and each $b \in \Delta$, $(v,q,1) \xrightarrow{\dagger | v} (v,q,2)$,\\
      $(v,q,2) \xrightarrow{b | \dagger} q_{f,a}$,\\
      $q_{f,a} \xrightarrow{b|b} q_{f,a}$,
      \item For each $v \in \overline{\Gamma}$ and each $b \in \Delta$, $(v,q,1) \xrightarrow{v | \dagger} (v,q,2)$,\\
      $(v,q,2) \xrightarrow{\dagger|0} q_{f,a}$,\\
      $q_{f,a} \xrightarrow{b|b} q_{f,a}$,
      \item $(\varepsilon,q,1) \xrightarrow{\dagger | \dagger} q_{f,a}$, and\\
      $q_{f,a} \xrightarrow{b|b} q_{f,a}$.
    \end{enumerate} 
  \end{itemize}
  \item $T_{\varepsilon}=(Q_{\varepsilon}, \Delta, \delta_{\varepsilon},q_{0,\varepsilon},q_{f,\varepsilon})$ is constructed as follows:
  \begin{itemize}
    \item $Q_{\varepsilon}=Q_{\init} \cup Q_{\move} \cup Q_{\final}$ where $Q_{\init}=\{ q_i \mid 0 \leq i \leq n+3\}$, $Q_{\move}=\{ q_{0,\varepsilon},q_{f,\varepsilon}\} \cup \left(\hat{\Gamma}\times \bigcup_{c=(\varepsilon,*),* \in \hat{\Gamma}} Q_c \times \{ 0,1,2\}\right)$ similar to the construction of $T_a$ for $ a \in \Sigma$, and $Q_{\final}=\{ q'_i \mid 0 \leq i \leq n + 3\}$.
    \item $\delta_{\varepsilon}$ contains:
    \begin{enumerate}
      \item All transitions between states in $Q_{\move}$ as defined in the case of $T_a$ for $ a \in \Sigma$, just with $\varepsilon$ instead of $a$,
      \item For states in $Q_{\init}$:
      \begin{enumerate}
        \item $q_{0,\varepsilon} \xrightarrow{\varepsilon | \varepsilon} q_0$,
        \item $q_0 \xrightarrow{0 | \$} q_1$,
        \item For each $i, 1 \leq i \leq N$, $q_i \xrightarrow{0|w_{0,i}} q_{i+1}$ where $w_0=w_{0,1}\ldots w_{0,n}$,
        \item $q_{n+1} \xrightarrow{0|\#}q_{n+2}$
        \item $q_{n+2} \xrightarrow{0|\dagger} q_{n+3}$,
        \item $q_{n+3} \xrightarrow{0|0} q_{n+3}$ and
        \item $q_{n+3} \xrightarrow{\varepsilon|\varepsilon} q_{f,\varepsilon}$.
      \end{enumerate}
      
      \item For states in $Q_{\final}$:
      \begin{enumerate}
        \item $q_{0,\varepsilon} \xrightarrow{\varepsilon | \varepsilon} q'_0$,
        \item $q'_0 \xrightarrow{\$|1} q'_1$
        \item For each $i, 1 \leq i \leq N$, $q'_i \xrightarrow{w_{f,i}|1} q'_{i+1}$ where $w_f=w_{f,1}\ldots w_{f,n}$,
        \item $q'_{n+1} \xrightarrow{\#|1}q'_{n+2}$
        \item For each $b \in \Delta$, $q'_{n+2} \xrightarrow{b|1} q'_{n+3}$,
        \item $q'_{n+3} \xrightarrow{\dagger|1} q_{f,\varepsilon}$.
      \end{enumerate}

    \end{enumerate}
  \end{itemize}
\end{enumerate}
The transitions in (4b) are used to initialize the $\NFA$ $\cB$, converting the state $0^M$ to state $\$w_0\#\dagger0^{h}$. The string between the symbols and $\$$ and $\#$ is used to represent the stack, which is empty in the beginning. \\

Let us see how $\cB$ simulates one step of $\cC$. 
The transition in (3a) ensures that the initialization has taken place by checking that $\$$ is the first letter. 
Using (3c) we guess a value for $v$ and start simulating $T_c$ on its state where $c=(a,v)$. 
The simulation is completed when the symbol $\#$ is reached and in (3c) we move to the next stage where we simulate the action on the stack. 
The rule (3d) is used to maintain the stack contents which are not at the top of stack and using (3e),(3f), and (3g) we guess the 
top of the stack in the case of $ v \in \Gamma, v \in \hat{\Gamma}$ and $v=\varepsilon$ respectively.
If the step of $\cC$ does not read an input letter, then the simulation is very similar, but uses the transitions in (4a) instead.

At any point between two steps of the simulation, we use the transitions in (4c) to guess that state $w_f$ has been reached and verify it letter by letter, converting any string of the form $\$w_f \# w' \dagger$ to the string $1^M$, which is the accepting state of $\cB$.

%% file: appendix-succinctTasks.tex
\section{Proof Details for Succinct Downward Closures of Tasks} %
\label{sec:appendix_succinct_dcl_tasks}

This section contains additional details for various proofs found in Section~\ref{sec:succinct_dcl_tasks}.

\subsection{Proof Details for Lemma~\ref{theta-downclosure-extraction}} %
\label{sec:appendix_theta-dcl-extract}

Let $\ell$ be the bit-length of $k$, let $\cB=((T_a)_{a \in \Sigma_{\varepsilon}},\Sigma,\Delta,M)$ and let $T_a=(Q_a, \Delta, \delta_a,q_{0,a},q_{f,a})$. Then $\cB'=((T'_a)_{a \in \Sigma_{\varepsilon}},\Sigma,\Delta',M')$ is defined as follows:
\begin{itemize}
  \item $\Delta'=\Delta \cup \{ \$,\#\}$ where $\$,\#$ are new symbols, 
  \item $M'=M+2+\ell$ and 
  \item $T'_a=(Q'_a, \Delta', \delta'_a,q'_{0,a},q'_{f,a})$ is defined as follows: 
  \begin{enumerate}
    \item If $a \in \Theta$ then $Q'_a=\{q'_{0,a},q'_{f,a} \} \cup Q_a \cup Q_{\ctr}$ where $Q_{\ctr}=\{ q_0,q_1,q_2\}$\\
    else if $a \in \Sigma \setminus \Theta$ then $Q'_a=\{q'_{0,a},q'_{f,a},q \} \cup Q_a$\\
    else $a =\varepsilon$ and $Q'_{\varepsilon}=Q_{\init}\cup Q_{\final} \cup \{q'_{0,a},q'_{f,a},q \} \cup Q_a$ where $Q_{\init}=\{ q_i \mid 0 \leq i \leq \ell+1\}$ and $Q_{\final}=\{ q'_i \mid 0 \leq i \leq \ell+1\}$. 
    \item $\delta'_a$ is defined as follows: 
    \begin{enumerate}
      \item For each $a \in \Sigma_{\varepsilon}$, transitions on $Q_a$ are inherited from $T_a$. 
      \item If $a \in \Theta$, \\
      $q'_{0,a} \xrightarrow{\$|\$} q_0$ where $q_0 \in Q_{\ctr}$\\
      else $a \in \Sigma_{\varepsilon} \setminus \Theta$ and\\
      $q'_{0,a} \xrightarrow{\$|\$} q$. 
      \item $q'_{0,\varepsilon} \xrightarrow{0 | \$} q_0$ where $q_0 \in Q_{\init}$ and\\
      $q'_{0,\varepsilon} \xrightarrow{\$|1} q'_0$ where $q'_0 \in Q_{\final}$. 
      \item For each $a \in \Sigma_{\varepsilon}$, 
      $q_{f,a} \xrightarrow{\varepsilon|\varepsilon} q'_{f,a}$. 
      \item If $a \in \Theta$ then\\
      $q_0 \xrightarrow{0|1} q_1$\\
      $q_0 \xrightarrow{1|0} q_2$\\
      $q_1 \xrightarrow{0|0} q_1$\\
      $q_1 \xrightarrow{1|1} q_1$\\
      $q_2 \xrightarrow{0|1} q_1$\\
      $q_2 \xrightarrow{1|0} q_2$\\
      $q_1 \xrightarrow{\#|\#} q_{0,a}$. 
      \item If $a \in \Sigma \setminus \Theta$ then\\
      For each $b \in \Delta$, $q \xrightarrow{b|b} q$ and\\
      $q \xrightarrow{\#|\#} q_{0,a}$. 
      \item If $a=\varepsilon$ then\\
      \begin{enumerate}
        \item Transitions between states in $Q_{\init}$ are given as\\
        for each $i, 0 \leq i \leq \ell-1$, $q_i \xrightarrow{0|0} q_{i+1}$\\
        $q_{\ell} \xrightarrow{0|\#} q_{\ell+1}$\\
        $q_{\ell+1} \xrightarrow{0|0} q_{\ell+1}$\\
        $q_{\ell+1} \xrightarrow{\varepsilon|\varepsilon} q_{0,a}$
        \item Let $k=b_1b_2\ldots b_{\ell}$ in binary, where $b_1$ is the least significant digit. Transitions between states in $Q_{\final}$ are given as\\
        For each $i, 0 \leq i \leq \ell-1$, $q'_i \xrightarrow{b_{i+1}|1} q'_{i+1}$\\
        $q'_{\ell} \xrightarrow{\#|1} q'_{\ell+1}$\\
        $q'_{\ell+1} \xrightarrow{1|1} q'_{\ell+1}$\\
        $q'_{\ell+1} \xrightarrow{\varepsilon|\varepsilon} q'_{f,a}$. 
      \end{enumerate}
    \end{enumerate}
  \end{enumerate}
\end{itemize}
The construction is very similar to that used in Lemma \ref{bounded-stack-nfa}. 
$\cB'$ first converts its state $0^{M'}$ to $\$0^{\ell}\#0^M$ using the first transition in (2c) and the transitions in (2gi).\\
After initialization, $\cB$ begins simulating the behaviour of $\cB$ in the part of the state after $\#$ using (2a) and accepting using (2d). The part of the state between $\$$ and $\#$ is used to store a counter which is incremented each time an input letter from $\Theta$ is used using the transitions in (2e). In case the input letter is not from $\Theta$, then the counter is simply retained without change using (2f). As before, $Q_{\final}$ states and transitions are added to $T_{\varepsilon}$ in order to verify that the counter value is indeed $\ell$ before accepting. 

%% file: appendix-tcvass.tex
\section{Proof Details of Theorem~\ref{dcps-to-tcvass}} %
\label{sec:appendix_tcvass}

The details of $\cV=((T'_i)_{i \in \hatie},\Delta',M')$ are given below. We use $R$ (with subscripts, superscripts) for the set of transducer states and $F$ for the set of edges. In particular,
the task-$\dsNFA$ $\cB_{(g,\gamma)}$ contains transducers $T_{a,g,\gamma}=(R_{a,g,\gamma}, \Delta, F_{a,g,\gamma},r_{0,a,g,\gamma},r_{f,a,g,\gamma})$ for each symbol $a \in \Sigma$. In the following description, we refer to the global state of $\cV$ as the 'tape configuration' since we will need to talk about the global state of the DCPS and how it is stored and transferred. The state of a transducer is explicitly mentioned.

\begin{itemize}
  \item $\Delta'=\Delta \mdotcup \{ \$_1,\#_1,\dagger_1 \} \mdotcup G \mdotcup (G \times \Gamma \times \{ 0,1\}) \mdotcup \{ \bot\}$,
  \item $M'=N(M+2)+2$,
  \item $T'_i=(R'_i, \Delta', F'_i,r'_{0,i},r'_{f,i})$ is defined as follows:
  \begin{itemize}
    \item If $i \in I$ then $R'_i= \mdotbigcup_{g \in G,\gamma \in \Gamma} R_{\gamma_i,g,\gamma} \mdotcup 
    \{ r'_{0,i},r'_{1,i},r'_{2,i},r'_{f,i} \} $\\
    else if $i \in \barI$ then $R'_i=\{r'_{0,\bari},r'_{f,\bari} \}
    \cup ( \{ r_1,r_2\} \times G)$ \\
    else $i=\varepsilon$ and $R'_{\varepsilon}=R'_{\init} \mdotcup R'_{\final} \mdotcup R'_{\clean} \mdotcup R'_{\term} \mdotcup R'_{\inter} \mdotcup
    \{ r'_{0,\varepsilon},r'_{f,\varepsilon} \} \mathbin{\dot\cup} R'_{\emp} $.\\
    where $R'_{\emp}=\{ r'_1,r'_2,r'_3,r'_4\} \mdotcup G \mdotcup \mdotbigcup_{g \in G, \gamma \in \Gamma} R_{\varepsilon,g,\gamma}$,\\
    $R'_{\clean}=\{ r'_1,r'_2,r'_3,r'_4,r'_5\}$,\\
    $R'_{\final}=\{ r'_1,r'_2,r'_3,r'_4\}$,\\
    $R'_{\term}=(\{ r'_1,r'_2,r'_3\} \times G) \mdotcup \{ r'_4\} \mdotcup \mdotbigcup_{a \in \Sigma_T,g\in G, \gamma \in \Gamma} R_{a,g,\gamma}$ ,\\
    $R'_{\inter}=(\{ r'_1,r'_2,r'_3,\} \times G)\mdotcup \{ r'_4\} \mdotcup \mdotbigcup_{a \in \Sigma_I 
    ,g\in G, \gamma \in \Gamma} R_{a,g,\gamma}$,\\

    \item $F'_i$ is given as follows:
    \begin{enumerate}
      \item If $i \in I$ then we have
      \begin{enumerate}
        \item the transitions on $R_{\gamma_i,g,\gamma}$ inherited from $\cB_{(g,\gamma)}$,
        \item $r'_{0,i} \autstep[\$_1|\$_1] r'_{1,i}$%
        \item for each $b \in \Delta'\setminus\{ \bot\}$, $r'_{1,i} \xrightarrow{b|b} r'_{1,i}$,%
       \item $r'_{1,i} \xrightarrow{\#_1|\#_1} r'_{2,i}$,%
       \item for each $g \in G, \gamma \in \Gamma$, 
       $r'_{2,i} \autstep[(g,\gamma,1)|(g,\gamma,1)] r_{0,\gamma_i,g,\gamma}$, 
       \item for each $g \in G, \gamma \in \Gamma$, 
       $r_{f,\gamma_i,g,\gamma} \xrightarrow{\#_1|\#_1} r'_{f,i}$,\\
       $r_{f,\gamma_i,g,\gamma} \xrightarrow{\dagger_1|\dagger_1} r'_{f,i}$,
       \item for each $b \in \Delta'\setminus\{ \bot\}$, $r'_{f,i} \xrightarrow{b|b} r'_{f,i}$.
      \end{enumerate}
      \item If $i \in \barI$ then we have
      \begin{enumerate}
        \item for each $g \in G$, $r'_{0,\bari} \autstep[g|\$_1] (r'_1,g)$,%
        \item for each $g \in G$ and $b \in \Delta' \setminus \{ \bot\}$, $(r'_1,g) \autstep[b|b] (r'_1,g)$,%
        \item for each $g \in G$, $(r'_1,g) \autstep[\dagger_1|\#_1] (r'_2,g)$,%
        \item for each $g \in G$, $(r'_2,g) \autstep[0|(g,\gamma_i,1)] r'_{f,\bari}$,%
        \item for each $g \in G$ and $b \in \Delta'\setminus\{ \bot\}$, 
        $r'_{f,\bari} \autstep[b|b] r'_{f,\bari}$
      \end{enumerate}
      \item If $i=\varepsilon$ then we have
      \begin{enumerate}
        \item transitions on $R'_{\emp}$:
        \begin{enumerate}
          \item transitions on $\mdotbigcup_{g \in G, \gamma \in \Gamma} R_{\varepsilon,g,\gamma}$ inherited from the corresponding $\dsNFA$ $\cB_{(g,\gamma)}$ for each $g \in G, \gamma \in \Gamma$,
          \item $r'_{0,\varepsilon} \autstep[\varepsilon|\varepsilon] r'_1$
          \item $r'_1 \autstep[\$_1|\$_1] r'_2$
          \item for each $b \in \Delta'\setminus\{ \bot\}$, $r'_2 \autstep[b|b] r'_2$
          \item $r'_2 \autstep[\#_1 | \#_1] r'_3$
          \item for each $g \in G, \gamma \in \Gamma$, $r'_3 \autstep[(g,\gamma,1)| (g,\gamma,1)] r_{0,\varepsilon, g, \gamma}$
          \item $r_{f,\varepsilon, g, \gamma} \autstep[\#_1|\#_1] r'_4$ and \\
          $r_{f,\varepsilon, g, \gamma} \autstep[\dagger_1|\dagger_1] r'_4$
          \item for each $b \in \Delta'\setminus\{ \bot\}$ $r'_4 \autstep[b | b] r'_4$
          \item $r'_4 \autstep[\varepsilon|\varepsilon] r'_{f,\varepsilon}$
          \item for each $g \in G$,  $r'_1 \autstep[g | \$_1] g$
          \item for each $g \in G, b \in \Delta'\setminus\{ \bot\}$, $g \autstep[b | b] g$
          \item for each $g \in G, \gamma \in \Gamma$, $g \autstep[(g,\gamma,0)|(g,\gamma,1)] r'_4$
        \end{enumerate}
        
        \item transitions on $R'_{\init}$ which convert the string $0^{M'}$ to the string 
        $\$_1\#_1(g_0,\gamma_0,1)0^M\dagger_1 (0^{M+1} \dagger_1)^{N-1}$ by Lemma \ref{lem:dsNFA-create-separators},
        and then running the transducer initializing a thread of type $(g_0,\gamma_0)$.
        \item transitions on $R'_{\clean}$ given as follows:
        \begin{enumerate}
          \item $r'_{0,\varepsilon} \autstep[\varepsilon|\varepsilon] r'_1 $,
          \item for each $g \in G$, $r'_1 \autstep[g|g] r'_2 $,
          \item for each $b \in \Delta' \setminus \{ \bot\}$, $r'_2 \autstep[b|b] r'_2 $,
          \item $r'_2 \autstep[\#_1|\dagger_1] r'_3 $,%
          \item $r'_3 \autstep[\bot|0] r'_4 $,%
          \item for each $b \not\in \{ \#_1,\dagger_1\}$ $r'_4 \autstep[b|0] r'_4 $,
          \item $r'_4 \autstep[\#_1|\#_1] r'_5 $ and $r'_4 \autstep[\dagger_1|\dagger_1] r'_5$, %
          \item for each $b \in \Delta' \setminus \{ \bot\}$, $r'_5 \autstep[b|b] r'_5 $,
          \item $r'_5 \autstep[\varepsilon|\varepsilon] r'_{f,\varepsilon} $
        \end{enumerate}
  
        \item transitions on $R'_{\inter}$ given as follows:
        \begin{enumerate}
          \item  for each $g \in G$, $r'_{0,\varepsilon} \autstep[\varepsilon|\varepsilon] (r'_1,g) $,

          \item for each $g \in G$, $(r'_1,g) \autstep[\$_1|g] (r'_2,g)$%
          \item for each $g \in G, b \in \Delta'\setminus\{ \bot\}$, $(r'_2,g) \autstep[b|b] (r'_2,g)$
          \item for each $g \in G$,$(r'_2,g) \autstep[\#_1|\#_1] (r'_3,g) $
          \item for each $g,g',g'' \in G,\gamma,\gamma'' \in \Gamma$ and $a =(g,\gamma,g') \in \Sigma_I$,\\
          $(r'_3,g) \autstep[(g'',\gamma',1)|(g',\gamma,0)] r_{0,a,g'',\gamma'} $
          \item for each $a \in \Sigma_I, g'' \in G, \gamma' \in \Gamma$,\\
            $r_{f,a,g'',\gamma'} \autstep[\#_1,|\#_1] r'_4$ and \\
          $r_{f,a,g'',\gamma'} \autstep[\dagger_1,|\dagger_1] r'_4$
          \item for all $b \in \Delta'\setminus\{ \bot\}$, $r'_4 \autstep[b|b] r'_4$
          \item $r'_4 \autstep[\varepsilon|\varepsilon] r'_{f,\varepsilon}$
        \end{enumerate}
        \item transitions on $R'_{\term}$ given as follows:
        \begin{enumerate}
          \item for each $g \in G$, $r'_{0,\varepsilon} \autstep[\varepsilon|\varepsilon] (r'_1,g) $,

          \item for each $g \in G$, $(r'_1,g) \autstep[\$_1|g] (r'_2,g)$%
          \item for each $g \in G, b \in \Delta'\setminus\{ \bot\}$, $(r'_2,g) \autstep[b|b] (r'_2,g)$
          \item for each $g \in G$, $(r'_2,g) \autstep[\#_1|\#_1] (r'_3,g) $
          \item for each $g,g' \in G, \gamma \in \Gamma$, 
          $(r'_3,g) \autstep[(g',\gamma,1)|\bot] r_{0,(g,\bot),g',\gamma}$%

          \item for each $g,g' \in G, \gamma \in \Gamma$, 
          $r_{f,(g,\bot),g',\gamma} \autstep[\#_1|\#_1] r'_4$ and 
          $r_{f,(g,\bot),g',\gamma} \autstep[\dagger_1|\dagger_1] r'_4$

          \item for each $g \in G, b \in \Delta'\setminus\{ \bot\}$, $r'_4 \autstep[b|b] r'_4 $
          \item $r'_4 \autstep[\varepsilon|\varepsilon] r'_{f,\varepsilon} $
        \end{enumerate}
        \item transitions on $R'_{\final}$ given as follows:
        \begin{enumerate}
          \item  $r'_{0,\varepsilon} \autstep[\varepsilon|\varepsilon] r'_1 $,
          \item $r'_1 \autstep[g_f|1] r'_2$
          \item for each $b \in \Delta' \setminus \{ \bot\}$, $r'_2 \autstep[b|1] r'_2$
          \item $r'_2 \autstep[\varepsilon|\varepsilon] r'_{f,\varepsilon}$
        \end{enumerate}
      \end{enumerate}

    \end{enumerate}
  \end{itemize}
\end{itemize}

Suppose there is a run of the  $\DCPS$ $\cD=(G,\Gamma,E_D,g_0,\gamma_0,1)$ which reaches the global state $g_f$ using at most $N$ threads
each of which is switched at most $K$ times. Let us see how $\cV'$
simulates such a run.  $\cV'$ first initialises
itself to the tape configuration $\$_1\#_1(g_0,\gamma_0,1)0^M\dagger_1 (0^{M+1} \dagger_1)^{N-1}$ and with all counter values being $0$. In general, the tape configuration is of the form $a\Box_1w_1\Box_2w_2\ldots \Box_N w_N \Box_{N+1}$ where:
\begin{itemize}
  \item $a$ is a single symbol which is either $\$_1$ indicating that there is an active thread which is running, or $a$ is an element of $G$ in which case the $\DCPS$ is at a schedule point and thus all threads are inactive. 
  \item Each $w_i$ is a string which either represents the configuration of a thread (in which case we say the segment is \emph{occupied}) or filled with blanks (in which case we say the segment is \emph{unoccupied}); and $w_i$ starts with a symbol of the form $(g,\gamma,j)$ where $g \in G, \gamma \in \Gamma, j \in \{ 0,1\}$. We call the $w_i$'s \emph{segments}. At most one of the segments starts with a symbol where $j=1$ and this indicates the active thread.  
  \item Each $\Box_i$ is a symbol which is either $\#_1$ or $\dagger_1$. A $\#_1$ (resp. $\dagger_1$) indicates that the segment immediately after is occupied (resp. unoccupied).
\end{itemize}

The simulation
of steps taken by the active thread is done using (1a) to (1g) for the
case when the step creates a new spawn and by using (3ai) to (3aviii) 
when the step does not create a spawn. This is done as follows:
\begin{enumerate}[(A)]
  \item Check if the first symbol on the tape is a $\$_1$ which means that there is an active thread and retain it (rules 1b,3aiii).
  \item Nondeterministically guess the segment containing the active thread, retaining the symbols prior to it (rules 1c and 1d, 3aiv and 3av).
  \item Verify that the first symbol of the guessed segment is indeed of the form $(g,\gamma,1)$ , signifying the active segment (rules 1e, 3av). At this point, control is passed to the corresponding transducer $\cB_{(g,\gamma)}$. 
  \item This is followed by the 'local' computation inside the segment (rules 1a, 3ai). Recall that the task-$\PDA$ has states of the form $(g,\gamma)$ which stores the surface configuration when the thread makes a context switch. As we describe later, $\cV$ ensures that
  the state information in the first symbol corresponds to that stored internally by $\cB_{g,\gamma}$ while simulating the context switches.  
  \item Control is then passed out of $\cB_{g,\gamma}$ from the final state of the appropriate transducer on reading a $\#_1$ or $\dagger_1$ which indicate the end of the segment (rules 1f,3avii).
  \item The rest of the segments are ignored (rules 1g, 3aviii and 3aix). 
\end{enumerate}

\noindent \textbf{Context switching} works in two phases: the first phase is used to transfer the global state information from the active thread to the first symbol of the state of $\cV$. In the second phase, this information is used to wake up one of the inactive threads and make it the active thread. Note that while the task-$\PDA$ reads a symbol while doing a context switch, this is no longer the case in $\cV$. Thus context switching is done in the transducer $T'_{\epsilon}$.\\
 \underline{Phase 1:} 
  $\cV$ uses the transitions of the $\inter$ module in (3d). Let $a=(g,\gamma,g') \in \Sigma_I$ be the context switch. The active thread is made inactive as follows:
  \begin{enumerate}[(A)]
    \item $\cV$ guesses $g$ and stores it in the state of the transducer (3di).
    \item The first symbol $\$_1$ is rewritten to $g$ using (3dii).
    \item As before, we guess the active segment (3div) while ignoring the ones preceding it (3diii).
    \item The first symbol $(g'',\gamma',1)$ is replaced by $(g',\gamma,0)$ and control is passed to the transducer $T_{a,g'',\gamma'}$ (3dv).
    \item $T_{a,g'',\gamma'}$ applies steps corresponding to the transition $g_1 \autstep[(g,\gamma,g')|\gamma'''/w] (g',\gamma)$ in (4) of \cref{sub:threads} for the task-$\PDA$.
    \item Control is passed back out of  $T_{a,g'',\gamma'}$ at the end of the segment (3dvi).
    \item The rest of the segments are ignored (3dvii,3dviii).
  \end{enumerate}
Note that at this point, all threads have 
 been made inactive and we are at a schedule point.\\ 
 \underline{Phase 2:} We have two choices at a schedule point: \\
 \underline{Choice 1:} We may load a new task into an empty segment. 
 This is accomplised by transitions in (2) which decrement a counter
 $i$ and load it as the active thread in the following sequence of steps:
 \begin{enumerate}[(A)]
   \item Transfer the global state information $g$ from the first symbol into the state of the transducer (2a).
   \item Guess an unoccupied segment (2c) while ignoring the segments prior to it (2b).
   \item Rewrite the first 0 of the segment by $(g,\gamma_i,1)$, making it the active segment (2d). Note that the first symbol is a 0 if and only if the segment is unoccupied, which serves to verify our guess.
   \item Ignore the rest of the segments (2e).
 \end{enumerate}
 \underline{Choice 2:} Alternately, we may load in an existing thread
 which had previously been switched out. This is done by transitions
 in  (3aviii) to (3axii) as follows:
\begin{enumerate}[(A)]
   \item Transfer the information $g \in G$ in the first symbol of the tape to the state of the transducer (3ax). Note that we replace $g$ by $\$$ since at the end of this step, we will once again have an active thread.
   \item Choose an inactive thread to activate (3axii) while ignoring the previous ones (3axi).
   \item Ignore the rest of the segments (3aviii) and accept (3aix).
 \end{enumerate} 

\noindent\textbf{Termination} of the active thread corresponds to the reading of a symbol $(g,\bot) \in \Sigma_T$ in the task-$\PDA$. Similar to context switching, this is carried out in $T'_{\varepsilon}$ using the following steps:
\begin{enumerate}[(A)]
   \item Guess the current state $g$ of the active thread (3ei) and store it in the first symbol of the tape (3eii).
   \item Guess the active segment (3eiv) and replace its first symbol $(g',\gamma',1)$ by $\bot$ (3ev) while passing control to the transducer $T_{(g,\bot),g',\gamma'}$. At this point, the active thread should already be in the state $(g,\bot)$ using a rule $g' \autstep[\varepsilon|\gamma/w] (g,\bot)$ in (5) of \cref{sub:threads} for the task-$\PDA$. This allows the use of the rule $(g,\bot) \xrightarrow{(g,\bot)|\gamma_2/\gamma_2} \mathsf{end}$ from (5) of \cref{sub:threads} in order to terminate the thread.
   \item The control is passed back from $T_{(g,\bot),g',\gamma'}$ (3evi) and the rest of the segments are ignored (3evii) before acceptance (3eviii).
 \end{enumerate} 
 Note that the 
crucial difference between context switching and termination is that
the $\bot$ symbol that is placed by the module $\term$ in the segment which contained the terminated thread cannot be read by any module other than $\clean$.
This ensures that the transitions in (3c) are used immediately after
thread termination to ensure clean up of the segment as follows:
\begin{enumerate}[(A)]
  \item Choose the $\clean$ module (3ci)
   \item Check that the $\DCPS$ at a schedule point and the first symbol on the tape is some $g \in G$ (3cii).
   \item Guess the segment in which the thread has just been terminated (3civ) while ignoring the  previous segments (3ciii).
   \item Replace the $\bot$ symbol by $0$ (3cv) and then all other symbols in the current segment also by $0$ (3cvi).
   \item Ignore the rest of the segments (3cvii,3cviii) and accept (3cix). 
 \end{enumerate} 
 The continuation
of the simulation after this cleanup process is basically the Phase 2 of the context
switching process. 

\noindent\textbf{Checking State Reachability}
The transitions in (3f) can be used at any point to check if the $\DCPS$ has reached the state $g_f$. The $\final$ module is constructed to only check if the global state is $g_f$ at a schedule point. This is because the $\DCPS$ can be easily modified to include swap rules for every possible top of stack for the target state $g_f$. Thus there is always a corresponding run where $g_f$ occurs in the first cell of the tape if there is a run of the $\DCPS$ where a particular thread reaches the state $g_f$, by performing a context switch immediately after this happens. Thus all we need to do is:
\begin{enumerate}[(A)]
  \item Choose the $\final$ module (3fi).
\item Check if the first symbol is $g_f$  and replace it by $1$ (3fii).
\item Continue to convert all symbols to $1$ (3fiii) and accept (3fiv).
\end{enumerate}
At the end of a successful run of the $\final$ module, $\cV$ reaches its final state consisting of all $1$'s on the tape.

By construction, if there is a run of the $\DCPS$ which reaches the state $g_f$, then there is also a run of $\cV$ which reaches its final state. Conversely, suppose $\cV$ has a run reaching its final state. This means that the active segment contains a symbol of the form $(g_f,c)$ which implies that the global state $g_f$ has been reached by the $\DCPS$. The two points of concern regarding whether this run corresponds to a valid run of the $\DCPS$ are at the places where control is passed from one thread to another i.e. context switching and termination. The rest of the simulation is internal to a particular segment and its correctness is guaranteed by previous constructions. \\
Context switching: Steps (D) and (E) in Phase 1 of context switching ensure that the internal state $g$ of the active thread after applying a context switch rule is the same as the first component of the context switch symbol $(g,\gamma,g')$. Further, step (E) ensures that the internal state $(g',\gamma)$ of the thread which has just been made inactive corresponds to the third and second components of $(g,\gamma,g')$. The state $g$ which is stored in the first cell of the tape is used in the second phase to wake up an inactive thread.\\
In (B) of Phase 2 of context switching, rule (3axii) is used to wake up a thread and the rule ensures that the state $(g,\gamma)$ of the thread being woken up has the same first component as the $g$ stored in the first cell of the tape. In the previous phase, we have already seen that the first symbol $(g,\gamma,0)$ is made to match the internal state $(g,\gamma)$ of the thread. Internal to the thread, the rule $(g, \gamma) \xrightarrow{\varepsilon|\gamma/\gamma} g$ from (4) in \cref{sub:threads} is applied and ensures that the top of stack is as specified. Alternately, we can load a new thread in Phase 2, but this is more straightforward and no consistency check is required as in the case of switching to a running thread.\\
Termination: The guess of the state $g$ stored in the active segment in (A) of termination is verified in (B). Termination is only possible if the active thread has already moved to the state 
$(g,\bot)$ prior to applying the termination step to reach the state $\mathsf{end}$. Further, the symbol $\bot$ can only be read by the module $\clean$ and this ensures that immediately after termination, the segment is made unoccupied. 

%% file: appendix-binarybound.tex
\section{Proof Details for Complexity of Context-bounded Reachability} %
\label{sec:appendix_binary_bound}

Like the previous proof of a lower bound by \citet{BaumannMajumdarThinniyamZetzsche2020a}, our proof proceeds in three steps: We reduce from termination of \emph{bounded counter programs} to termination of \emph{succinct recursive net programs} to coverability of \emph{succinct transducer-defined Petri nets} to finally context-bounded reachability for $\DCPS$ (without thread-pooling). The various problems and models involved in these reductions are explained in their respective subsections.

\subsection{Recursive net programs}

A \emph{counter program} is a finite sequence of \emph{labelled commands} separated by semicolons. 
Let $l, l_1, l_2$ be \emph{labels} and $x$ be a \emph{variable} (also called a \emph{counter}). 
The labelled commands have one of the following five forms:
\begin{alignat*}{2}
  (1)~l&: \text{\textbf{inc }}x; &&~\text{ \texttt{//} increment}\\
  (2)~l&: \text{\textbf{dec }}x; &&~\text{ \texttt{//} decrement}\\
  (3)~l&: \text{\textbf{halt}}\\
  (4)~l&: \text{\textbf{goto }} l_1; &&~\text{ \texttt{//} unconditional jump}\\
  (5)~l&: \text{\textbf{if }} x = 0 \text{\textbf{ then goto }} l_1 \text{\textbf{ else goto }} l_2; &&~\text{ \texttt{//} conditional jump}
\end{alignat*}

Variables can hold values over the natural numbers, labels have to be
pairwise distinct, but can otherwise come from some arbitrary set.
For convenience, we require each program to contain exactly one \textbf{halt} command at the very end. 
The \emph{size} $|C|$ of a counter program $C$ is the number of its labelled commands.

During execution, all variables start with initial value $0$. 
The semantics of programs follows from the syntax, except for the case of decrementing a variable whose value is already~$0$. 
In this case, the program aborts, which is different from proper termination, i.e., the execution of the \textbf{halt} command. 
It is easy to see that each counter program has only one execution, meaning it is deterministic. 
This execution is \emph{$k$-bounded} if none of the variables ever
reaches a value greater than $k$ during it. 

Let $\exp^{m+1}(x) := \exp(\exp^m(x))$ and $\exp^1(x) = \exp(x) := 2^x$.
The $M$-fold exponentially bounded \emph{halting problem} (also called \emph{termination}) for counter programs ($\HP[M]$) is given by:
\begin{description}
  \item[Given] A unary number $n \in \nats$ and a counter program $C$.
  \item[Question] Does $C$ have an $\exp^{M}(n)$-bounded execution that reaches the \textbf{halt} command?
\end{description}
We make use of the following well-known result regarding this problem:
\begin{theorem}
  For each $M > 0$, the problem $\HP[M+1]$ is $M$-$\EXPSPACE$-complete.
\end{theorem}
The definition of \emph{recursive net programs} ($\rnp$) also involves sequences of labelled commands separated by semicolons. Let $l, l_1, l_2$ be labels, $x$ be a variable, and \texttt{proc} be a procedure name. Then the labelled commands can still have one of the previous forms (1) to (4). However, form (5) changes from a conditional to a nondeterministic jump, and there are two new forms for procedure calls:
\begin{alignat*}{2}
  (1)~l&: \text{\textbf{inc }}x; &&~\text{ \texttt{//} increment}\\
  (2)~l&: \text{\textbf{dec }}x; &&~\text{ \texttt{//} decrement}\\
  (3)~l&: \text{\textbf{halt}}\\
  (4)~l&: \text{\textbf{goto }} l_1; &&~\text{ \texttt{//} unconditional jump}\\
  (5)~l&: \text{\textbf{goto }} l_1 \text{\textbf{ or goto }} l_2; &&~\text{ \texttt{//} nondeterministic jump}\\
  (6)~l&: \text{\textbf{call} \texttt{proc}}; &&~\text{ \texttt{//} procedure call}\\
  (7)~l&: \text{\textbf{return}}; &&~\text{ \texttt{//} end of procedure}
\end{alignat*}

In addition to labelled commands, these programs consist of a finite set $\mathsf{PROC}$ of procedure names and also a \emph{maximum recursion depth} $k \in \mathbb{N}$. Furthermore, they not only contain one sequence of labelled commands to serve as the main program, but also include two additional sequences of labelled commands for each procedure name $\text{\texttt{proc}} \in \mathsf{PROC}$. The second sequence for each \texttt{proc} is not allowed to contain any \textbf{call} commands and serves as a sort of ``base case'' only to be called at the maximum recursion depth. Each label has to be unique among all sequences and each jump is only allowed to target labels of the sequence it belongs to. Each $\rnp$ contains exactly one \textbf{halt} command at the end of the main program. For $\text{\texttt{proc}} \in \mathsf{PROC}$ let $\#c(\texttt{proc})$ be the number of commands in both of its sequences added together and let $\#c(\texttt{main})$ be the number of commands in the main program. Then the \emph{size} of an $\rnp$ $R$ is defined as $|R| = \lceil\log{k}\rceil + \#c(\texttt{main}) + \sum_{\text{\texttt{proc}} \in \mathsf{PROC}} \#c(\texttt{proc})$.

The semantics here is quite different compared to counter programs: If
the command ``$l:\text{\textbf{call} \texttt{proc}}$'' is executed, the label $l$ gets pushed onto the call stack. Then if the stack contains less than $k$ labels, the first command sequence pertaining to \texttt{proc}, which we now call $\text{\texttt{proc}}_{<\text{max}}$, is executed. If the stack already contains $k$ labels, the second command sequence, $\text{\texttt{proc}}_{=\text{max}}$, is executed instead. Since $\text{\texttt{proc}}_{=\text{max}}$ cannot call any procedures by definition, the call stack's height (i.e.\ the \emph{recursion depth}) is bounded by $k$. On a $\text{return}$ command, the last label gets popped from the stack and we continue the execution at the label occurring right after the popped one.

How increments and decrements are executed depends on the current recursion depth $d$ as well. For each variable $x$ appearing in a command, $k + 1$ copies $x_0$ to $x_k$ are maintained during execution. The commands \textbf{inc}~$x$ resp.\ \textbf{dec}~$x$ are then interpreted as increments resp.\ decrements on $x_d$ (and not $x$ or any other copy). As before, all these copies start with value~$0$ and decrements fail at value~$0$, which is different from proper termination.

Instead of a conditional jump, we now have a nondeterministic one, that allows the program execution to continue at either label. Regarding termination we thus only require there to be at least one execution that reaches the \textbf{halt} command. This gives us the following halting problem for $\rnp$:
\begin{description}
  \item[Given] An $\rnp$ $R$
  \item[Question] Is there an execution of $R$ that reaches the \textbf{halt} command?
\end{description}

Using a reduction from $\HP[M]$, \citet{BaumannMajumdarThinniyamZetzsche2020a} have shown that the halting problem is $\EXPSPACE$-hard resp.\ $\TWOEXPSPACE$-hard for $\rnp$ with unary-encoded resp.\ binary-encoded recursion depth. Since we want to show $\THREEEXPSPACE$-hardness we need to encode the recursion depth even more succinctly:

A \emph{succinct recursive net program at level $M$} ($M$-$\srnp$) for $M \in \N\backslash\{0\}$ is defined in the same way as an $\rnp$, except that its recursion depth is now $\exp^M(k)$ instead of just $k$. Additionally, its \emph{size} is defined as $M + k + \#c(\texttt{main}) + \sum_{\text{\texttt{proc}} \in \mathsf{PROC}} \#c(\texttt{proc})$ ($M$ and $k$ are unary-encoded).

It follows from the construction in \cite{BaumannMajumdarThinniyamZetzsche2020a} that the halting problem for $M$-$\srnp$ is $(M+1)$-$\EXPSPACE$-hard, meaning in particular it is $\THREEEXPSPACE$-hard for $2$-$\srnp$. Note that $M$-$\srnp$ are quite similar to counter system with chained counters at level $M$ from \cite{DBLP:conf/lics/DemriFP13}, where they also establish an $(M+1)$-$\EXPSPACE$-hardness result for their model.

\subsection{Transducer-defined Petri nets}

\begin{definition}
  \label{def:petri_net} 
  A \textbf{Petri net} is a tuple $N = (P,E,F,p_0,p_f)$ where
  $P$ is a finite
  set of \emph{places}, $E$ is a finite set of \emph{transitions} with
  $E \cap P = \varnothing$, $F \subseteq (P \times E) \cup (E \times
  P)$ is its \emph{flow relation}, and $p_0 \in P$ (resp. $p_f \in P$)
  its
  \emph{initial place} (resp. \emph{final place}). 
  A \emph{marking} of $N$ is a multiset $\mmap\in \multiset{P}$. 
  For a marking $\mmap$ and a place $p$ we say that there
  are $\mmap(p)$ \emph{tokens} on $p$. Corresponding to the initial (resp.\ final) 
  place we have the initial marking $\mmap_0=\multi{p_0}$ 
  (resp.\ final marking $\mmap_f=\multi{p_f}$).
  The \emph{size} of $N$ is defined as $|N| = |P|+|E|$.
  
  A transition $e \in E$ is enabled at a marking $\mmap$ if $ \set{p \mid (p,e) \in F} \preceq \mmap $. If $e$ is enabled in $\mmap$, $e$ can be fired, which leads to a marking $\mmap'$ with $\mmap' = \mmap \oplus \set{p \mid (e,p) \in F} \ominus \set{p \mid (p,e) \in F}$. In this case we write $\mmap \autstep[e] \mmap'$.
  A marking $\mmap$ is \emph{coverable} in $N$ if there is a sequence
  $\mmap_0 \autstep[e_1] \mmap_1 \autstep[e_2] \ldots \autstep[e_l] \mmap_l$ such that $\mmap \preceq \mmap_l$. We call such a sequence a \emph{run} of $N$.
\end{definition}
The \emph{coverability problem} for Petri nets is defined as:
\begin{description}
  \item[Given] A Petri net $N$.
  \item[Question] Is $\mmap_f$ coverable in $N$?
\end{description}

Note that Petri nets are essentially $\VASS$ without states (and just the operations $1, -1 \in \mathbb{Z}$ on counters). However, since Petri net places can simulate $\VASS$ states and counters, the two models are equivalent. We use $\VASS$ in other parts of the paper because we need the states as a notion of \emph{control}. In contrast we use Petri nets here because \citet{BaumannMajumdarThinniyamZetzsche2020a} did so as well and because they make the next definition simpler:

\begin{definition}
  \label{def:tdpn}
  A \textbf{transducer-defined Petri net} $\mathcal{N} = (w_{
    \mathit{init}},$ $w_{
    \mathit{final}},$ $
  T_{move},$ $T_{fork},$ $T_{join})$
  consists of two words $w_{\mathit{init}},w_{
    \mathit{final}} \in \Sigma^l$ for some $l \in \nats$, a binary
  transducer $
  T_{move}$ and two ternary transducers $T_
  {fork}$ and $T_{join}$. Additionally, all three
  transducers share $\Sigma$ as their alphabet. This
  defines an \emph{explicit} Petri net $\cE(\cN) = 
  (P,E,F,p_0,p_f)$ :
  \begin{itemize}
    \item $P=\Sigma^l$. 
    \item $E$ is the disjoint union of $E_{\mathit{move
    }}$, $E_{\mathit{fork}}$ and $E_{\mathit{join}}$\footnote{Note that a tuple $
      (w,w',w'') \in E_{\mathit{join}}$ is different from the
      same tuple in $E_{\mathit{fork}}$. In the interest of readability,
      we have chosen not to introduce
      a $4^{th}$ coordinate to distinguish the two.} where
    \begin{itemize}
      \item $E_{\mathit{move}} =\{(w,w') \in \Sigma^l \times \Sigma^l \mid 
      (w,w') \in L(T_{\mathit{move}}) \}$,
      \item $E_{\mathit{fork}} =\{(w,w',w'') \in \Sigma^l \times \Sigma^l \times \Sigma^l \mid
      (w,w',w'') \in L(T_{\mathit{fork}}) \}$, and
      \item $E_{
        \mathit{join}}= \{(w,w',w'') \in \Sigma^l \times \Sigma^l \times
      \Sigma^l \mid
      (w,w',w'') \in L(T_{\mathit{join}}) \}$.
    \end{itemize}
    \item $ p_0 = w_{\mathit{init}}$ and $p_f=w_{\mathit{final}}$.
    \item $\forall t \in T\colon$
    \begin{itemize}
      \item If $t = (p_1,p_2) \in E_{\mathit{move}}$ then $(p_1,t),
      (t,p_2) \in
      F$.
      \item If $t = (p_1,p_2,p_3) \in E_{\mathit{fork}}$ then $(p_1,t),
      (t,p_2),
      (t,p_3) \in F$.
      \item If $t = (p_1,p_2,p_3) \in E_{\mathit{join}}$ then $(p_1,t),
      (p_2,t),
      (t,p_3) \in F$.
    \end{itemize}
  \end{itemize}
  An accepting run of one of the transducers, which corresponds to
  a single transition of $\cE(\cN)$, is called a \emph{transducer-move}.
  The \emph{size} of $\mathcal{N}$ is defined as $|\mathcal{N}| = l + |T_{move}| + |T_{fork}| + |T_{join}|$.
\end{definition}

\begin{figure}[t]
  \def \dist{.4}
  \centering
  \begin{tikzpicture}[]
    \node[place,label=left:$p_1$] (p1) at (0,0) {};
    \node[place,label=right:$p_2$] (p2) at (2,0) {};
    \node[transition,label=below:$t$] at (1,0) {}
    edge[pre] (p1)
    edge[post] (p2);
    \node at (1, -1.2) {$t = (p_1,p_2) \in L(T_{move})$};
    
    \node[place,label=left:$p_1$] (q1) at (4.5,0) {};
    \node[place,label=right:$p_2$] (q2) at (6.5,\dist) {};
    \node[place,label=right:$p_3$] (q3) at (6.5,-\dist) {};
    \node[transition,label=below:$t$] at (5.5,0) {}
    edge[pre] (q1)
    edge[post] (q2)
    edge[post] (q3);
    \node at (5.5, -1.2) {$t = (p_1,p_2,p_3) \in L(T_{fork})$};
    
    \node[place,label=left:$p_1$] (r1) at (9,\dist) {};
    \node[place,label=left:$p_2$] (r2) at (9,-\dist) {};
    \node[place,label=right:$p_3$] (r3) at (11,0) {};
    \node[transition,label=below:$t$] at (10,0) {}
    edge[pre] (r1)
    edge[pre] (r2)
    edge[post] (r3);
    \node at (10, -1.2) {$t = (p_1,p_2,p_3) \in L(T_{join})$};
  \end{tikzpicture}
  \caption{The types of transitions defined by the three transducers of a $\TDPN$.}
  \label{fig:PNtransducers}
\end{figure}
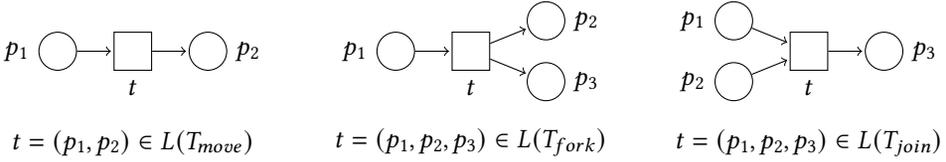

A Petri net defined by transducers in this way can only contain three different types of transitions, each type corresponding to one of the three transducers. These transition types are depicted in \cref{fig:PNtransducers}. 
The \emph{coverability problem} for $\TDPN$ is given by:
\begin{description}
  \item[Given] A $\TDPN$ $\cN$.
  \item[Question] Is $\mmap_f=\multi{w_{\mathit{final}}}$ coverable in
  the
  corresponding explicit
  Petri net $\cE(\cN)$?
\end{description}

Note that $\TDPN$ are quite different from $\TCVASS$ in that for the latter only the control is succinctly defined, while we apply succinctness to the whole model for the former.

Using a reduction from termination for $\rnp$, \citet{BaumannMajumdarThinniyamZetzsche2020a} have shown that the coverability problem is $\TWOEXPSPACE$-hard for $\TDPN$. Like before, we define a more succinct version of $\TDPN$ to show $\THREEEXPSPACE$-hardness:

\begin{definition} \label{def:stdpn}
  A \textbf{succinct transducer-defined Petri net} ($\sTDPN$) $\mathcal{N} = (k,$ $
  T_\mathit{move},$ $T_\mathit{fork},$ $T_\mathit{join})$ consists of a (binary-encoded) number $k \in \N$, a binary
  transducer $T_\mathit{move}$ and two ternary transducers $T_\mathit{fork}$ and $T_\mathit{join}$.
  It is defined in the same way as the $\TDPN$ $(w_{\mathit{init}},$ $w_{\mathit{final}},$
  $T_\mathit{move},$ $T_\mathit{fork},$ $T_\mathit{join})$, where $w_{\mathit{init}} = 0^k$ and $w_{\mathit{final}} = 10^{k - 1}$.
  Additionally, all three
  transducers share $\{0,1\}$ as their alphabet.
  
  The \emph{size} of $\mathcal{N}$ is defined as $|\mathcal{N}| = \lceil\log{k}\rceil + |T_\mathit{move}| + |T_\mathit{fork}| + |T_\mathit{join}|$.
\end{definition}

Given a $2$-$\srnp$ $R$ with maximum recursion depth $\exp^2(k)$ we construct an
$\sTDPN$ $\cN = (k',$ $T_{\mathit{move}},$ $T_{\mathit{fork}},$ $T_{\mathit{join}})$, which defines the Petri net $\cE(\cN) = (P,T,F,p_0,p_f)$, such that $\multi{p_f}$ is coverable in $\cE(\mathcal{N})$ iff there is a terminating execution of $R$.
We can reuse the Petri net construction from \cite{BaumannMajumdarThinniyamZetzsche2020a} to argue about the shape of $\cE(\mathcal{N})$:

The idea is for $\cE(\mathcal{N})$ to have one place per variable of $R$ and one place per label of $R$, as well as at most one auxiliary place for each command of $R$, including an auxiliary place $w_{\mathit{halt}}$ for the \textbf{halt} command. Let the number of all these places be $h$. Then each such place gets copied $\exp^2(k)+1$ times, so that a copy exists for each possible recursion depth. Transitions are added to simulate the commands of $R$ and only act on places of equal or off-by-one recursion depth, due to the semantics of $\srnp$.

Regarding the transducers, we need to identify every place with an address in $\{0,1\}$. To this end, each such address $w = u.v$ shall consist of a prefix $u$ of length $\lceil\log{h}\rceil$ and a postfix $v$ of length $\lceil\log\big(\exp^2(k) + 1\big)\rceil \approx \exp^{1}(k)$. We assign each of the $h$ places that $\cE(\mathcal{N})$ started with (before copying) a number from $0$ to $h-1$. The binary representation of this number (with leading zeros) is used for the $u$-part of its address, and we make sure that the place for the initial label gets an address in $0^*$, whereas the place for the \textbf{halt} command gets an address in $10^*$. For the $v$-part, we use the binary representation of the recursion depth $d$ (also with leading zeros), that a particular copy of this place corresponds to. This way the address of the place of the initial label at recursion depth $0$ matches $w_{\mathit{init}}$ from the $\srnp$-definition, and the one corresponding to the \textbf{halt} command at recursion depth $0$ matches $w_{\mathit{final}}$.

For the $u$-parts of the addresses, the transducers just branch over all possible pairs (in case of move) or triples (fork or join). Then for the $v$-parts, each transducer just needs to check for equality, if all places correspond to the same recursion depth, or for off-by-one, otherwise. A special case are places for labels from a $\text{\texttt{proc}}_{=\text{max}}$-definition, where we also need to check that the $v$-part encodes the maximum recursion depth $\exp^2(k)$. How a binary transducer can perform these checks is depicted in \cref{fig:depthChecks}, and the case of a ternary transducer is very similar. If more than one check needs to be performed on a tuple of addresses, one can just use a product construction.

\begin{figure}[t]
  \centering
  \def \dist{.6cm}
  \begin{tikzpicture}[initial text={}]
    \node[state, initial] (q0) {$q_0$};
    \node[left=\dist of q0] {\textbf{1.)}};
    \node[state, accepting, right=\dist of q0] (q1) {$q_1$};
    \path[->] (q0) edge[loop above] node{$\left( \begin{smallmatrix} 0 \\ 0 \end{smallmatrix} \right), \left( \begin{smallmatrix} 1 \\ 1 \end{smallmatrix} \right)$} (q0)
      (q0) edge node[above]{$\left( \begin{smallmatrix} 0 \\ 1 \end{smallmatrix} \right)$} (q1)
      (q1) edge[loop above] node{$\left( \begin{smallmatrix} 1 \\ 0 \end{smallmatrix} \right)$} (q1);

    \node[right=\dist of q1] (l2) {\textbf{2.)}};
    \node[state, initial, accepting, right=\dist of l2] (r0) {$q_0$};
    \path[->] (r0) edge[loop above] node{$\left( \begin{smallmatrix} 0 \\ 0 \end{smallmatrix} \right), \left( \begin{smallmatrix} 1 \\ 1 \end{smallmatrix} \right)$} (r0);

    \node[right=\dist of r0] (l3) {\textbf{3.)}};
    \node[state, initial, right=\dist of l3] (s0) {$q_0$};
    \node[state, accepting, right=\dist of s0] (s1) {$q_1$};
    \path[->] (s0) edge node[above]{$\left( \begin{smallmatrix} 1 \\ 1 \end{smallmatrix} \right)$} (s1)
      (s1) edge[loop above] node{$\left( \begin{smallmatrix} 0 \\ 0 \end{smallmatrix} \right)$} (s1);
  \end{tikzpicture}
  \caption{How to perform different checks on pairs of addresses over $\{0,1\}$ using a binary transducer: \textbf{1.)} off-by-one, \textbf{2.)} equality, \textbf{3.)} whether one address encodes the maximum recursion depth $\exp^2(k)$, which for addresses of length $\lceil\log\big(\exp^2(k) + 1\big)\rceil$ is equivalent to checking membership in $10^*$.}
  \label{fig:depthChecks}
\end{figure}
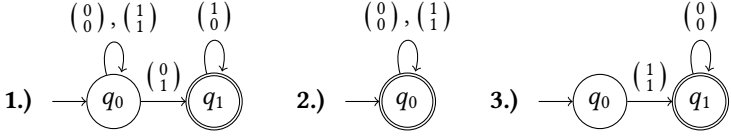

Since it is clear from the $u$-parts, whether the recursion depths should all match or not and whether we need to check for maximum recursion depth, we can just connect the correct $u$-branch to the appropriate checks on $v$ in the transducer, which then lead to a final state. Of course we only connect those $u$-branches that actually correspond to transitions in $\cE(\cN)$, which is feasible since we distinguish all possibilities.

The branches for the $u$-parts in the transducers are of length $\lceil\log{h}\rceil$ and need to distinguish between up to $2^{3\lceil\log{h}\rceil} \leq 2^{3(\log{h} + 1)} = 8h^3$ possibilities, meaning they require polynomially in $h$ many states. Since $\TDPN$s and by extension $\sTDPN$s only consider addresses of a fixed length $k'$ by definition, the checks on the $v$-parts can be performed by making use of cycles in the transducers and require only a constant number of states (see \cref{fig:depthChecks}).
Since $h$ is linear in the number of commands in $R$, the transducers are of polynomial size compared to $R$. We choose $k' = \lceil\log{h}\rceil + \lceil\log\big(\exp^2(k) + 1\big)\rceil$, which is exponential in $k + h$. With $k'$ being encoded in binary this means $\mathcal{N}$ is of polynomial size compared to $R$. Furthermore, the construction is also feasible in polynomial time, because we do not need to construct the whole Petri net $\cE(\mathcal{N})$ first, one copy of each of the $h$ places (for only a single recursion depth) suffices. Thus, the coverability problem for $\sTDPN$ is $\THREEEXPSPACE$-hard.

\subsection{Reduction to $\DCPS$}

\subsubsection*{Idea} Given an $\sTDPN$ $\cN = (k,$ $T_\mathit{move},$ $T_\mathit{fork},$ $T_\mathit{join})$, we want capture every possible marking of the Petri net $\cE(\cN)$ as a configuration of the $\DCPS$ $\cD$ that we construct. Note that since there is no bound on the thread pool of $\cD$ we always have new thread available to execute any given task.

The general idea for this is essentially the same as in \cite{BaumannMajumdarThinniyamZetzsche2020a}: For a marking $\mmap$ of $\cE(\cN)$ and each place $w \in \Sigma^k$ the corresponding configuration of $\cD$ contains exactly $\mmap(w)$ many tasks that all have stack content $w$. However, there are several challenges not handled by the previous approach, all stemming from the fact that we now have exponentially big stack contents, but only polynomially many global states:

Firstly, we need to initially create one task that has the stack content $w_\mathit{init} = 0^{k}$. This is done by using the stack to simulate a context-free grammar of polynomial size that produces just this one word. Secondly, to properly reduce to context-bounded reachability with global state $\mathit{halt}$ as input, we should be able to reach this state iff a task with stack content $w_\mathit{final} = 10^{k-1}$ can occur. For this we just make sure that stack contents always have exactly length $k$, so that this reduces to checking membership in the regular language $10^*$, for which we can just encode a finite automaton in the global states.

Finally, we need to simulate the transitions defined by the three transducers. To this end, we encode all three of them in the global states, and whenever one of their transitions would be fired, we simulate their corresponding run step-by-step. For each step, we context switch to the 1-2 tasks corresponding to the transition inputs and pop a symbol each, then we context switch to 1-2 newly created tasks corresponding to the outputs and push a symbol each. After $k$ context switches per task we have then emptied the stacks of the input tasks and built-up the stacks of the output tasks. Since we need to be able to use each task for both output at first, and input later, our context switch bound is $2k$, which is of bit-length polynomial in the size of $\cN$.

\subsubsection*{Initialization}

We construct a part of our $\DCPS$ $\cD$ that initially creates a task with stack content $0^k\bot$ and context switch number $\leq k$. With our bound of $2k$, this context switch number makes sure that we still have one switch per stack symbol remaining. This does not include the symbol $\bot$, which just serves as a marker for the bottom of the stack.

We define $(\mathit{init},0)$ and $\bot$ to be the initial state and initial stack symbol of $\cD$, respectively. Let $a_0\ldots a_{\ell-1} \in \{0,1\}^\ell$ be the binary, most significant bit first encoding of $k$. At first, $\cD$ creates the stack content $\tilde{B}_{\ell-1}\ldots\tilde{B}_0\bot$ for the initial task, where $\tilde{B}_i = B_i$ if $a_i = 1$ and $\tilde{B}_i = \varepsilon$ otherwise. This uses the following transition rules:
\newcounter{saveEnum}
\begin{enumerate}
  \item $(\mathit{B},0) \mapsto (\mathit{B},0) \lhd \bot \in E_\mathsf{r}$,
  \item $(\mathit{B},i)|\gamma \hookrightarrow (\mathit{B},i+1)|\tilde{B}_i\gamma \in E_\mathsf{c}$ for all $i \in \{0,\ldots,\ell-2\}$, and all $\gamma \in \{B_0,\ldots,B_{\ell-1},\bot\}$, and
  \item $(\mathit{B},\ell-1)|\gamma \hookrightarrow (\mathit{init},\bar{0})|\tilde{B}_{\ell-1}\gamma \triangleright \bot \in E_\mathsf{c}$. \label{rule:createWinit}
  \setcounter{saveEnum}{\value{enumi}}
\end{enumerate}
Rule (\ref{rule:createWinit}) creates a new task with symbol $\bot$, and this task will be used to build the stack content of $0^k\bot$. To this end, each $B_i$ on the stack of the initial task represents $2^i$ many $0$s, for a total of $k$ many, that all need to be pushed to the new task. We facilitate this by introducing rules that replace each $B_{i+1}$ by $B_iB_i$ on the stack, with the rule for $B_0$ replacing it by just $0$. This is essentially simulating a context free grammar with nonterminals $B_0$ to $B_{\ell-1}$. Furthermore, there are rules to transfer a single $0$ from one task's stack to the other, and then switching back. Finally, the initial task needs to terminate upon reaching~$\bot$:
\begin{enumerate}
  \setcounter{enumi}{\value{saveEnum}}
  \item $(\mathit{init},\bar{0})|B_i \hookrightarrow (\mathit{init},\bar{0})|B_{i+1}B_{i+1} \in E_\mathsf{c}$ for all $i \in \{1,\ldots,\ell-1\}$, \label{rule:grammar1}
  \item $(\mathit{init},\bar{0})|B_0 \hookrightarrow (\mathit{init},\bar{0})|0 \in E_\mathsf{c}$, \label{rule:grammar2}
  \item $(\mathit{init},\bar{0})|0 \hookrightarrow (\mathit{init},\varepsilon)|\varepsilon \in E_\mathsf{c}$, \label{rule:pop0}
  \item $(\mathit{init},\varepsilon)|\gamma \mapsto (\mathit{init},0)|\gamma \in E_\mathsf{i}$ for all $\gamma \in \{B_0,\ldots,B_{\ell-1}\}$,
  \item $(\mathit{init},0) \mapsto (\mathit{init},0) \lhd \gamma \in E_\mathsf{r}$ for all $\gamma \in \{0,\bot\}$,
  \item $(\mathit{init},0)|\gamma \mapsto (\mathit{init},\bar{0})|0\gamma \in E_\mathsf{i}$ for all $\gamma \in \{0,\bot\}$,
  \item $(\mathit{init},\bar{0}) \mapsto (\mathit{init},\bar{0}) \lhd \gamma \in E_\mathsf{r}$ for all $\gamma \in \{B_0,\ldots,B_{\ell-1}\}$, \label{rule:switchBack}
  \item $(\mathit{init},\varepsilon)|\bot \hookrightarrow (\mathit{init},\varepsilon)|\varepsilon \in E_\mathsf{c}$, \label{rule:popBot}
  \item $(\mathit{init},\varepsilon) \mapsto (\mathit{main},0) \in E_\mathsf{t}$,
  \item $(\mathit{main},0) \mapsto (\mathit{main},0) \lhd \gamma \in E_\mathsf{r}$ for all $\gamma \in \{0,\bot\}$, and
  \item $(\mathit{main},0)|\gamma \hookrightarrow \mathit{main}|0\gamma \in E_\mathsf{c}$ for all $\gamma \in \{0,\bot\}$. \label{rule:finishInit}
  \setcounter{saveEnum}{\value{enumi}}
\end{enumerate}
Rules (\ref{rule:grammar1}-\ref{rule:grammar2}) replace the $B_i$ on the stack, (\ref{rule:pop0}-\ref{rule:switchBack}) transfer a $0$ from the initial to the new task and then switch back to the former, and (\ref{rule:popBot}-\ref{rule:finishInit}) terminate the initial task and push the last $0$ to the new task. The resume rules correctly distinguish between the two tasks by choosing the one with a $B_i$ on top or the one with a $0$ or $\bot$. Upon finishing this process, $\cD$ ends up in a state $\mathit{main}$, from which the rest of the simulation continues.

\subsubsection*{Simulating Transducer-Moves}

A transducer-move removes a token from up to two places and additionally puts a token on up to two places. We call the former its inputs and the latter its outputs. To simulate this our $\DCPS$ $\cD$ first creates as many new tasks as there are output tokens, which are initially marked with $\bot_\mathit{out}$ and $\widetilde{\bot}_\mathit{out}$, respectively. Then it switches between all tasks for the input tokens and pops a symbol each to figure out a suitable transition in the transducer. With such a transition found it then switches between all tasks for the output tokens and pushes the corresponding symbols. This repeats until the input token tasks have emptied their stack. Since those stacks initially contained $k$ symbols, this uses exactly $k$ context switches for each involved task. The tasks involved in this procedure are marked with $\top_\mathit{in}$, $\widetilde{\top}_\mathit{in}$, $\top_\mathit{out}$ and $\widetilde{\top}_\mathit{out}$, respectively on top of their stacks to correctly identify them throughout. Finally, the current state in the run of the transducer is kept in the $\cD$'s global state to check that a final transducer state is reached at the end (otherwise the $\DCPS$ does not continue).

Let $T_\mathit{move} = ( Q_\mathit{move}, \{0,1\}, q_{\mathit{move}0} , F_\mathit{move}, E_\mathit{move} )$. The following transition rules simulate a transducer move involving $T_\mathit{move}$:
\begin{enumerate}
  \item $\mathit{main} \mapsto \mathit{main} \lhd \gamma \in E_\mathsf{r}$ for all $\gamma \in \{0,1\}$, \label{rule:maincs}
  \item $\mathit{main}|\gamma \hookrightarrow (q_{\mathit{move}0},\#,\#)|\top_\mathit{in}\gamma \triangleright \bot_\mathit{out} \in E_\mathsf{c}$ for all $\gamma \in \{0,1\}$, \label{rule:move1}
  \item $(q,\#,\#)|\top_\mathit{in} \hookrightarrow (q,\#,\#)|\varepsilon \in E_\mathsf{c}$ for all $q \in Q_\mathit{move}$, \label{rule:moveRead1}
  \item $(q,\#,\#)|\gamma \hookrightarrow (q,\gamma,\#)|\varepsilon \in E_\mathsf{c}$ for all $q \in Q_\mathit{move}$, and all $\gamma \in \{0,1\}$, \label{rule:moveRead2}
  \item $(q,\gamma_1,\#)|\gamma_2 \mapsto (q,\gamma_1,\#)|\top_\mathit{in}\gamma_2 \in E_\mathsf{i}$ for all $q \in Q_\mathit{move}$, and all $\gamma_1,\gamma_2 \in \{0,1\}$, \label{rule:moveNotBot}
  \item $(q,\gamma_1,\#) \mapsto (q',\#,\gamma_2) \lhd \bot_\mathit{out} \in E_\mathsf{r}$ for all $(q,\gamma_1,\gamma_2,q') \in E_\mathit{move}$, \label{rule:movePick1}
  \item $(q,\gamma_1,\#) \mapsto (q',\#,\gamma_2) \lhd \top_\mathit{out} \in E_\mathsf{r}$ for all $(q,\gamma_1,\gamma_2,q') \in E_\mathit{move}$, \label{rule:movePick2}
  \item $(q,\#,\gamma)|\bot_\mathit{out} \hookrightarrow (q,\#,\gamma)|\top_\mathit{out}\bot \in E_\mathsf{c}$ for all $q \in Q_\mathit{move}$, and all $\gamma \in \{0,1\}$, \label{rule:moveMarkOut}
  \item $(q,\#,\gamma)|\top_\mathit{out} \hookrightarrow (q,\#,\#)|\top_\mathit{out}\gamma \in E_\mathsf{c}$ for all $q \in Q_\mathit{move}$, and all $\gamma \in \{0,1\}$, \label{rule:moveWrite}
  \item $(q,\#,\#)|\top_\mathit{out} \mapsto (q,\#,\#)|\top_\mathit{out} \in E_\mathsf{i}$ for all $q \in Q_\mathit{move}$, \label{rule:moveBack1}
  \item $(q,\#,\#) \mapsto (q,\#,\#) \lhd \top_\mathit{in} \in E_\mathsf{r}$ for all $q \in Q_\mathit{move}$,  \label{rule:moveBack2}
  \item $(q,\gamma,\#)|\bot \hookrightarrow (\mathit{end},q,\gamma,\#)|\varepsilon \in E_\mathsf{c}$ for all $q \in Q_\mathit{move}$, and all $\gamma \in \{0,1\}$, \label{rule:moveEnd1}
  \item $(\mathit{end},q,\gamma,\#) \mapsto (\mathit{end},q,\gamma,\#) \in E_\mathsf{t}$ for all $q \in Q_\mathit{move}$, and all $\gamma \in \{0,1\}$, \label{rule:moveEnd2}
  \item $(\mathit{end},q,\gamma_1,\#) \mapsto (\mathit{end},q_f,\#,\gamma_2) \lhd \top_\mathit{out} \in E_\mathsf{r}$ for all $(q,\gamma_1,\gamma_2,q_f) \in E_\mathit{move}$ with $q_f \in F_\mathit{move}$, \label{rule:moveEnd3}
  \item $(\mathit{end},q_f,\#,\gamma)|\top_\mathit{out} \hookrightarrow (\mathit{end},q_f,\#,\#)|\top_\mathit{out}\gamma \in E_\mathsf{c}$ for all $q_f \in F_\mathit{move}$, and all $\gamma \in \{0,1\}$, and \label{rule:moveEnd4}
  \item $(\mathit{end},q_f,\#,\#)|\top_\mathit{out} \mapsto \mathit{rev}1|\top_\mathit{in} \in E_\mathsf{i}$ for all $q_f \in F_\mathit{move}$. \label{rule:moveFinal}
\end{enumerate}
Rule (\ref{rule:maincs}) is shared for all three transducers and makes sure that we begin in $\mathit{main}$ with an active task, (\ref{rule:move1}) starts the simulation for $T_\mathit{move}$ in $q_{\mathit{move}0}$ as well as marking the input token task with $\top_\mathit{in}$ and spawning the output token task, (\ref{rule:moveRead1}-\ref{rule:moveRead2}) remove a symbol from the input token task and save it in the global state, (\ref{rule:moveNotBot}) switches out the input token task if its stack is not yet empty, (\ref{rule:movePick1}-\ref{rule:movePick2}) simulate a suitable transition of $E_\mathit{move}$ in the global state and switch to the output token task, (\ref{rule:moveMarkOut}) marks the output token task with $\top_\mathit{out}$ and gives it a proper bottom-of-stack symbol $\bot$, (\ref{rule:moveWrite}) transfers a symbol from the global state to the output token task, (\ref{rule:moveBack1}-\ref{rule:moveBack2}) switch back to the input token task, (\ref{rule:moveEnd1}-\ref{rule:moveEnd2}) begin the end of the simulation upon reaching the bottom of the stack for the input token task and terminating said task, (\ref{rule:moveEnd3}-\ref{rule:moveEnd4}) check that the simulated run of $T_\mathit{move}$ has reached a final state and transfer the last symbol to the output token task, and (\ref{rule:moveFinal}) ends the simulation for $T_\mathit{move}$ by moving to global state $\mathit{rev1}$ and marking the output token task with now $\top_\mathit{in}$ on top of the stack.

During the simulation for $T_\mathit{move}$ we read symbols top-down from the input token task's stack, but write them bottom up to the output token task's stack. This means that in the end the order of symbols is the reverse of what it should be, which is why the global state $\mathit{rev1}$ is used to signify that one stack needs to be reversed. Because of this, we also used the marker $\top_\mathit{in}$ at the end, as the corresponding task will be used as input for the reversing process. We consider said process another step in the overall $\TDPN$-simulation and explain it below. For the moment, we continue with the simulation of the transducer-moves involving $T_\mathit{fork}$ and $T_\mathit{join}$.

Let $T_\mathit{fork} = ( Q_\mathit{fork}, \{0,1\}, q_{\mathit{fork}0} , F_\mathit{fork}, E_\mathit{fork} )$. To not cause any overlap with $T_\mathit{move}$, we assume $Q_\mathit{move} \cap Q_\mathit{fork} = \emptyset$. The following transition rules simulate a transducer move involving $T_\mathit{fork}$:
\begin{enumerate}
  \setcounter{enumi}{1}
  \item $\mathit{main}|\gamma \hookrightarrow (\mathit{start},q_{\mathit{fork}0},\#,\#,\#)|\top_\mathit{in}\gamma \triangleright \bot_\mathit{out} \in E_\mathsf{c}$,
  \item $(\mathit{start},q_{\mathit{fork}0},\#,\#,\#)|\top_\mathit{in} \hookrightarrow (q_{\mathit{fork}0},\#,\#,\#)|\top_\mathit{in} \triangleright \widetilde{\bot}_\mathit{out} \in E_\mathsf{c}$,
  \item $(q,\#,\#,\#)|\top_\mathit{in} \hookrightarrow (q,\#,\#,\#)|\varepsilon \in E_\mathsf{c}$ for all $q \in Q_\mathit{fork}$,
  \item $(q,\#,\#,\#)|\gamma \hookrightarrow (q,\gamma,\#,\#)|\varepsilon \in E_\mathsf{c}$ for all $q \in Q_\mathit{fork}$, and all $\gamma \in \{0,1\}$,
  \item $(q,\gamma_1,\#,\#)|\gamma_2 \mapsto (q,\gamma_1,\#,\#)|\top_\mathit{in}\gamma_2 \in E_\mathsf{i}$ for all $q \in Q_\mathit{fork}$, and all $\gamma_1,\gamma_2 \in \{0,1\}$,
  \item $(q,\gamma_1,\#,\#) \mapsto (q',\#,\gamma_2,\gamma_3) \lhd \bot_\mathit{out} \in E_\mathsf{r}$ for all $(q,\gamma_1,\gamma_2,\gamma_3,q') \in E_\mathit{fork}$,
  \item $(q,\gamma_1,\#,\#) \mapsto (q',\#,\gamma_2,\gamma_3) \lhd \top_\mathit{out} \in E_\mathsf{r}$ for all $(q,\gamma_1,\gamma_2,\gamma_3,q') \in E_\mathit{fork}$,
  \item $(q,\#,\gamma_2,\gamma_3)|\bot_\mathit{out} \hookrightarrow (q,\#,\gamma_2,\gamma_3)|\top_\mathit{out}\bot \in E_\mathsf{c}$ for all $q \in Q_\mathit{fork}$, and all $\gamma_2,\gamma_3 \in \{0,1\}$,
  \item $(q,\#,\gamma_2,\gamma_3)|\top_\mathit{out} \hookrightarrow (q,\#,\#,\gamma_3)|\top_\mathit{out}\gamma_2 \in E_\mathsf{c}$ for all $q \in Q_\mathit{fork}$, and all $\gamma_2,\gamma_3 \in \{0,1\}$,
  \item $(q,\#,\#,\gamma)|\top_\mathit{out} \mapsto (q,\#,\#,\gamma)|\top_\mathit{out} \in E_\mathsf{i}$ for all $q \in Q_\mathit{fork}$, and all $\gamma \in \{0,1\}$,
  \item $(q,\#,\#,\gamma) \mapsto (q,\#,\#,\gamma) \lhd \widetilde{\bot}_\mathit{out} \in E_\mathsf{r}$ for all $q \in Q_\mathit{fork}$, and all $\gamma \in \{0,1\}$,
  \item $(q,\#,\#,\gamma) \mapsto (q,\#,\#,\gamma) \lhd \widetilde{\top}_\mathit{out} \in E_\mathsf{r}$ for all $q \in Q_\mathit{fork}$, and all $\gamma \in \{0,1\}$,
  \item $(q,\#,\#,\gamma)|\widetilde{\bot}_\mathit{out} \hookrightarrow (q,\#,\#,\gamma)|\widetilde{\top}_\mathit{out}\bot \in E_\mathsf{c}$ for all $q \in Q_\mathit{fork}$, and all $\gamma \in \{0,1\}$,
  \item $(q,\#,\#,\gamma)|\widetilde{\top}_\mathit{out} \hookrightarrow (q,\#,\#,\#)|\widetilde{\top}_\mathit{out}\gamma \in E_\mathsf{c}$ for all $q \in Q_\mathit{fork}$, and all $\gamma \in \{0,1\}$,
  \item $(q,\#,\#,\#)|\widetilde{\top}_\mathit{out} \mapsto (q,\#,\#,\#)|\widetilde{\top}_\mathit{out} \in E_\mathsf{i}$ for all $q \in Q_\mathit{fork}$,
  \item $(q,\#,\#,\#) \mapsto (q,\#,\#,\#) \lhd \top_\mathit{in} \in E_\mathsf{r}$ for all $q \in Q_\mathit{fork}$,
  \item $(q,\gamma,\#,\#)|\bot \hookrightarrow (\mathit{end},q,\gamma,\#,\#)|\varepsilon \in E_\mathsf{c}$ for all $q \in Q_\mathit{fork}$, and all $\gamma \in \{0,1\}$,
  \item $(\mathit{end},q,\gamma,\#,\#) \mapsto (\mathit{end},q,\gamma,\#,\#) \in E_\mathsf{t}$ for all $q \in Q_\mathit{fork}$, and all $\gamma \in \{0,1\}$,
  \item $(\mathit{end},q,\gamma_1,\#,\#) \mapsto (\mathit{end},q_f,\#,\gamma_2,\gamma_3) \lhd \top_\mathit{out} \in E_\mathsf{r}$ for all $(q,\gamma_1,\gamma_2,\gamma_3,q_f) \in E_\mathit{fork}$ with $q_f \in F_\mathit{fork}$,
  \item $(\mathit{end},q_f,\#,\gamma_2,\gamma_3)|\top_\mathit{out} \hookrightarrow (\mathit{end},q_f,\#,\#,\gamma_3)|\top_\mathit{out}\gamma_2 \in E_\mathsf{c}$ for all $q_f \in F_\mathit{fork}$, and all $\gamma_2,\gamma_3 \in \{0,1\}$,
  \item $(\mathit{end},q_f,\#,\#,\gamma)|\top_\mathit{out} \mapsto (\mathit{end},q_f,\#,\#,\gamma)|\top_\mathit{in} \in E_\mathsf{i}$ for all $q_f \in F_\mathit{fork}$, and all $\gamma \in \{0,1\}$,
  \item $(\mathit{end},q_f,\#,\#,\gamma) \mapsto (\mathit{end},q_f,\#,\#,\gamma) \lhd \widetilde{\top}_\mathit{out} \in E_\mathsf{r}$ for all $q_f \in F_\mathit{fork}$, and all $\gamma \in \{0,1\}$,
  \item $(\mathit{end},q_f,\#,\#,\gamma)|\widetilde{\top}_\mathit{out} \hookrightarrow (\mathit{end},q_f,\#,\#,\#)|\widetilde{\top}_\mathit{out}\gamma \in E_\mathsf{c}$ for all $q_f \in F_\mathit{fork}$, and all $\gamma \in \{0,1\}$, and
  \item $(\mathit{end},q_f,\#,\#,\#)|\widetilde{\top}_\mathit{out} \mapsto \mathit{rev}2|\widetilde{\top}_\mathit{in} \in E_\mathsf{i}$ for all $q_f \in F_\mathit{fork}$.
\end{enumerate}
Rule (\ref{rule:maincs}) is shared with $T_\mathit{move}$ and therefore does not need to be defined again. The main difference to the simulation for $T_\mathit{move}$ is that we need to spawn an additional output token task (marked with $\widetilde{\bot}_\mathit{out}$ and later $\widetilde{\top}_\mathit{out}$), which we also need to switch to and transfer a symbol to for each $E_\mathit{fork}$-transition that we simulate. We use $\mathit{start}$ in the global state to remember that this task still needs to be spawned at the beginning. Having two output token tasks also means that at the end there are two tasks marked with $\top_\mathit{in}$ and $\widetilde{\top}_\mathit{in}$, respectively, whose stacks need to be reversed. We use global state $\mathit{rev}2$ to signify that the reversing process needs to be performed twice.

Let $T_\mathit{join} = ( Q_\mathit{join}, \{0,1\}, q_{\mathit{join}0} , F_\mathit{join}, E_\mathit{join} )$. To not cause any overlap with the other transducers, we assume $Q_\mathit{move} \cap Q_\mathit{join} = \emptyset = Q_\mathit{fork} \cap Q_\mathit{join}$. The following transition rules simulate a transducer move involving $T_\mathit{join}$:
\begin{enumerate}
  \setcounter{enumi}{1}
  \item $\mathit{main}|\gamma \hookrightarrow (\mathit{start},q_{\mathit{join}0},\gamma,\#,\#)|\top_\mathit{in} \triangleright \bot_\mathit{out} \in E_\mathsf{c}$ for all $\gamma \in \{0,1\}$,
  \item $(\mathit{start},q_{\mathit{join}0},\gamma,\#,\#)|\top_\mathit{in} \mapsto (\mathit{start},q_{\mathit{join}0},\gamma,\#,\#)|\top_\mathit{in} \in E_\mathsf{i}$ for all $\gamma \in \{0,1\}$,
  \item $(\mathit{start},q_{\mathit{join}0},\gamma_1,\#,\#) \mapsto (\mathit{start},q_{\mathit{join}0},\gamma_1,\#,\#) \lhd \gamma_2 \in E_\mathsf{r}$ for all $\gamma_1,\gamma_2 \in \{0,1\}$, \label{rule:joinPickInput2}
  \item $(\mathit{start},q_{\mathit{join}0},\gamma_1,\#,\#)|\gamma_2 \hookrightarrow (q_{\mathit{join}0},\gamma_1,\#,\#)|\widetilde{\top}_\mathit{in} \in E_\mathsf{c}$ for all $\gamma_1,\gamma_2 \in \{0,1\}$,
  \item $(q,\#,\#,\#)|\top_\mathit{in} \hookrightarrow (q,\#,\#,\#)|\varepsilon \in E_\mathsf{c}$ for all $q \in Q_\mathit{join}$,
  \item $(q,\#,\#,\#)|\gamma \hookrightarrow (q,\gamma,\#,\#)|\varepsilon \in E_\mathsf{c}$ for all $q \in Q_\mathit{join}$, and all $\gamma \in \{0,1\}$,
  \item $(q,\gamma_1,\#,\#)|\gamma_2 \mapsto (\mathit{next},q,\gamma_1,\#,\#)|\top_\mathit{in}\gamma_2 \in E_\mathsf{i}$ for all $q \in Q_\mathit{join}$, and all $\gamma_1,\gamma_2 \in \{0,1\}$, \label{rule:joinNext}
  \item $(\mathit{next},q,\gamma,\#,\#) \mapsto (\mathit{next},q,\gamma,\#,\#) \lhd \widetilde{\top}_\mathit{in} \in E_\mathsf{r}$ for all $q \in Q_\mathit{join}$, and all $\gamma \in \{0,1\}$,
  \item $(\mathit{next},q,\gamma,\#,\#)|\widetilde{\top}_\mathit{in} \hookrightarrow (\mathit{next},q,\gamma,\#,\#)|\varepsilon \in E_\mathsf{c}$ for all $q \in Q_\mathit{join}$, and all $\gamma \in \{0,1\}$,
  \item $(\mathit{next},q,\gamma_1,\#,\#)|\gamma_2 \hookrightarrow (q,\gamma_1,\gamma_2,\#)|\varepsilon \in E_\mathsf{c}$ for all $q \in Q_\mathit{join}$, and all $\gamma_1,\gamma_2 \in \{0,1\}$, \label{rule:joinRead2}
  \item $(q,\gamma_1,\gamma_2,\#)|\gamma_3 \mapsto (q,\gamma_1,\gamma_2,\#)|\widetilde{\top}_\mathit{in}\gamma_3 \in E_\mathsf{i}$ for all $q \in Q_\mathit{join}$, and all $\gamma_1,\gamma_2,\gamma_3 \in \{0,1\}$,
  \item $(q,\gamma_1,\gamma_2,\#) \mapsto (q',\#,\#,\gamma_3) \lhd \bot_\mathit{out} \in E_\mathsf{r}$ for all $(q,\gamma_1,\gamma_2,\gamma_3,q') \in E_\mathit{join}$,
  \item $(q,\gamma_1,\gamma_2,\#) \mapsto (q',\#,\#,\gamma_3) \lhd \top_\mathit{out} \in E_\mathsf{r}$ for all $(q,\gamma_1,\gamma_2,\gamma_3,q') \in E_\mathit{join}$,
  \item $(q,\#,\#,\gamma)|\bot_\mathit{out} \hookrightarrow (q,\#,\#,\gamma)|\top_\mathit{out}\bot \in E_\mathsf{c}$ for all $q \in Q_\mathit{join}$, and all $\gamma \in \{0,1\}$,
  \item $(q,\#,\#,\gamma)|\top_\mathit{out} \hookrightarrow (q,\#,\#,\#)|\top_\mathit{out}\gamma \in E_\mathsf{c}$ for all $q \in Q_\mathit{join}$, and all $\gamma \in \{0,1\}$,
  \item $(q,\#,\#,\#)|\top_\mathit{out} \mapsto (q,\#,\#,\#)|\top_\mathit{out} \in E_\mathsf{i}$ for all $q \in Q_\mathit{join}$,
  \item $(q,\#,\#,\#) \mapsto (q,\#,\#,\#) \lhd \top_\mathit{in} \in E_\mathsf{r}$ for all $q \in Q_\mathit{join}$,
  \item $(q,\gamma,\#,\#)|\bot \hookrightarrow (\mathit{end},q,\gamma,\#,\#)|\varepsilon \in E_\mathsf{c}$ for all $q \in Q_\mathit{join}$, and all $\gamma \in \{0,1\}$,
  \item $(\mathit{end},q,\gamma,\#,\#) \mapsto (\mathit{end},q,\gamma,\#,\#) \in E_\mathsf{t}$ for all $q \in Q_\mathit{join}$, and all $\gamma \in \{0,1\}$,
  \item $(\mathit{end},q,\gamma,\#,\#) \mapsto (\mathit{end},q,\gamma,\#,\#) \lhd \widetilde{\top}_\mathit{in} \in E_\mathsf{r}$ for all $q \in Q_\mathit{join}$, and all $\gamma \in \{0,1\}$,
  \item $(\mathit{end},q,\gamma,\#,\#)|\widetilde{\top}_\mathit{in} \hookrightarrow (\mathit{end},q,\gamma,\#,\#)|\varepsilon \in E_\mathsf{c}$ for all $q \in Q_\mathit{join}$, and all $\gamma \in \{0,1\}$,
  \item $(\mathit{end},q,\gamma_1,\#,\#)|\gamma_2 \hookrightarrow (\mathit{end},q,\gamma_1,\gamma_2,\#)|\varepsilon \in E_\mathsf{c}$ for all $q \in Q_\mathit{join}$, and all $\gamma_1,\gamma_2 \in \{0,1\}$,
  \item $(\mathit{end},q,\gamma_1,\gamma_2,\#)|\bot \hookrightarrow (\mathit{end},q,\gamma_1,\gamma_2,\#)|\varepsilon \in E_\mathsf{c}$ for all $q \in Q_\mathit{join}$, and all $\gamma_1,\gamma_2 \in \{0,1\}$,
  \item $(\mathit{end},q,\gamma_1,\gamma_2,\#) \mapsto (\mathit{end},q,\gamma_1,\gamma_2,\#) \in E_\mathsf{t}$ for all $q \in Q_\mathit{join}$, and all $\gamma \in \{0,1\}$,
  \item $(\mathit{end},q,\gamma_1,\gamma_2,\#) \mapsto (\mathit{end},q_f,\#,\#,\gamma_3) \lhd \top_\mathit{out} \in E_\mathsf{r}$ for all $(q,\gamma_1,\gamma_2,\gamma_3,q_f) \in E_\mathit{join}$ with $q_f \in F_\mathit{join}$,
  \item $(\mathit{end},q_f,\#,\#,\gamma)|\top_\mathit{out} \hookrightarrow (\mathit{end},q_f,\#,\#,\#)|\top_\mathit{out}\gamma \in E_\mathsf{c}$ for all $q_f \in F_\mathit{join}$, and all $\gamma \in \{0,1\}$, and
  \item $(\mathit{end},q_f,\#,\#,\#)|\top_\mathit{out} \mapsto \mathit{rev}1|\top_\mathit{in} \in E_\mathsf{i}$ for all $q_f \in F_\mathit{join}$.
\end{enumerate}
Again, Rule (\ref{rule:maincs}) has already been defined. Since $T_\mathit{join}$ defines transitions with two input tokens, we need to mark two tasks with $\top_\mathit{in}$ and $\widetilde{\top}_\mathit{in}$ in the beginning. Our $\DCPS$ remembers that a second input task needs to be marked via the $\mathit{start}$ in the global state, while it already reads the first symbol on the stack of the first input token task, to make sure that each context switch corresponds to one symbol read or written. The $\mathit{next}$ in the global state is to remember that the first input task has been already handled, which makes sure that rules (\ref{rule:joinNext}) and (\ref{rule:joinRead2}) do not overlap. Other than that this part of the simulation is very similar to the one for $T_\mathit{move}$.

Our simulation ensures that each context switch of a task corresponds to one symbol popped or pushed. In general our tasks experience $k$ context switches when used as input tokens, and another $k$ when used as output tokens later. This results in $2k$ context switches in total, which is exactly our bound. Furthermore all transducer-move simulations assume that $k \geq 2$, which is most apparent in that the final part of each simulation ($\mathit{end}$ in the global state) assumes that the marker $\bot_\mathit{out}$ has already been replaced by $\top_\mathit{out}$ for the first output token task. It is however easy to see that the simulation for $k = 1$ is simple, as the entire transducer move then consists of only $3$ symbols an can therefore be kept in the global state.

\subsubsection*{Reversing Stacks}

The simulation of transducer-moves results in tasks, whose stack contents are reversed. By this we mean that each resulting task is supposed to correspond to a place $w = a_1\ldots a_k \in \Sigma^k$, but it has stack content $a_k\ldots a_1$, which is the reverse of $w$. To remedy this, we can transfer the stack symbols of such a task to a newly spawned task one by one, which reverses the stack content. We call the former the input task of the reversing process, and the latter the output task. After the transducer-move simulation the input task is at $k$ context switches and the reversing process requires another $k$, one per stack symbol, which is in line with our bound of $2k$. Similarly the newly spawned output task will also experience $k$ context switches, meaning it can still context switch $k$ more times, which is exactly the amount required for it to be used in a transducer-move simulation later.

The following transition rules reverse the stack content of one task, starting in $\mathit{rev}1$:
\begin{enumerate}
  \item $\mathit{rev}1 \mapsto \mathit{rev}1 \lhd \top_\mathit{in} \in E_\mathsf{r}$ for all $\gamma \in \{0,1\}$, \label{rule:rev1Start}
  \item $\mathit{rev}1|\top_\mathit{in} \hookrightarrow (\mathit{rev}1,\#)|\top_\mathit{in} \triangleright \bot_\mathit{out} \in E_\mathsf{c}$, \label{rule:rev1Spawn}
  \item $(\mathit{rev}1,\#)|\top_\mathit{in} \hookrightarrow (\mathit{rev}1,\#)|\varepsilon \in E_\mathsf{c}$, \label{rule:rev1Read1}
  \item $(\mathit{rev}1,\#)|\gamma \hookrightarrow (\mathit{rev}1,\gamma)|\varepsilon \in E_\mathsf{c}$ for all $\gamma \in \{0,1\}$, \label{rule:rev1Read2}
  \item $(\mathit{rev}1,\gamma_1)|\gamma_2 \mapsto (\mathit{rev}1,\gamma_1)|\top_\mathit{in}\gamma_2 \in E_\mathsf{i}$ for all $\gamma_1,\gamma_2 \in \{0,1\}$, \label{rule:rev1NotBot}
  \item $(\mathit{rev}1,\gamma) \mapsto (\mathit{rev}1,\gamma) \lhd \bot_\mathit{out} \in E_\mathsf{r}$ for all $\gamma \in \{0,1\}$, \label{rule:rev1Switch1}
  \item $(\mathit{rev}1,\gamma) \mapsto (\mathit{rev}1,\gamma) \lhd \top_\mathit{out} \in E_\mathsf{r}$ for all $\gamma \in \{0,1\}$, \label{rule:rev1Switch2}
  \item $(\mathit{rev}1,\gamma)|\bot_\mathit{out} \hookrightarrow (\mathit{rev}1,\gamma)|\top_\mathit{out}\bot \in E_\mathsf{c}$ for all $\gamma \in \{0,1\}$, \label{rule:rev1MarkOut}
  \item $(\mathit{rev}1,\gamma)|\top_\mathit{out} \hookrightarrow (\mathit{rev}1,\#)|\top_\mathit{out}\gamma \in E_\mathsf{c}$ for all $\gamma \in \{0,1\}$, \label{rule:rev1Write}
  \item $(\mathit{rev}1,\#)|\top_\mathit{out} \mapsto (\mathit{rev}1,\#)|\top_\mathit{out} \in E_\mathsf{i}$, \label{rule:rev1Back1}
  \item $(\mathit{rev}1,\#) \mapsto (\mathit{rev}1,\#) \lhd \top_\mathit{in} \in E_\mathsf{r}$ for all $q \in Q_\mathit{fork}$, \label{rule:rev1Back2}
  \item $(\mathit{rev}1,\gamma)|\bot \hookrightarrow (\mathit{end},\mathit{rev}1,\gamma)|\varepsilon \in E_\mathsf{c}$ for all $\gamma \in \{0,1\}$, \label{rule:rev1End1}
  \item $(\mathit{end},\mathit{rev}1,\gamma) \mapsto (\mathit{end},\mathit{rev}1,\gamma) \in E_\mathsf{t}$ for all $\gamma \in \{0,1\}$, \label{rule:rev1End2}
  \item $(\mathit{end},\mathit{rev}1,\gamma) \mapsto (\mathit{end},\mathit{rev}1,\gamma) \lhd \top_\mathit{out} \in E_\mathsf{r}$ for all $\gamma \in \{0,1\}$, \label{rule:rev1End3}
  \item $(\mathit{end},\mathit{rev}1,\gamma)|\top_\mathit{out} \hookrightarrow (\mathit{end},\mathit{rev}1,\#)|\gamma \in E_\mathsf{c}$ for all $\gamma \in \{0,1\}$, \label{rule:rev1End4}
  \item $(\mathit{end},\mathit{rev}1,\#)|\gamma \mapsto \mathit{main}|\gamma \in E_\mathsf{i}$ for all $\gamma \in \{0,1\}$. \label{rule:rev1Final}
\end{enumerate}
Rules (\ref{rule:rev1Start}-\ref{rule:rev1Spawn}) switch to the input task and spawn the output task with $\bot_\mathit{out}$, (\ref{rule:rev1Read1}-\ref{rule:rev1Read2}) remove a symbol from the input task and save it in the global state, (\ref{rule:rev1NotBot}) switches out the input task if its stack is not yet empty, (\ref{rule:rev1Switch1}-\ref{rule:rev1Switch2}) switch to the output task, (\ref{rule:rev1MarkOut}) marks the output task with $\top_\mathit{out}$ and gives it a proper bottom-of-stack symbol $\bot$, (\ref{rule:rev1Write}) transfers a symbol from the global state to the output task, (\ref{rule:rev1Back1}-\ref{rule:rev1Back2}) switch back to the input task, (\ref{rule:moveEnd1}-\ref{rule:moveEnd2}) begin the end of the simulation upon reaching the bottom of the stack for the input token task and terminating said task, (\ref{rule:rev1End1}-\ref{rule:rev1End2}) begin the end of the reversing process upon reaching the bottom of the stack for the input task and terminating said task, (\ref{rule:rev1End3}-\ref{rule:rev1End4}) switch to and transfer the last symbol to the output task, and (\ref{rule:rev1Final}) ends the reversing process by moving to global state $\mathit{main}$ and switching out the output task.

The input task was already marked with $\top_\mathit{in}$ at the end of the transducer-move simulation, and the output task is spawned with $\bot_\mathit{out}$ and later marked with $\top_\mathit{out}$. These markers are used to correctly identify both tasks throughout the reversing process. All in all the transition rules are very similar to the simulation of a transducer-move for $T_\mathit{move}$, except that there is no need to simulate the run of a transducer.

When our $\DCPS$ $\cD$ simulates a transducer-move involving $T_\mathit{join}$ it results in two tasks, marked with $\top_{in}$ and $\widetilde{\top}_{in}$, whose stacks need to be reversed. We define a second set of transition rules starting from $\mathit{rev}2$ to reverse the stack of the latter task. This second reversal process then finishes in state $\mathit{rev}1$, from where the previously defined transition rules are used to reverse the stack of the former task. The rules for the second reversal process are almost the exact same as the ones defined above, with the main difference being that all stack symbols marking the various tasks are now decorated with a $\sim$:
\begin{enumerate}
  \item $\mathit{rev}2 \mapsto \mathit{rev}2 \lhd \widetilde{\top}_\mathit{in} \in E_\mathsf{r}$ for all $\gamma \in \{0,1\}$,
  \item $\mathit{rev}2|\widetilde{\top}_\mathit{in} \hookrightarrow (\mathit{rev}2,\#)|\widetilde{\top}_\mathit{in} \triangleright \widetilde{\bot}_\mathit{out} \in E_\mathsf{c}$,
  \item $(\mathit{rev}2,\#)|\widetilde{\top}_\mathit{in} \hookrightarrow (\mathit{rev}2,\#)|\varepsilon \in E_\mathsf{c}$,
  \item $(\mathit{rev}2,\#)|\gamma \hookrightarrow (\mathit{rev}2,\gamma)|\varepsilon \in E_\mathsf{c}$ for all $\gamma \in \{0,1\}$,
  \item $(\mathit{rev}2,\gamma_1)|\gamma_2 \mapsto (\mathit{rev}2,\gamma_1)|\widetilde{\top}_\mathit{in}\gamma_2 \in E_\mathsf{i}$ for all $\gamma_1,\gamma_2 \in \{0,1\}$,
  \item $(\mathit{rev}2,\gamma) \mapsto (\mathit{rev}2,\gamma) \lhd \widetilde{\bot}_\mathit{out} \in E_\mathsf{r}$ for all $\gamma \in \{0,1\}$,
  \item $(\mathit{rev}2,\gamma) \mapsto (\mathit{rev}2,\gamma) \lhd \widetilde{\top}_\mathit{out} \in E_\mathsf{r}$ for all $\gamma \in \{0,1\}$,
  \item $(\mathit{rev}2,\gamma)|\widetilde{\bot}_\mathit{out} \hookrightarrow (\mathit{rev}2,\gamma)|\widetilde{\top}_\mathit{out}\bot \in E_\mathsf{c}$ for all $\gamma \in \{0,1\}$,
  \item $(\mathit{rev}2,\gamma)|\widetilde{\top}_\mathit{out} \hookrightarrow (\mathit{rev}2,\#)|\widetilde{\top}_\mathit{out}\gamma \in E_\mathsf{c}$ for all $\gamma \in \{0,1\}$,
  \item $(\mathit{rev}2,\#)|\widetilde{\top}_\mathit{out} \mapsto (\mathit{rev}2,\#)|\widetilde{\top}_\mathit{out} \in E_\mathsf{i}$,
  \item $(\mathit{rev}2,\#) \mapsto (\mathit{rev}2,\#) \lhd \widetilde{\top}_\mathit{in} \in E_\mathsf{r}$ for all $q \in Q_\mathit{fork}$,
  \item $(\mathit{rev}2,\gamma)|\bot \hookrightarrow (\mathit{end},\mathit{rev}2,\gamma)|\varepsilon \in E_\mathsf{c}$ for all $\gamma \in \{0,1\}$,
  \item $(\mathit{end},\mathit{rev}2,\gamma) \mapsto (\mathit{end},\mathit{rev}2,\gamma) \in E_\mathsf{t}$ for all $\gamma \in \{0,1\}$,
  \item $(\mathit{end},\mathit{rev}2,\gamma) \mapsto (\mathit{end},\mathit{rev}2,\gamma) \lhd \widetilde{\top}_\mathit{out} \in E_\mathsf{r}$ for all $\gamma \in \{0,1\}$,
  \item $(\mathit{end},\mathit{rev}2,\gamma)|\widetilde{\top}_\mathit{out} \hookrightarrow (\mathit{end},\mathit{rev}2,\#)|\gamma \in E_\mathsf{c}$ for all $\gamma \in \{0,1\}$,
  \item $(\mathit{end},\mathit{rev}2,\#)|\gamma \mapsto \mathit{rev}1|\gamma \in E_\mathsf{i}$ for all $\gamma \in \{0,1\}$.
\end{enumerate}

\subsubsection*{Checking Coverability}

The purpose of our $\DCPS$ $\cD$ is to simulate the $\sTDPN$ $\cN$ and check whether $w_\mathit{final}$ is coverable. To be more specific, we reduce to reachability of the global state $\mathit{halt}$ for $\cD$. This means that the part of $\cD$ that checks whether $w_\mathit{final}$ is covered should move to $\mathit{halt}$ if said check is successful.

Whenever $\cD$ is in global state $\mathit{main}$, it can pick a task to become active (via rule (\ref{rule:maincs}) for simulating a transducer-move involving $T_\mathit{move}$). From there it just needs to check that the stack content is in $10^*$, since $w_\mathit{final} = 10^{k-1}$ and our simulation ensures that all tasks have a stack content of length $k$, not counting the symbol $\bot$ at the bottom. To this end, the check starts with popping a $1$ and then continuously pops $0$s. If $\bot$ is encountered this way, the check was successful and $\cD$ can move to the global state $\mathit{halt}$, ending the whole simulation. This is facilitated by the following transition rules:
\begin{enumerate}
  \item $\mathit{main}|1 \hookrightarrow \mathit{check}|\varepsilon \in E_\mathsf{c}$,
  \item $\mathit{check}|0 \hookrightarrow \mathit{check}|\varepsilon \in E_\mathsf{c}$, and
  \item $\mathit{check}|\bot \hookrightarrow \mathit{halt}|\varepsilon \in E_\mathsf{c}$.
\end{enumerate}

\subsubsection*{Correctness}

From our description of the various steps of the simulation, it is clear that for each run of the $\sTDPN$ $\cN$, there is a run of the $\DCPS$ $\cD$ that corresponds to it in the desired way.

For the other direction, note that most of our simulation steps are deterministic. We made sure that for the majority of global states and top of stack symbols there is only one applicable rule of the $\DCPS$. The exceptions are the following: Firstly, whenever there is a schedule point in global state $\mathit{main}$, any of the inactive threads can be resumed, using rule~(\ref{rule:maincs}) for $T_\mathit{move}$. This is by design, as this guesses a token that can either be removed by a transition of the $\sTDPN$, or is already a token on the place we want to cover. A similar guess is made for a second token to be removed using rule~(\ref{rule:joinPickInput2}) for $T_\mathit{join}$. Secondly, all rules that depend on the existence of transitions in one of the three transducers are also nondeterministic. This is because the transducers are nondeterministic themselves, and the $\DCPS$ just guesses the next transition leading to an accepting run at the end.

These exceptions do not cause the $\DCPS$ to reach the state $\mathit{halt}$ if the performed guess is wrong. Our determinism in the rest of the construction causes the system to have no applicable rule at some point, if we ever guess a wrong token or transducer transition: The simulation of transducer moves will either get stuck with no applicable transducer transition, or with no final state at the end. The check for coverability of $w_\mathit{final}$ can also not advance, if we resumed a thread with the wrong stack content. All in all, the $\DCPS$ $\cN$ can only reach $g_\mathit{halt}$ if its run faithfully simulated a run of the $\sTDPN$ $\cN$.